\documentclass[a4paper]{article}

\usepackage[british]{babel}
\usepackage[T1]{fontenc}
\usepackage[ansinew]{inputenc}
\usepackage{array}
\usepackage{a4}
\usepackage{a4wide}
\usepackage{Common/theorems}
\usepackage{Common/prooftree}

\usepackage[a4paper,colorlinks,pagebackref,linkcolor=blue,citecolor=red,urlcolor=cyan]{hyperref}
\usepackage{lmodern}
\usepackage{amsmath,amssymb,stmaryrd}

\long\def\ignore#1{\relax}

\newcommand\struto[1][15pt]{{\raise #1 \hbox{\strut}}}%
\newcommand\strutb[1][15pt]{{\raise-#1 \hbox{\strut}}}%

\newcommand\upline{\hline\struto[12pt]}
\newcommand\midline[1][4]{\\[+ #1pt]\hline\struto[#1pt]}

\newcommand\downline[1][12]{\\[+ #1pt]\hline}

\newcommand\pushright[1]{{              
    \parfillskip=0pt            
    \widowpenalty=1000000000         
    \displaywidowpenalty=10000  
    \finalhyphendemerits=0      
    \leavevmode                 
    \unskip                     
    \nobreak                    
    \hfil                       
    \penalty50                  
    \hskip.2em                  
    \null                       
    \hfill                      
    {#1}                        
    \par}}                      

\newcommand\olditem{}
\newcommand\olditemize{}
\newcommand\oldenditemize{}
\newcommand\oldenumerate{}
\newcommand\oldendenumerate{}
\let\olditem\item
\let\olditemize\itemize
\let\oldenditemize\enditemize
\let\oldenumerate\enumerate
\let\oldendenumerate\endenumerate

\newcommand\myitem{}
\makeatletter
\def\myitem{\@ifnextchar[\@myitemwith\@myitemwithout}
\long\def\@myitemwith[#1]{\olditem[{#1}]\unskip}
\long\def\@myitemwithout{\olditem\unskip}
\makeatother

\renewenvironment{itemize}[1][0]{%
  \def\item{\removelastskip\myitem}%
  \removelastskip\olditemize\removelastskip}
{\removelastskip\oldenditemize\removelastskip%
  \def\item\olditem%
}

\renewenvironment{enumerate}[1][0]{%
  \def\item{\removelastskip\myitem}%
  \removelastskip\oldenumerate\removelastskip}
{\removelastskip\oldendenumerate\removelastskip%
  \def\item\olditem%
}

\newcommand\mybox[1]{\fbox{\vbox{#1}}}
\renewcommand\[[1][3]{\par\removelastskip\vskip#1pt\vbox\bgroup\hrule height0pt\vfil\hbox to\hsize\bgroup\hfil\(}
\renewcommand\][1][3]{\)\hfil\egroup\vfil\hrule height0pt\egroup\vskip#1pt\nointerlineskip\noindent}

\newbox\columnsbox
\newbox\tmpbox
\newdimen\columnsheight
\newdimen\columnwidth
\newdimen\remainingwidth
\newdimen\textwidthsave
\def\mycolumnsheight{}

\newcommand\columns[1]{%
  \def\mycolumnsheight{}%
  \setlength\remainingwidth\textwidth%
  \setbox\columnsbox=\vbox\bgroup\vskip0pt\vfil\hbox to\textwidth\bgroup#1\egroup\vfil\egroup%
  \columnsheight=\ht\columnsbox%
  \def\mycolumnsheight{to\columnsheight}%
  \hrule height 0pt\vtop{\hbox to\wd\columnsbox\bgroup#1\egroup}%
}

\makeatletter

\def\commonpart{%
  \setlength\columnwidth{\wd\tmpbox}%
  \vtop{\vskip0pt\hbox to\columnwidth{{\box\tmpbox}}}%
  \advance\remainingwidth-\columnwidth%
  \setlength\textwidth\textwidthsave%
  \hsize\textwidthsave%
}
\def\column{\unskip\setlength\textwidthsave\textwidth\@ifnextchar[\@columnwith\@columnwithout}
\long\def\@columnwith[#1]#2{%
  \def\newhsize{#1\dimexpr\textwidth\relax}%
  \hsize\newhsize%
  \ifdim\hsize<0.1pt\hsize\remainingwidth\fi%
  \setlength\textwidth\hsize%
  \setbox\tmpbox=\hbox to\hsize\bgroup\hfil\vtop\mycolumnsheight{\vskip0pt#2\vskip0pt}\hfil\egroup%
  \commonpart%
}
\long\def\@columnwithout#1{%
  \hsize\remainingwidth%
  \setlength\textwidth\hsize%
  \setbox\tmpbox=\hbox\bgroup\vtop\mycolumnsheight{\vskip0pt#1\vskip0pt}\egroup%
  \commonpart%
}
\makeatother

\newenvironment{centre}{\par\vbox\bgroup\null\hfill}{\hfill\null\egroup\par\vskip2pt}

\newcommand{\eqdef}{:=\ }
\newcommand{\recdef}{::=\ }
\newcommand{\sqin}{\textsf{\footnotesize{E}}}

\newcommand{\Gam}{\Gamma}

\newcommand{\Del}{\Delta}

\newcommand\LKF{\textsf{LKF}}

\newcommand\mathFomega{F_\omega}
\newcommand\Fomega{\ifmmode\mathFomega\else$\mathFomega$\fi}
\newcommand\mathFomegaC{F_\omega^{\mathcal C}}
\newcommand\FomegaC{\ifmmode\mathFomegaC\else$\mathFomegaC$\fi}
\newcommand\mathDNE{\mathrm{DNE}}
\newcommand\DNE{\ifmmode\mathDNE\else$\mathDNE$\fi}

\renewcommand{\iff}{if and only if}

\newcommand{\ie}{i.e.~}
\newcommand{\eg}{e.g.~}

\newcommand{\wrt}{w.r.t.~}
\newcommand{\resp}{resp.~}

\ifdefined\url
\else
\def\url[#1]#2{\texttt{#2}}
\fi

\let\oldurl\url

\makeatletter
\ifdefined\href
\def\myurl{\@ifnextchar[\@myurlwith\@myurlwithout}
\long\def\@myurlwith[#1]#2{\mbox{\href{#2}{#1}}}
\long\def\@myurlwithout#1{\mbox{\href{#1}{#1}}}
\else
\def\myurl{\@ifnextchar[\@myurlwith\@myurlwithout}
\long\def\@myurlwith[#1]#2{\mbox{\oldurl[#1]{#2}}}
\long\def\@myurlwithout#1{\mbox{\oldurl{#1}}}
\fi
\makeatother

\def\url{\myurl}

\newcommand\monthdisplay[1]{\unskip}

\newcommand\Id[1][]{\textsf{Id}_{#1}}

\newcommand{\multiset}[1]{\{ \!\! \{ #1\} \!\! \} }

\ignore{

}

\newcommand{\sep}{\mbox{$\;|\;$}}    

\newcommand\FV[2][{}]{\textsf{FV}_{#1}(#2)}

\newcommand{\subst}[3]{ \left\{{}^{#3}\hspace{-6pt}\diagup\hspace{-2pt}_{#2} \right\}\hspace{-1pt} #1 }

\newcommand{\seqg}[3]{\mbox{$\ {#1}_{#2}^{#3}\ $}}

\newcommand{\seqf}[2][]{\seqg{\vdash}{#1}{#2}}

\newcommand{\seqFOL}[2][]{\seqg{\vdash_{\textsf{FOL}}}{#1}{#2}} 
\newcommand{\seq}[1][]{\seqf[#1]{}}

\newcommand{\seqFol}[1][]{\seqFOL[#1]{}}  

\newcommand\Seq[3][]{#2\seq[{#1}] #3}

\newcommand\SeqF[3][]{#2\seqFol[{#1}] #3} 

\newcommand\DerOSPos[3][\mathcal P]{\DerPos{#2}{#3}{}{#1}}
\newcommand\DerOSNeg[3][\mathcal P]{\DerNeg{#2}{#3}{}{#1}}

\newcommand\DerOSFl[3][]{\Seq[#1]{#2}{#3}}  
\newcommand\DerOSFOL[3][]{\SeqF[#1]{#2}{#3}} 

\newcommand\DerPos[5][]{{#2}\seqf[#1]{#5}{[#3]}}
\newcommand\DerNeg[5][]{{#2}\seqf[#1]{#5}{#3}}

\makeatletter
\newcommand\@DerOSExwith[1]{}
\newcommand\@DerOSExwithout[1]{}
\newcommand\DerOSEx[4][]{%
  \renewcommand\@DerOSExwith[1][\mathcal P]{#2\seqf[#1]{##1}[#3]#4}%
  \@DerOSExwith
  }

\makeatother

\newcommand\DerOSLKEx[3]{\DerOSEx[\mbox{\scriptsize\LKThEx}]{#1}{#2}{#3}}

\newcommand\DerPosLK [3][\mathcal P] {\DerPos[\mbox{\scriptsize\LKThp}]{#2}{#3}{}{#1}}
\newcommand\DerNegLK [3][\mathcal P] {\DerNeg[\mbox{\scriptsize\LKThp}]{#2}{#3}{}{#1}}

\newcommand\cut{\textsf{cut}}

\newcommand\daggerL{\raise3pt\hbox{\rotatebox{-40}{$\dagger$}}}
\newcommand\daggerR{\raise0pt\hbox{\rotatebox{40}{$\dagger$}}}

\newcommand\andP{{\wedge^+}}
\newcommand\andN{{\wedge^-}}
\newcommand\orP{{\vee^+}}
\newcommand\orN{{\vee^-}}

\newcommand\UP[1][\mathcal P]{\textsf{U}_{#1}}

\newcommand\andF{\wedge}
\newcommand\orF{\vee}

\newcommand\EX[2]{\exists #1 #2}
\newcommand\FA[2]{\forall #1 #2}

\newcommand{\non}[1]{{#1}^{\perp}}

\newcommand\Theory[2][\mathcal T]{#2\models_{#1}}

\newcommand\LKThEx{\textsf{LK}${}^+({\mathcal T}$)}

\newcommand\LKTh[1][\mathcal T]{\textsf{LK}($#1$)}
\newcommand\LKThp[1][\mathcal T]{\textsf{LK}$^p$($#1$)}

\newcommand\size[1]{\sharp(#1)}

\newcommand\DPLLTh{\textsf{DPLL}($\mathcal T$)}

\newcommand\closubst[1]{#1^\downarrow}

\newcommand\polar[2][\mathcal P]{#1;#2}

\newcommand\atmCtxt[2][\mathcal P]{\textsf{lit}_{#1} (#2)}
\newcommand\atmCtxtP[2]{\textsf{lit}_{#1} {(#2)} }

\newcommand\weak[1][l]{\textsf{W}_{#1}}
\newcommand\contr[1][l]{\textsf{C}_{#1}}
\newcommand\Init[1][]{\textsf{Init}_{#1}}
\newcommand\Iden[1][]{\textsf{Id}_{#1}}
\newcommand\Inst[1][]{\textsf{Inst}_{#1}}

\newcommand\Release{\textsf{Release}}
\newcommand\Select[1][]{\textsf{Select}^{#1}}
\newcommand\Store[1][]{\textsf{Store}^{#1}}
\newcommand\Pol{\textsf{Pol}}

\renewcommand\sqin{\raise1pt\hbox{\large $\,\epsilon\,$}}

\newcommand\TLemma[1][\mathcal T]{\Psi_{#1}}

\newcommand\Index[2][]{\emph{#2}}

\begin{document}
\title{Sequent calculi with procedure calls}

\author{Mahfuza Farooque${}^1$, Stéphane Graham-Lengrand${}^{1,2}$\\[15pt]
  ${}^1$ CNRS\\
  ${}^2$ Ecole Polytechnique\\
  Project PSI: ``Proof Search control in Interaction with domain-specific methods''\\
  ANR-09-JCJC-0006
}

\date{\today}

\maketitle \thispagestyle{empty} 

\abstract{ 

  In this paper, we introduce two focussed sequent calculi, \LKThp\ and
  \LKThEx, that are based on Miller-Liang's
  \LKF\ system~\cite{liang09tcs} for polarised classical logic.  The
  novelty is that those sequent calculi integrate the possibility to
  call a decision procedure for some background theory $\mathcal T$,
  and the possibility to polarise literals "on the fly" during
  proof-search.

  These features are used in other
  works~\cite{farooqueTR12,farooque13} to simulate the
  \DPLLTh\ procedure~\cite{Nieuwenhuis06} as proof-search in the
  extension of \LKThp\ with a cut-rule.

  In this report we therefore prove cut-elimination in \LKThp.

  Contrary to what happens in the empty theory, the polarity of
  literals affects the provability of formulae in presence of a theory
  $\mathcal T$. On the other hand, changing the polarities of
  connectives does not change the provability of formulae, only the
  shape of proofs.

  In order to prove this, we introduce a second sequent calculus,
  \LKThEx\, that extends \LKThp\ with a relaxed focussing discipline,
  but we then show an encoding of \LKThEx\ back into the more
  restrictive system \LKTh.

  We then prove completeness of \LKThp\ (and therefore of \LKThEx)
  with respect to first-order reasoning modulo the ground
  propositional lemmas of the background theory $\mathcal T$.
}

\tableofcontents

\newpage

\section{\LKThp: Definitions}
\label{sec:LKTh}
The sequent calculus \LKThp\ manipulates the formulae of first-order logic, with the specificity that connectives are of one of two kinds: positive ones and negative ones, and each \emph{boolean} connective comes in two versions, one of each kind. This section develops the preliminaries and the definition of the \LKThp\ system.

\begin{definition}[Terms and literals]
Consider an infinite set of elements called \Index[variable]{variables}. 

The set of \Index[term]{terms} over a first-order (function) signature $F_{\Sigma}$  is defined by:

\[ t,t_1,t_2,\ldots  \eqdef x \mid f(t_1,\ldots,t_n) \]
with $f/n$ ($f$ of arity $n$) ranging over $F_{\Sigma}$ and $x$ ranging over variables. 

Let $P_{\Sigma}$ be a first-order predicate signature equipped with an
involutive and arity-preserving function called \Index{negation}. The
negation of a predicate symbol $P$ is denoted $\non P$.

Let $\mathcal L^\top$ be the set $\{ P(t_1,\ldots,t_n) \mid P/n \in
P_{\Sigma}, t_1,\ldots t_n\mbox{ terms}\}$, to which we extend the
involutive function of negation with:
\[ \non{(P(t_1,\ldots,t_n))}  \eqdef \non{P}(t_1,\ldots,t_n) \]

The substitution, in a term $t'$, of a term $t$ for a variable $x$,
denoted $\subst {t'} x t$, is defined as usual, and straightforwardly
extended to elements of $\mathcal L^\top$.

In the rest of this chapter, we consider a subset $\mathcal
{L}\subseteq\mathcal L^\top$, of elements called
\Index[literal]{literals} and denoted $l,l_1,l_2\ldots$, that is
closed under negation and under substitution.\footnote{Very often we will take $\mathcal L=\mathcal L^\top$, but it is not a necessity.}
\label{def:subst}

For a set $\mathcal A$ of literals, we write $\subst {\mathcal A} x t$
for the set $\{\subst l x t \mid l\in\mathcal A\}$.  The closure of $\mathcal
A$ under all possible substitutions is denoted $\closubst {\mathcal
  A}$.
\end{definition}

%
%
%

\begin{notation}
  We often write $\mathcal V,\mathcal V'$ for the set or multiset union of $\mathcal V$ and $\mathcal V'$.
\end{notation}

\begin{remark}Negation obviously commutes with substitution.
\end{remark}

\begin{definition}[Inconsistency predicates]\strut
  \label{def:inconsistency}

  An \Index{inconsistency predicate} is a predicate over sets of literals 
  \begin{enumerate}
  \item satisfied by the set $\{l,\non l\}$ for every literal $l$;
  \item that is upward closed (if a subset of a set satisfies the predicate, so does the set);
  \item such that if the sets $\mathcal A,l$ and $\mathcal A,\non l$ satisfy it, then so does $\mathcal A$.
  \item such that if a set $\mathcal A$ satisfies it, then so does $\subst {\mathcal A} x t$. 
  \end{enumerate}

  The smallest inconsistency predicate is called the
  \Index{syntactical inconsistency} predicate\footnote{It is the
    predicate that is true of a set $\mathcal A$ of literals iff
    $\mathcal A$ contains both $l$ and $\non l$ for some $l\in\mathcal
    L$.}.  If a set $\mathcal A$ of literals satisfies the
  syntactically inconsistency predicate, we say that $\mathcal{P}$ is
  \emph{syntactically inconsistent}, denoted $\mathcal
  P\models$. Otherwise $\mathcal A$ is \Index[syntactical
    consistency]{syntactically consistent}.


  In the rest of this chapter, we specify a ``theory'' $\mathcal T$ by
  considering another inconsistency predicate called the
  \Index{semantical inconsistency} predicate.  If a set $\mathcal A$
  of literals satisfies it, we say that $\mathcal A$ is
  \Index[semantical inconsistency]{semantically inconsistent}, denoted
  by $\mathcal A\models_{\mathcal T}$. Otherwise $\mathcal A$ is
  \Index[semantically consistent]{semantically consistent}.
\end{definition}

\begin{remark}
\label{rem:contract}\strut
\begin{itemize}
\item
In the conditions above, (1) corresponds to \Index{basic inconsistency}, (2) corresponds to \Index{weakening}, (3) corresponds to \Index{cut-admissibility} and (4) corresponds to \Index{stability under instantiation}. \emph{Contraction} is built-in because inconsistency predicates are predicates over sets of literals (not multisets).
\item If $\mathcal A$ is syntactically consistent, $\subst {\mathcal A} x t$ might not be syntactically consistent.
\end{itemize}
\end{remark}

\begin{definition}[Formulae]\strut\label{def:polarformulae}

The \Index{formulae} of polarised classical logic are given by the following grammar:
  \[
  \begin{array}{lrl}
    \mbox{Formulae }&A,B,\ldots\ \strut\recdef &l\\
    &\mid &A\andP B\mid A\orP B \mid \EX x A \mid \top^+\mid \bot^+\qquad\strut\\
    &\mid &A\andN B\mid A\orN B \mid \FA x A \mid \top^-  \mid \bot^-
  \end{array}
  \]
where $l$ ranges over $\mathcal L$.

The set of \Index[free variable]{free variables} of a formula $A$, denoted $\FV A$, and $\alpha$-conversion, are defined as usual so that both $\EX x A$ and $\FA x A$ bind $x$ in $A$.

The size of a formula $A$, denoted $\size A$, is its size as a tree (number of nodes).

\Index[negation]{Negation} is extended from literals to all formulae:
  \[
  \begin{array}{|ll|ll|}
    \hline
    \non{(A\andP B)}&\eqdef \non A \orN\non B&\non{(A\andN B)}&\eqdef \non A \orP\non B    \\
    \non{(A\orP B)}&\eqdef \non A \andN\non B&\non{(A\orN B)}&\eqdef \non A \andP\non B    \\
    \non{(\EX x A)}&\eqdef \FA x {\non A}& \non{(\FA x A)}&\eqdef \EX x {\non A}   \\
    \non{(\top^+)}&\eqdef \bot^-& \non{(\top^-)}&\eqdef \bot^+   \\
    \non{(\bot^+)}&\eqdef \top^-& \non{(\bot^-)}&\eqdef \top^+   \\
    \hline
  \end{array}
  \]

The \Index{substitution} in a formula $A$ of a term
$t$ for a variable $x$, denoted $\subst A x t$, is defined in the usual capture-avoiding
way.
\end{definition}

\begin{notation}
For a set (\resp multiset) $\mathcal V$ of literals / formulae,
$\non{\mathcal V}$ denotes $\{\non A\mid A\in\mathcal V\}$ (\resp
$\multiset{\non A\mid A\in\mathcal V}$).  Similarly, we write
$\subst {\mathcal V} x t$ for $\{\subst A x t \mid A\in\mathcal V\}$
(\resp $\multiset{\subst A x t \mid A\in\mathcal V}$), and $\FV{\mathcal V}$ for the set  $\bigcup_{A\in\mathcal V}\FV A$.
\end{notation}

\begin{definition}[Polarities]\strut

A \Index{polarisation set} $\mathcal P$ is a set of literals ($\mathcal P\subseteq\mathcal L$) that is syntactically consistent, and 
such that $\FV{\mathcal P}$ is finite.

Given such a set, we define \Index{$\mathcal P$-positive formulae} and \Index{$\mathcal P$-negative formulae} as the formulae generated by the following grammars:
\[
\begin{array}{lll}
  \mathcal P\mbox{-positive formulae }&P,\ldots&\recdef p\mid A\andP B\mid A\orP B\mid \EX x A \mid \top^+ \mid \bot^+\\
  \mathcal P\mbox{-negative formulae }&N,\ldots&\recdef \non p\mid A\andN B\mid A\orN B \mid \FA x A \mid \top^-  \mid \bot^-\\
\end{array}
\]
where $p$ ranges over $\mathcal P$.

In the rest of the chapter, $p$, $p'$,\ldots will denote a literal that is $\mathcal P$-positive, when the polarisation set $\mathcal P$ is clear from context.

Let $\UP$ be the set of all \Index{$\mathcal P$-unpolarised literals}, \ie literals that are neither $\mathcal P$-positive nor $\mathcal P$-negative.
\end{definition}

\begin{remark}
  Notice that the negation of a $\mathcal P$-positive formula is $\mathcal P$-negative and vice versa. On the contrary, nothing can be said of the polarity of the result of substitution on a literal \wrt the polarity of the literal: \eg $l$ could be in $\mathcal P$-positive, while $\subst l x t$ could be $\mathcal P$-negative or $\mathcal P$-unpolarised.
\end{remark}

\begin{definition}[\LKThp]

  The sequent calculus \LKThp\ manipulates two kinds of sequents:
  \begin{centre}
    \begin{tabular}{ l l }
      \mbox{Focused sequents }&
      $\DerPos \Gamma A {}{\mathcal P}$ \\
      \mbox{Unfocused sequents }&
      $\DerNeg \Gamma {\Del}{}{\mathcal P}$
    \end{tabular}
  \end{centre}
  where $\mathcal P$ is a polarisation set, $\Gam$ is a (finite) multiset of literals and $\mathcal P$-negative formulae, $\Delta$ is a (finite) multiset of formulae, and $A$ is said to be in the \Index{focus} of the (focused) sequent.

By $\atmCtxtP{\mathcal P}\Gam$ we denote the sub-multiset of $\Gam$ consisting of its $\mathcal P$-positive literals (\ie $\mathcal P\cap\Gamma$ as a set).

The rules of \LKThp, given in Figure~\ref{fig:LKThp}, are of three kinds: \Index{synchronous rules}, \Index[asynchronous rules]{asynchronous rules}, and \Index{structural rules}. These correspond to three alternating phases in the proof-search process that is described by the rules.
\end{definition}

\begin{figure}[!h]
  \[
  \begin{array}{|c|}
    \upline
    \textsf{Synchronous rules}
    \hfill\strut\\[3pt]
    \infer[{[(\andP)]}]{\DerPos{\Gamma}{A\andP B}{}{\mathcal{P}}}
    {\DerPos{\Gamma}{A}
      {}{\mathcal P} \qquad \DerPos{\Gamma}{B}{}{\mathcal{P}}}
    \qquad
    \infer[{[(\orP)]}]{\DerPos{\Gamma}{A_1\orP A_2}{}{\mathcal{P}}}
    {\DerPos{\Gamma}{A_i}{}{\mathcal{P}}}
       \qquad 
    \infer[{[(\EX)]}]{\DerPos{\Gamma}{\EX x A}{}{\mathcal{P}}}
    {\DerPos{\Gamma}{\subst A x t}{}{\mathcal{P}}}   
    \\[15pt]
      \infer[{[{(\top^+)}]}]{\DerPos{\Gamma}{\top^+} {} {\mathcal{P}}} {\strut}
       \qquad 
\infer[{[({\Init[1]})]l \mbox{ is $\mathcal P$-positive}}]
	{\DerPos{\Gamma  }{l}{} {\mathcal{P}}} {\atmCtxtP{\mathcal P}\Gam,\non l\models_{\mathcal T}}    \qquad
    \infer[{[(\Release)]N \mbox{ is not $\mathcal P$-positive}}]{\DerPos {\Gam} {N} {} {\mathcal{P}}}
    {\DerNeg {\Gam} {N} {} {\mathcal{P}}}\\
    \midline
    \textsf{Asynchronous rules}
    \hfill\strut\\[3pt]
    \infer[{[(\andN)]}]{\DerNeg{\Gamma}{A\andN B,\Delta} {} {\mathcal{P}}}
    {\DerNeg{\Gamma}{A,\Delta} {} {\mathcal{P}} 
      \qquad \DerNeg{\Gamma}{B,\Delta} {} {\mathcal{P}}}
    \qquad
    \infer[{[(\orN)]}]{\DerNeg {\Gamma} {A_1\orN A_2,\Delta} {} {\mathcal{P}}}
    {\DerNeg {\Gamma} {A_1,A_2,\Delta} {} {\mathcal{P}}}
    
   \qquad
       \infer[{[(\FA)] x\notin\FV{\Gam,\Delta,\mathcal P} }]
       {\DerNeg{\Gamma}{(\FA x A),\Delta}{} {\mathcal{P}}}
    {\DerNeg {\Gamma} {A,\Delta}{} {\mathcal{P}}}
\\\\
    \infer[{[{(\bot^-)}]}]{\DerNeg{\Gamma} {\Del,\bot^-} {} {\mathcal{P}}}
    {\DerNeg{\Gamma} {\Del} {} {\mathcal{P}}}
     \qquad
    \infer[{[{(\top^-)}]}]{\DerNeg{\Gamma} {\Del,\top^-} {} {\mathcal{P}}}{\strut}
     \qquad
    \infer[{[({\Store})]%
      \begin{array}l%
        A\mbox{ is a literal}\\
        \mbox{or is $\mathcal P$-positive}\\
      \end{array}}]{\DerNeg \Gam {A,\Del} {} {\mathcal{P}}} 
    {\DerNeg {\Gam,\non A} {\Del} {} {\polar{\non A}}}  
    \\\midline
      \textsf{Structural rules}
    \hfill\strut\\[3pt]
    \infer[{[(\Select)]\mbox{$P$ is not $\mathcal P$-negative}}]
    {\DerNeg {\Gam,\non P} {}{} {\mathcal{P}}} 
    {\DerPos {\Gam,\non P} {P} {} {\mathcal{P}}}
    \qquad
    \infer[{[({\Init[2]})]}]{\DerNeg {\Gam} {}{} {\mathcal{P}}}{\atmCtxtP{\mathcal P}\Gam\models_{\mathcal T}}
    \downline
  \end{array}
  \]
  \begin{tabular}{lll}%
    where& $\polar A \eqdef \mathcal P,A$& if $A\in\UP$\\
    &$\polar A \eqdef \mathcal P$& if not
  \end{tabular}%

  \caption{System \LKThp}
  \label{fig:LKThp}
\end{figure}

The gradual proof-tree construction defined by the inference rules
of \LKThp\ is a goal-directed mechanism whose intuition can be given
as follows:

Asynchronous rules are invertible: $(\andN)$ and $(\orN)$ are applied
eagerly when trying to construct the proof-tree of a given sequent;
$(\Store)$ is applied when hitting a literal or a positive
formula on the right-hand side of a sequent, storing its negation on
the left.

When the right-hand side of a sequent becomes empty, a sanity check
can be made with $({\Init[2]})$ to check the semantical consistency of
the stored (positive) literals (\wrt the theory), otherwise a choice must
be made to place a formula in focus which is not $\mathcal P$-negative, 
before applying synchronous rules like $(\andP)$ and $(\orP)$. Each such rule
decomposes the formula in focus, keeping the revealed sub-formulae in
the focus of the corresponding premises, until a positive literal or a
non-positive formula is obtained: the former case must be closed
immediately with $({\Init[1]})$ calling the decision procedure, and
the latter case uses the $(\Release)$ rule to drop the focus and start
applying asynchronous rules again. The synchronous and the structural
rules are in general not invertible, and each application of those
yields a potential backtrack point in the proof-search.

\begin{remark} The polarisation of literals (if not already polarised) happens in the $(\Store)$ rule, where the construction $\polar A$ plays a crucial role. It will be useful to notice the commutation $\polar[\polar A] B=\polar[\polar B] A$ unless $A=\non B\in\UP$.
\end{remark}

\section{Admissibility of basic rules}
\label{sec:adm}
In this section, we show the admissibility and invertibility of some rules, in order to prove the meta-theory of \LKThp.

\begin{lemma}[Weakening and contraction]
  The following rules are height-preserving admissible in \LKThp:
  \[\begin{array}{c}
    \infer[{[({\weak})]}]{\DerNeg {\Gam,A}{\Delta}{}{\mathcal P}}
    {\DerNeg \Gam \Del {}{\mathcal P}}
    \qquad
    \infer[{[({\weak[f]})]}]{\DerPos {\Gam,A}{B}{}{\mathcal P}}
    {\DerPos \Gam B {}{\mathcal P}}
    \\\\
    \infer[{[({\contr})]}]{\DerNeg {\Gam,A}{\Delta}{}{\mathcal P}}
    {\DerNeg {\Gam,A,A} \Del {}{\mathcal P}}
    \qquad
    \infer[{[({\contr[f]})]}]{\DerPos {\Gam,A}{B}{}{\mathcal P}}
    {\DerPos {\Gam,A,A} B {}{\mathcal P}}
    \qquad
    \infer[{[({\contr[r]})]}]{\DerNeg {\Gam}{\Delta,A}{}{\mathcal P}}
    {\DerNeg {\Gam} {\Del,A,A} {}{\mathcal P}}
  \end{array}
  \]
\end{lemma}
\begin{proof}
  By induction on the derivation of the premiss.
\end{proof}

\begin{lemma}[Identities]
\label{LadamIden}
  The identity rules are admissible in \LKThp:
  \[
  \infer[{[({\Iden[1]})] l \mbox{ is $\mathcal P$-positive}}]
    {\DerPos{\Gamma,l}{l}{}{\mathcal P}}
    {\strut}
  \qquad
  \infer[{[({\Iden[2]})] }]
    {\DerNeg{\Gamma,l,\non l}{}{}{\mathcal P}}
    {\strut}
  \]
\end{lemma}

\begin{proof}
It is trivial to prove \Index{$\Iden[1]$}.

If $l$ or $\non l$ is $\mathcal P$-positive, the \Index{$\Iden[2]$} rule can be obtained by a derivation of the following form:
\[
\infer{\DerNeg{\Gamma,l,\non l}{}{}{\mathcal P}}
{ \infer[{[({\Iden[1]})] }]{\DerPos{\Gamma,l,\non l} {l} {} {\mathcal P}}
{}
}
\]
where $l$ is assumed to be the $\mathcal P$-positive literal.

If $l\in\UP$, we polarise it positively with
\[
\infer{\DerNeg{\Gamma,l,\non l}{}{}{\mathcal P}}
      {
        \infer[{[({\Release})] }]{\DerPos{\Gamma,l,\non l} {\non l} {} {\mathcal P}}
              {
                \infer[{[({\Store})] }]{\DerNeg{\Gamma,l,\non l} {\non l} {} {\mathcal P}}
                      {
                        {\DerNeg{\Gamma,l,\non l,l} {} {} {\mathcal P, l}}
                      }
              }
      }
\]
\end{proof}

\section{Invertibility of the asynchronous phase} 
\label{sec:Invert}
We have mentioned that the asynchronous rules are invertible; now
in this section, we prove it.

\begin{lemma}[Invertibility of asynchronous rules]\strut
  \label{lem:invert}
  All asynchronous rules are invertible in \LKTh.
\end{lemma}
\begin{proof}
  By induction on the derivation proving the conclusion of the asynchronous rule considered.
  \begin{itemize}
  \item Inversion of $A \andN B$: by case analysis on the last rule actually used
    \begin{itemize}
    \item ($\andN$)
    $$\infer{\DerNeg{\Gamma}{A\andN B,C\andN D,\Delta'}{} {\mathcal P}} 
    {\DerNeg{\Gamma}{A\andN B,C,\Delta'}{} {\mathcal P} \quad \DerNeg{\Gamma}{A\andN B,D,\Delta'}{}{\mathcal P} }$$

      By induction hypothesis we get
      $$\infer{\DerNeg{\Gamma}{A,C\andN D,\Delta'}{} {\mathcal P}} {\DerNeg{\Gamma}{A,C,\Delta'}{} {\mathcal P} \qquad \DerNeg{\Gamma}{A,D,\Delta'}{} {\mathcal P}}$$
      and
      $$\infer{\DerNeg{\Gamma}{B,C\andN D,\Delta'}{} {\mathcal P}} {\DerNeg{\Gamma}{B,C,\Delta'}{} {\mathcal P} \quad \DerNeg{\Gamma}{B,D,\Delta'}{} {\mathcal P}}$$

    \item ($\orN$)
      $$\infer{\DerNeg{\Gamma}{A\andN B,C\orN D,\Delta'}{} {\mathcal P}} {\DerNeg{\Gamma}{A\andN B,C, D,\Delta'}{}{\mathcal P}}$$

      By induction hypothesis we get
      $$\infer{\DerNeg{\Gamma}{A,C\orN D,\Delta'}{} {\mathcal P}} {\DerNeg{\Gamma}{A,C, D,\Delta'}{} {\mathcal P}}$$ and
      $$\infer{\DerNeg{\Gamma}{B,C\orN D,\Delta'}{} {\mathcal P}} {\DerNeg{\Gamma}{B,C, D,\Delta'}{} {\mathcal P}}$$

    \item ($\FA$)
      $$\infer[x\notin\FV{\Gam,\Delta',A\andN B}]
      {\DerNeg{\Gamma}{A\andN B,(\FA x C),\Delta'}{} {\mathcal P}} 
      {\DerNeg{\Gamma}{A\andN B,C,\Delta'}{} {\mathcal P}}$$

      By induction hypothesis we get
      $$\infer[x\notin\FV{\Gam,\Delta',A}]
      {\DerNeg{\Gamma}{A,(\FA x C),\Delta'}{} {\mathcal P}} 
      {\DerNeg{\Gamma}{A,C,\Delta'}{} {\mathcal P}}$$ 
      and
      $$\infer[x\notin\FV{\Gam,\Delta',B}]{\DerNeg{\Gamma}{B,(\FA x C),\Delta'}{} {\mathcal P}} {\DerNeg{\Gamma}{B,C,\Delta'}{} {\mathcal P}}$$

    \item ($\Store$)
      $$\infer[C\mbox{ literal or $\mathcal P$-positive formula}]
      {\DerNeg{\Gamma}{A\andN B,C,\Delta'}{} {\mathcal P}} 
      {\DerNeg{\Gamma,\non C}{A\andN B,\Delta'}{} {\polar {\non C}}}$$  

      By induction hypothesis we get
      $$\infer[C\mbox{ literal or $\mathcal P$-positive formula}]
      {\DerNeg{\Gamma}{A,C,\Delta'}{} {\mathcal P}} 
      {\DerNeg{\Gamma,\non C}{A,\Delta'}{} {\polar {\non C}}}$$   
      and
      $$\infer[C\mbox{ literal or $\mathcal P$-positive formula}]{\DerNeg{\Gamma}{ B,C,\Delta'}{} {\mathcal P}} {\DerNeg{\Gamma,\non C}{B,\Delta'}{} {\polar {\non C}}} $$ 

    \item ($\bot^-$)
      $$\infer{\DerNeg{\Gamma}{A\andN B,\bot^-,\Delta'}{} {\mathcal P}} 
      {\DerNeg{\Gamma}{A\andN B,\Delta'}{} {\mathcal P}}$$  

      By induction hypothesis we get\\[5pt]
      $$\infer{\DerNeg{\Gamma}{A,\bot^-,\Delta'}{} {\mathcal P}} 
      {\DerNeg{\Gamma}{A,\Delta'}{} {\mathcal P}}$$
      and
      $$\infer{\DerNeg{\Gamma}{ B,\bot^-,\Delta'}{} {\mathcal P}} 
      {\DerNeg{\Gamma}{B,\Delta'}{} {\mathcal P}} $$ 

\item ($\top^-$)
      $$\infer{\DerNeg{\Gamma}{A\andN B,\top^-,\Delta'}{} {\mathcal P}} {}$$  

      We get\\[5pt]
      $\infer{\DerNeg{\Gamma}{A,\top^-,\Delta'}{} {\mathcal P}} {}$   and
      $\infer{\DerNeg{\Gamma}{ B,\top^-,\Delta'}{} {\mathcal P}} {} $ 
   \end{itemize}

  \item Inversion of $A \orN B$: by case analysis on the last rule    
    \begin{itemize}
    \item ($\andN$)
    $$\infer{\DerNeg{\Gamma}{A\orN B,C\andN D,\Delta'}{} {\mathcal P}} 	 	{\DerNeg{\Gamma}{A\orN B,C,\Delta'}{} {\mathcal P} \quad 
    \DerNeg{\Gamma}{A\orN B,D,\Delta'}{} {\mathcal P}}$$   

      By induction hypothesis we get
      $$ \infer{\DerNeg{\Gamma}{A,B,C\andN D,\Delta'}{} {\mathcal P}} 
      {\DerNeg{\Gamma}{A,B,C,\Delta'}{} {\mathcal P} \quad
      \DerNeg{\Gamma}{A,B,D,\Delta'}{} {\mathcal P}}$$ 

    \item ($\orN$)
      $$\infer{\DerNeg{\Gamma}{A\orN B,C\orN D,\Delta'}{} {\mathcal P}} {\DerNeg{\Gamma}{A\orN B,C, D,\Delta'}{} {\mathcal P}}$$  

      By induction hypothesis we get
   $$\infer{\DerNeg{\Gamma}{A,B,C\orN D,\Delta'}{} {\mathcal P}} 
   {\DerNeg{\Gamma}{A,B,C, D,\Delta'}{} {\mathcal P}}$$ 

    \item ($\FA$)
      $$\infer[x\notin\FV{\Gam,\Delta'}]
      {\DerNeg{\Gamma}{A\orN B,(\FA x C),\Delta'}{} {\mathcal P}}
      {\DerNeg{\Gamma}{A\orN B,C,\Delta'}{} {\mathcal P}}$$   

      By induction hypothesis we get 
       $$\infer[x\notin\FV{\Gam,\Delta'}]
       {\DerNeg{\Gamma}{A,B,(\FA x C),\Delta'}{} {\mathcal P}} 
       {\DerNeg{\Gamma}{A,B,C,\Delta'}{} {\mathcal P}}$$  

    \item ($\Store$)
      $$\infer[C\mbox{ literal or $\mathcal P$-postive formula}]
      {\DerNeg{\Gamma}{A\orN B,C,\Delta'}{} {\mathcal P}} 
      {\DerNeg{\Gamma,\non C}{A\orN B,\Delta'}{} { \polar {\non C}}}$$  

      By induction hypothesis we get
 $$\infer[C\mbox{ literal or $\mathcal P$-positive formula}]
 {\DerNeg{\Gamma}{A,B,C,\Delta'}{}{\mathcal P}} 
 {\DerNeg{\Gamma,\non C}{A,B,\Delta'}{}{\polar {\non C}}}$$   

    \item ($\bot^-$)
      $$\infer{\DerNeg{\Gamma}{A\orN B,\bot^-,\Delta'}{} {\mathcal P}} 
      {\DerNeg{\Gamma}{A\orN B,\Delta'}{} {\mathcal P}}$$  

      By induction hypothesis we get
 	$$\infer{\DerNeg{\Gamma}{A,B,\bot^-,\Delta'}{} {\mathcal P}} 
	{\DerNeg{\Gamma}{A,B,\Delta'}{} {\mathcal P}}$$   

    \item ($\top^-$)
      $$\infer{\DerNeg{\Gamma}{A\orN B,\top^-,\Delta'}{} {\mathcal P}} {}$$  

      We get
 	$$\infer{\DerNeg{\Gamma}{A,B,\top^-,\Delta'}{} {\mathcal P}} {}$$   
 
    \end{itemize}

  \item Inversion of $\FA x A$: by case analysis on the last rule 

    \begin{itemize}
    \item ($\andN$) 
    $$\infer{\DerNeg{\Gamma}{(\FA x A),C\andN D,\Delta'}{} {\mathcal P}} 
    {\DerNeg{\Gamma}{(\FA x A),C,\Delta'}{} {\mathcal P} \quad 
    \DerNeg{\Gamma}{(\FA x A),D,\Delta'}{} {\mathcal P}}
    $$ 

      By induction hypothesis  we get 
      $$\infer[x\notin\FV{\Gam,\Delta'}]
      {\DerNeg{\Gamma}{A,C\andN D,\Delta'}{} {\mathcal P}} 
      {{\DerNeg{\Gamma}{A,C,\Delta'}{} {\mathcal P}}$  
      \quad ${\DerNeg{\Gamma}{A,D,\Delta'}{} {\mathcal P}}}$$ 
      
    \item ($\orN$)
      $$\infer{\DerNeg{\Gamma}{(\FA x A),C\orN D,\Delta'}{} {\mathcal P}} {\DerNeg{\Gamma}{(\FA x A),C, D,\Delta'}{} {\mathcal P}}$$  

      By induction hypothesis we get 
      $$\infer{\DerNeg{\Gamma}{A,C\orN D,\Delta'}{} {\mathcal P}} 
      {\DerNeg{\Gamma}{A,C, D,\Delta'}{} {\mathcal P}}$$ 

    \item ($\FA$) 
      $$\infer[x\notin\FV{\Gam,\Delta'}]
      {\DerNeg{\Gamma}{(\FA x A),(\FA x D),\Delta'}{} {\mathcal P}}   	
      {\DerNeg{\Gamma}{(\FA x A),D,\Delta'}{} {\mathcal P}}$$  

      By induction hypothesis we get 
      $$\infer[x\notin\FV{\Gam,\Delta'}]
      {\DerNeg{\Gamma}{A,(\FA x D),\Delta'}{} {\mathcal P}} 
      {\DerNeg{\Gamma}{A,D,\Delta'}{} {\mathcal P}} $$

    \item ($\Store$)
      $$\infer[C\mbox{ literal or $\mathcal P$-positive formula}]
      {\DerNeg{\Gamma}{(\FA x A),C,\Delta'}{} {\mathcal P}} 
      {\DerNeg{\Gamma,\non C}{(\FA x A),\Delta'}{} {\polar {\non C}}}$$  

      By induction hypothesis we get
      $$\infer[C\mbox{ literal or $\mathcal P$-positive formula}]
      {\DerNeg{\Gamma}{A,C,\Delta'}{} {\mathcal P}}
      {\DerNeg{\Gamma,\non C}{A,\Delta'}{} {\polar {\non C}}}$$
       
       \item ($\bot^-$)
      $$\infer{\DerNeg{\Gamma}{(\FA x A),\bot^-,\Delta'}{} {\mathcal P}} {\DerNeg{\Gamma}{(\FA x A),\Delta'}{} {\mathcal P}}$$  

      By induction hypothesis we get\\[5pt]
      $$\infer{\DerNeg{\Gamma}{A,\bot^-,\Delta'}{} {\mathcal P}} 
      {\DerNeg{\Gamma}{A,\Delta'}{} {\mathcal P}}$$
       
       \item ($\top^-$)
      $$\infer{\DerNeg{\Gamma}{(\FA x A),\top^-,\Delta'}{} {\mathcal P}} {}$$  

      We get
      $$\infer{\DerNeg{\Gamma}{A,\top^-,\Delta'}{} {\mathcal P}} {}$$
           
    \end{itemize}

  \item Inversion of ($\Store$): where $A$ is a literal or $\mathcal P$-positive formula.\\
  By case analysis on the last rule 
    \begin{itemize}
    \item ($\andN$)
    $$\infer{\DerNeg{\Gamma}{A,C\andN D,\Delta'}{} {\mathcal P}} 
    {\DerNeg{\Gamma}{A,C,\Delta'}{} {\mathcal P} \quad \DerNeg{\Gamma}{A,D,\Delta'}{} {\mathcal P}}$$  

      By induction hypothesis we get 
      $$\infer{\DerNeg{\Gamma,\non A}{C\andN D,\Delta'}{} { \polar {\non A} }} 	  
      {\DerNeg{\Gamma,\non A}{C,\Delta'}{} { \polar {\non A}} \quad 
      \DerNeg{\Gamma,\non A}{D,\Delta'}{} {\polar {\non A}} }$$ 

    \item ($\orN$)
      $$\infer{\DerNeg{\Gamma}{A,C\orN D,\Delta'}{} {\mathcal P}} 
      {\DerNeg{\Gamma}{A,C, D,\Delta'}{} {\mathcal P}}$$  

      By induction hypothesis 
      $$\infer{\DerNeg{\Gamma,\non A}{C\orN D,\Delta'}{} {\polar {\non A}}} 
      {\DerNeg{\Gamma,\non A}{C, D,\Delta'}{} {\polar {\non A}}}$$ 

    \item ($\FA$)
      $$\infer[x\notin\FV{\Gam,\Delta'}]
      {\DerNeg{\Gamma}{A,(\FA x D),\Delta'}{} {\mathcal P}} 
      {\DerNeg{\Gamma}{A,D,\Delta'}{} {\mathcal P}}$$  \\
      
      By induction hypothesis we get 
      $$\infer[x\notin\FV{\Gam,\Delta'}]
      {\DerNeg{\Gamma,\non A}{(\FA x D),\Delta'}{} {\polar {\non A} }} 
      {\DerNeg{\Gamma,\non A}{D,\Delta'}{} {\polar {\non A} }}$$  

    \item ($\Store$)
      $$\infer[B\mbox{ literal or $\mathcal P$-positive formula}]
      {\DerNeg{\Gamma}{A,B,\Delta'}{} {\mathcal P}} 
      {\DerNeg{\Gamma,\non B}{A,\Delta'}{} {\polar{\non B}}}$$

      By induction hypothesis we can construct:
      \[\infer{\DerNeg{\Gamma,\non A}{B,\Delta'}{} {\polar {\non A}}} 
      {\DerNeg{\Gamma,\non A,\non B}{\Delta'}{} {\polar[\polar{\non B}]{\non A}}}\]
      provided $\polar[\polar{\non B}]{\non A}=\polar[\polar{\non A}]{\non B}$, which is always the case unless $A=\non B$ and $A\in \UP$, in which case we build:
    \[
    \infer{\DerNeg{\Gamma,\non A}{B,\Delta}{}{\polar {\non A}}}
    { \infer[{[{(\Iden[2])}]}]{\DerNeg{\Gamma,\non A,\non B}{\Delta'}{}{\polar[\polar{\non A}]{\non B}}}
    		{}
    }
    \]

    \item ($\bot^-$)
      $$\infer{\DerNeg{\Gamma}{A,\bot^-,\Delta'}{} {\mathcal P}} 
      {\DerNeg{\Gamma}{A,\Delta'}{} {\mathcal P}}$$

      By induction hypothesis we get 
      $$\infer{\DerNeg{\Gamma,\non A}{\bot^-,\Delta'}{} {\polar {\non A}}} 
      {\DerNeg{\Gamma,\non A}{\Delta'}{} {\polar {\non A}}}$$ 
      
      \item ($\top^-$)
      $$\infer{\DerNeg{\Gamma}{A,\top^-,\Delta'}{} {\mathcal P}} {}$$

      We get 
      $$\infer{\DerNeg{\Gamma,\non A}{\top^-,\Delta'}{} {\polar {\non A}}} {}$$ 
    
    \end{itemize}
    
     \item Inversion of ($\bot^-$): by case analysis on the last rule
    \begin{itemize}
    \item ($\andN$)
    $$\infer{\DerNeg{\Gamma}{\bot^-,C\andN D,\Delta'}{} {\mathcal P}} 
    {\DerNeg{\Gamma}{\bot^-,C,\Delta'}{} {\mathcal P} \quad \DerNeg{\Gamma}{\bot^-,D,\Delta'}{} {\mathcal P}}$$  

      By induction hypothesis we get 
      $$\infer{\DerNeg{\Gamma}{C\andN D,\Delta'}{} {\mathcal P}} {\DerNeg{\Gamma}{C,\Delta'}{} {\mathcal P} \quad \DerNeg{\Gamma}{D,\Delta'}{} {\mathcal P}}$$ 

    \item ($\orN$)
      $$\infer{\DerNeg{\Gamma}{\bot^-,C\orN D,\Delta'}{} {\mathcal P}} 
      {\DerNeg{\Gamma}{\bot^-,C, D,\Delta'}{} {\mathcal P}}$$  

      By induction hypothesis 
      $$\infer{\DerNeg{\Gamma}{C\orN D,\Delta'}{} {\mathcal P}} 
      {\DerNeg{\Gamma}{C, D,\Delta'}{} {\mathcal P}}$$ 

    \item ($\FA$)
      $$\infer[x\notin\FV{\Gam,\Delta'}]
      {\DerNeg{\Gamma}{\bot^-,(\FA x D),\Delta'}{} {\mathcal P}} 
      {\DerNeg{\Gamma}{\bot^-,D,\Delta'}{} {\mathcal P}}$$  \\
      
      By induction hypothesis we get 
      $$\infer[x\notin\FV{\Gam,\Delta'}]
      {\DerNeg{\Gamma}{(\FA x D),\Delta'}{} {\mathcal P}} 
      {\DerNeg{\Gamma}{D,\Delta'}{} {\mathcal P}}$$  

    \item ($\Store$)
      $$\infer[B\mbox{ literal or $\mathcal P$-positive formula}]
      {\DerNeg{\Gamma}{\bot^-,B,\Delta'}{} {\mathcal P}} 
      {\DerNeg{\Gamma,\non B}{\bot^-,\Delta'}{} {\polar {\non B}}}$$

      By induction hypothesis we get 
      $$\infer[B\mbox{ literal or $\mathcal P$-positive formula}]
      {\DerNeg{\Gamma}{B,\Delta'}{} {\mathcal P}} 
      {\DerNeg{\Gamma,\non B}{\Delta'}{} {\polar {\non B}}}$$ 
    
    \item ($\bot^-$)
      $$\infer{\DerNeg{\Gamma}{\bot^-,\bot^-,\Delta'}{} {\mathcal P}} 
      {\DerNeg{\Gamma}{\bot^-,\Delta'}{} {\mathcal P}}$$

      By induction hypothesis we get 
      $$\infer{\DerNeg{\Gamma}{\bot^-,\Delta'}{} {\mathcal P}} 
      {\DerNeg{\Gamma}{\Delta'}{} {\mathcal P}}$$ 
      
      \item ($\top^-$)
      $$\infer{\DerNeg{\Gamma}{\top^-,\bot^-,\Delta'}{} {\mathcal P}} {}$$

       We get $$\infer{\DerNeg{\Gamma}{\top^-,\Delta'}{} {\mathcal P}} {}$$ 
        \end{itemize}
    
    \item Inversion of ($\top^-$): Nothing to do.
          
  \end{itemize}
\end{proof}


\section{On-the-fly polarisation}

The side-conditions of the \LKThp\ rules make it quite clear that the
polarisation of literals plays a crucial role in the shape of proofs.
The less flexible the polarisation of literals is, the more structure
is imposed on proofs. We therefore concentrated the polarisation of
literals in just one rule: $(\Store)$. In this section, we describe
more flexible ways of changing the polarity of literals without
modifying the provability of sequents.  We do this by showing the
admissibility and invertibility of some ``on-the-fly'' polarisation
rules.

\begin{lemma}[Invertibility]
The following rules are invertible in \LKThp:
\[
\infer[{[(\Pol)]\atmCtxtP{\mathcal {P},l}{\Gamma,\non \Del},\non l\models_{\mathcal T}}]
{\DerNeg{\Gamma}{\Del}{}{\mathcal P}}
{\DerNeg{\Gamma}{\Del}{}{\mathcal {P},l}}
\qquad
\infer[{[(\Pol_i)]\atmCtxtP{\mathcal {P},l}{\Gamma},\non l\models_{\mathcal T}}]
{\DerPos{\Gamma}{A}{}{\mathcal P}}
{\DerPos{\Gamma}{A}{}{\mathcal {P},l}}
\]
where $l\in \UP$.
\end{lemma}

\begin{proof} 
By simultaneous induction on the derivation of the conclusion (by case analysis on the last rule used in that derivation):
\begin{itemize}
\item ($\andN$),($\orN$),($\FA$),($\bot^-$),($\top^-$)\\
For these rules, whatever is done with the polarisation set $\mathcal P$ can be done with the polarisation set $\mathcal P,l$:

\[
\infer{\DerNeg{\Gamma}{A\andN B,\Delta}{} {\mathcal P}} 
    {\DerNeg{\Gamma}{A,\Delta}{} {\mathcal P} 
    \quad \DerNeg{\Gamma}{B,\Delta}{}{\mathcal P} }
\qquad
\infer{\DerNeg{\Gamma}{A\andN B,\Delta}{} {\mathcal P,l}} 
    {\DerNeg{\Gamma}{A,\Delta}{} {\mathcal P,l} 
    \quad \DerNeg{\Gamma}{B,\Delta}{}{\mathcal P,l} }
\]
\[
\infer{\DerNeg{\Gamma}{A\andN B,\Delta}{} {\mathcal P}} 
    {\DerNeg{\Gamma}{A,B,\Delta}{} {\mathcal P} }
\qquad
\infer{\DerNeg{\Gamma}{A\orN B,\Delta}{} {\mathcal P,l}} 
    {\DerNeg{\Gamma}{A,B,\Delta}{} {\mathcal P,l} }
\]
\[
\infer{\DerNeg{\Gamma}{\FA x A,\Delta}{} {\mathcal P}} 
    {\DerNeg{\Gamma}{A,\Delta}{} {\mathcal P} }
\qquad
\infer{\DerNeg{\Gamma}{\FA x A,\Delta}{} {\mathcal P,l}} 
    {\DerNeg{\Gamma}{A,\Delta}{} {\mathcal P,l} }
\]
\[
\infer{\DerNeg{\Gamma}{\bot^-,\Delta}{} {\mathcal P}} 
    {\DerNeg{\Gamma}{\Delta}{} {\mathcal P} }
\qquad
\infer{\DerNeg{\Gamma}{\bot^-,\Delta}{} {\mathcal P,l}} 
    {\DerNeg{\Gamma}{\Delta}{} {\mathcal P,l} }
\]
\[
\infer{\DerNeg{\Gamma}{\top^-,\Delta}{} {\mathcal P}} 
    {}
\qquad
\infer{\DerNeg{\Gamma}{\top^-,\Delta}{} {\mathcal P,l}} 
    {}
\]

\item  ($\Store$): 
We assume 
\[
\infer[\mbox{$A$ is a literal or is $\mathcal P$-positive}]{\DerNeg{\Gamma}{A,\Del}{}{\mathcal P}}
{\DerNeg{\Gamma,\non {A} }{\Del}{}{\polar {\non {A}}}}
\]

Notice that  $A$ is either a literal or a $\mathcal P,l$-positive formula, so can prove
\[\infer{\DerNeg{\Gamma}{A,\Delta}{} {\mathcal P,l}} 
        {\DerNeg{\Gamma,\non A}{\Delta}{} {\polar[\mathcal P,l]{\non A} } }
 	\]
provided we can prove the premiss.

\begin{itemize}
\item If $A \neq l$, then $\polar[\mathcal P,l]{\non A}=\polar[\mathcal P]{\non A},l$ and applying the induction hypothesis finishes the proof (unless $A=\non l$ in which case the derivable sequent $\DerNeg{\Gamma,\non A}{\Delta}{} { \polar[\mathcal P]{\non A}}$ is the same as the premiss to be proved);
\item If $A=l$, we build
  \[
  \Infer[{[({\Store})]}]{
    \DerNeg{\Gamma}{l,\Del}{}{\mathcal P,l}
  }{
    \infer[{[({\Store})]}]{
      \DerNeg{\Gamma,\Gamma'}{l}{}{\mathcal P',l}
    }{
      \infer{\DerNeg{\Gamma,\non {l},\Gamma' } {}{}{\mathcal P',l}}{
        \infer[{[({\Init[1]})]}]{
          \DerPos{\Gamma,\non {l},\Gamma' } { {l}}{}{\mathcal P',l}
        }{
          \atmCtxtP{\mathcal P',l}{\Gamma,\non {l},\Gamma' },\non l \models_{\mathcal T} 
        }
      }
    }
  }
  \]
for some $\mathcal P' \supseteq \mathcal P$ and some $\Gamma'\supseteq\atmCtxtP{\mathcal L}{\non \Del}$.
The closing condition $\atmCtxtP{\mathcal P',l}{\Gamma,\non {l},\Gamma' },\non l \models_{\mathcal T} $ holds, since
 $\atmCtxtP{\mathcal P, l}{\Gamma,\non l,\non \Del}, {\non l} \subseteq \atmCtxtP{\mathcal P',l}{\Gamma, \non {l}, \atmCtxtP{\mathcal L}{\non \Del}},{\non l}$ is assumed inconsistent.
\end{itemize}

\item $(\Select)$: We assume
  \[
  \infer[\begin{array}l
      A \mbox{ is not $\mathcal P$-negative}\\
      \non A\in \Gamma
  \end{array}]{
    \DerNeg{\Gamma}{}{}{\mathcal P}
  }{
    \DerPos{\Gam}{A}{}{\mathcal P}
  }
  \]
  \begin{itemize}
  \item If $A\neq \non l$, then $A$ is not $\mathcal P, l$-negative and we can use the induction hypothesis (invertibility of $\Pol_i$) to construct:
    \[
    \infer{\DerNeg{\Gamma}{}{}{\mathcal P,l}}{
      \DerPos{\Gam}{A}{}{\mathcal P,l}
    }
    \]
  \item If $A=\non l$, then $l\in\Gamma$ and the hypothesis can only be derived by
    \[ 
    \infer{\DerNeg{\Gamma}{}{}{\mathcal P}}{
      \infer{\DerPos{\Gam}{\non l}{}{\mathcal P}}{
        \infer{\DerNeg{\Gam}{\non l}{}{\mathcal P}}{
          {\DerNeg{\Gam,l}{}{}{\mathcal P, l}}{
          } 
        }
      }
    }
    \]
    as $\polar l=\mathcal P,l$; then we can construct:
    \[
    \infer[{[(\contr)]}]{\DerNeg{\Gamma}{}{}{\mathcal P,l}}{
      \DerNeg{\Gam,l}{}{}{\mathcal P,l}
    }
    \]
  \end{itemize}

\item ($\Init[2]$): 
We assume
$$
\infer{\DerNeg{\Gamma}{}{}{\mathcal P}}
{\atmCtxtP{\mathcal P}{\Gamma}\models_\mathcal T}
$$
We build
$$
\infer{\DerNeg{\Gamma}{}{}{\mathcal P,l}}
{\atmCtxtP{\mathcal P,l}{\Gamma}\models_\mathcal T}
$$

\item ($\andP$),($\orP$),($\EX$),($\top^+$)\\
Again, for these rules, whatever is done with the polarisation set $\mathcal P$ can be done with the polarisation set $\mathcal P,l$:

\[
\infer{\DerPos{\Gamma}{A\andP B}{} {\mathcal P}} 
    {\DerPos{\Gamma}{A}{} {\mathcal P} 
    \quad \DerPos{\Gamma}{B}{}{\mathcal P} }
\qquad
\infer{\DerPos{\Gamma}{A\andP B}{} {\mathcal P,l}} 
    {\DerPos{\Gamma}{A}{} {\mathcal P,l} 
    \quad \DerPos{\Gamma}{B}{}{\mathcal P,l} }
\]
\[\infer{\DerPos{\Gamma}{A_1\orP A_2}{} {\mathcal P}} 
    {\DerPos{\Gamma}{A_i}{} {\mathcal P} }
\qquad
\infer{\DerPos{\Gamma}{A_1\orP A_2}{} {\mathcal P,l}} 
    {\DerPos{\Gamma}{A_i}{} {\mathcal P,l} }
\]
\[
\infer{\DerPos{\Gamma}{\EX x A}{} {\mathcal P}} 
    {\DerPos{\Gamma}{\subst A x t}{} {\mathcal P} }
\qquad
\infer{\DerPos{\Gamma}{\EX x A}{} {\mathcal P,l}} 
    {\DerPos{\Gamma}{\subst A x t}{} {\mathcal P,l} }
\]
\[
\infer{\DerPos{\Gamma}{\top^+}{} {\mathcal P}} 
    {}
\qquad
\infer{\DerPos{\Gamma}{\top^+}{} {\mathcal P,l}} 
    {}
\]

\item (\Release): We assume
$$ \infer{\DerPos{\Gamma}{A}{}{\mathcal P}}
{\DerNeg{\Gamma}{A}{}{\mathcal P}}
$$
where $A$ is not $\mathcal P$-positive.

\begin{itemize}
\item If $A\neq l$, then we build: 
$$ \infer{\DerPos{\Gamma}{A}{}{\mathcal P,l}}
{\DerNeg{\Gamma}{A}{}{\mathcal P,l}}
$$
since $A$ is not $\mathcal P,l$-positive, and we close the branch by applying the induction hypothesis (invertibility of $\Pol$), whose side-condition $\atmCtxtP{\mathcal P,l}{\Gamma,\non A},\non l\models_{\mathcal T}$ is implied by \mbox{$\atmCtxtP{\mathcal P,l}{\Gamma},\non l\models_{\mathcal T}$}.

\item if $A=l$ then we build

\[\infer{\DerPos{\Gamma}{l}{}{\mathcal P,l}}
{\atmCtxtP{\mathcal P,l}{\Gamma},\non l \models_{\mathcal T}   }\]
where $\atmCtxtP{\mathcal P,l}{\Gamma},\non l \models_{\mathcal T}$ is the side-condition of ($\Pol_i$) that we have assumed.
\end{itemize}

 \item ($\Init[1]$) We assume
 \[
 \infer{\DerPos{\Gamma}{l'}{}{\mathcal P}}{} 
 \]
 with $\atmCtxtP{\mathcal P}{\Gamma},\non{l'}\models_{\mathcal T}$ and $l'$ is $\mathcal P$-positive.\\ 
 We build:
 \[
 \infer{\DerPos{\Gamma}{l'}{}{\mathcal P,l}}{} 
 \] 
 since $l'$ is $\mathcal P,l$-positive and $\atmCtxtP{\mathcal P,l}{\Gamma},\non{l'}\models_{\mathcal T}$.
 \end{itemize}
\end{proof}

\begin{corollary}
The following rules are admissible in \LKThp:
\[
\infer[{[({\Store[=]})]}]{\DerNeg\Gamma{A,\Delta}{}{\mathcal P}}
{\DerNeg{\Gamma,\non A}{\Delta}{}{\mathcal P}}
\qquad
\infer[{[({\weak[r]})]}]{\DerNeg\Gamma{\Delta,\Delta'}{}{\mathcal P}}
{\DerNeg{\Gamma}{\Delta}{}{\mathcal P}}
\]
where $A$ is a literal or a $\mathcal P$-positive formula. 
\end{corollary}

\begin{proof}
For the first rule:
if $A$ is polarised, we use $(\Store)$ and it does not change $\mathcal P$; otherwise $A$ is an unpolarised literal $l$ and we build
\[
\infer[{[({\Store})]}]{\DerNeg\Gamma{l,\Delta}{}{\mathcal P}}
{
  \infer{\DerNeg{\Gamma,\non l}{\Delta}{}{\mathcal P,\non l}}
        {
         {\DerNeg{\Gamma,\non l}{\Delta}{}{\mathcal P}} 
        }
}
\]
The topmost inference is the invertibility of $(\Pol)$, given that $\atmCtxtP{\mathcal P, \non l}{\Gamma,\non l},l\models_{\mathcal T}$.

For the second case, we simply do a multiset induction on $\Delta'$, using rule $({\Store[=]})$ for the base case, followed by a left weakening.
\end{proof}

Now we can show that removing polarities is admissible: 

\begin{lemma}[Admissibility]
The following rules are admissible in \LKThp:
\[
\infer[{[(\Pol)]}]
{\DerNeg{\Gamma}{\Del}{}{\mathcal P}}
{\DerNeg{\Gamma}{\Del}{}{\mathcal {P},l}}
\qquad
\infer[{[(\Pol_a)]l\notin\Gamma\mbox{ or }\atmCtxtP{\mathcal {P}}{\Gamma},\non l\models_{\mathcal T}}]
{\DerPos{\Gamma}{A}{}{\mathcal P}}
{\DerPos{\Gamma}{A}{}{\mathcal {P},l}}
\]
where $l\in \UP$.
\end{lemma}

\begin{proof} 
By a simultaneous induction on the derivation of the premiss, again by case analysis on the last rule used in the assumed derivation.
\begin{itemize}
\item ($\andN$),($\orN$),($\FA$),($\bot^-$),($\top^-$)\\
For these rules, whatever is done with the polarisation set $\mathcal P,l$ can be done with the polarisation set $\mathcal P$:

\[
\infer{\DerNeg{\Gamma}{A\andN B,\Delta}{} {\mathcal P,l}} 
    {\DerNeg{\Gamma}{A,\Delta}{} {\mathcal P,l} 
    \quad \DerNeg{\Gamma}{B,\Delta}{}{\mathcal P,l} }
\qquad
\infer{\DerNeg{\Gamma}{A\andN B,\Delta}{} {\mathcal P}} 
    {\DerNeg{\Gamma}{A,\Delta}{} {\mathcal P} 
    \quad \DerNeg{\Gamma}{B,\Delta}{}{\mathcal P} }
\]
\[
\infer{\DerNeg{\Gamma}{A\orN B,\Delta}{} {\mathcal P,l}} 
    {\DerNeg{\Gamma}{A,B,\Delta}{} {\mathcal P,l} }
\qquad
\infer{\DerNeg{\Gamma}{A\andN B,\Delta}{} {\mathcal P}} 
    {\DerNeg{\Gamma}{A,B,\Delta}{} {\mathcal P} }
\]
\[
\infer{\DerNeg{\Gamma}{\FA x A,\Delta}{} {\mathcal P,l}} 
    {\DerNeg{\Gamma}{A,\Delta}{} {\mathcal P,l} }
\qquad
\infer{\DerNeg{\Gamma}{\FA x A,\Delta}{} {\mathcal P}} 
    {\DerNeg{\Gamma}{A,\Delta}{} {\mathcal P} }
\]
\[
\infer{\DerNeg{\Gamma}{\bot^-,\Delta}{} {\mathcal P,l}} 
    {\DerNeg{\Gamma}{\Delta}{} {\mathcal P,l} }
\qquad
\infer{\DerNeg{\Gamma}{\bot^-,\Delta}{} {\mathcal P}} 
    {\DerNeg{\Gamma}{\Delta}{} {\mathcal P} }
\]
\[
\infer{\DerNeg{\Gamma}{\top^-,\Delta}{} {\mathcal P,l}} 
    {}
\qquad
\infer{\DerNeg{\Gamma}{\top^-,\Delta}{} {\mathcal P}} 
    {}
\]

\item ($\Store$):   We assume
  \[\infer[\mbox { $A$ is a literal or $\mathcal P ,l$-positive} ]
          {\DerNeg{\Gamma}{A,\Delta}{} {\mathcal P,l}} 
          {\DerNeg{\Gamma,\non A}{\Delta}{} { \polar[\mathcal P,l]{\non A}} }
 	  \]

Notice that  $A$ is either a literal or a $\mathcal P$-positive formula.
\begin{itemize}
\item If $A=\non l$, we build
  \[
  \infer[{[{(\Store)}]}]{
    \DerNeg{\Gamma}{A,\Delta}{} {\mathcal P}
  }{
    \DerNeg{\Gamma,\non A}{\Delta}{} {\mathcal P,\non A } 
  }
  \]
whose premiss is the derivable sequent $\DerNeg{\Gamma,\non A}{\Delta}{} { \polar[\mathcal P,l]{\non A}}$.
\item If $A = l$, we build
\[\infer[{[{(\Store[=])}]}]{\DerNeg{\Gamma}{A,\Delta}{} {\mathcal P}} 
        {\DerNeg{\Gamma,\non A}{\Delta}{} {\mathcal P } }
 	\]
using the admissibility of $\Store[=]$, and we can prove the premiss from the induction hypothesis, as we have
$\polar[\mathcal P,l]{\non A}=\mathcal P,l$.
\item In all other cases, we build
  \[
  \infer[{[{(\Store)}]}]{
    \DerNeg{\Gamma}{A,\Delta}{} {\mathcal P}
  }{
    \DerNeg{\Gamma,\non A}{\Delta}{} {\polar{\non A} } 
  }
  \]
whose premiss is provable from the induction hypothesis, as we have $\polar[\mathcal P,l]{\non A}=\polar[\mathcal P]{\non A},l$.
\end{itemize}

\item $(\Select)$: We assume
$$
\infer[\begin{array}l
\non A \in \Gamma \\
\mbox{ and $A$ not $\mathcal P,l$-negative}
\end{array}]{\DerNeg{\Gamma}{}{}{\mathcal P,l}}
{\DerPos{\Gamma}{A}{}{\mathcal P,l}}
$$

\begin{itemize}
\item If $l\in \Gamma$ then we can build:
  \[
  \infer{\DerNeg{\Gamma} {} {} {\mathcal P} }
        {\infer{\DerPos{\Gamma} {\non l} {} {\mathcal P} }
	  {\infer{\DerNeg{\Gamma} {\non l} {} {\mathcal P} }
	    {\infer[{[(\weak)]}] {\DerNeg {\Gamma,l} {}{} {\polar  l} }
	      {\DerNeg{\Gamma}{}{}{ \polar  l}}
            }
          } 
        }
        \]
and we close with the assumption since $\polar l= \mathcal P,l$.
\item If $l \not\in \Gamma$ then $\atmCtxtP{\mathcal P, l}{\Gamma} = \atmCtxtP{\mathcal P}{\Gamma}$

Using the induction hypothesis (admissibility of $\Pol_a$) we construct :
 $$
 \infer{\DerNeg{\Gamma}{}{}{\mathcal P}}
 {\DerPos{\Gamma}{A}{}{\mathcal P}}
 $$
since $A$ is not $\mathcal P$-negative.
\end{itemize}

\item ($\Init[2]$): We assume 
$$
\infer{\DerNeg{\Gamma}{}{}{\mathcal P,l}}
{\atmCtxtP{\mathcal P,l}{\Gamma}\models_\mathcal T}
$$

\begin{itemize}
\item If $l\in \Gamma$ then again we can build:
  \[
  \infer{\DerNeg{\Gamma} {} {} {\mathcal P} }
        {\infer{\DerPos{\Gamma} {\non l} {} {\mathcal P} }
	  {\infer{\DerNeg{\Gamma} {\non l} {} {\mathcal P} }
	    {\infer[{[(\weak)]}] {\DerNeg {\Gamma,l} {}{} {\polar  l} }
	      {\DerNeg{\Gamma}{}{}{ \polar  l}}
            }
          } 
        }
        \]
and we close with the assumption since $\polar l= \mathcal P,l$.
\item If $l\not\in \Gamma$, $\atmCtxtP{\mathcal P,l}{\Gamma}=\atmCtxtP{\mathcal P}{\Gamma}$,
then we can build:
$$
\infer{\DerNeg{\Gamma}{}{}{\mathcal P}}
{\atmCtxtP{\mathcal P}{\Gamma}\models_\mathcal T}
$$
\end{itemize}

\item ($\andP$),($\orP$),($\EX$),($\top^+$)\\
Again, for these rules, whatever is done with the polarisation set $\mathcal P,l$ can be done with the polarisation set $\mathcal P$:

\[
\infer{\DerPos{\Gamma}{A\andP B}{} {\mathcal P,l}} 
    {\DerPos{\Gamma}{A}{} {\mathcal P,l} 
    \quad \DerPos{\Gamma}{B}{}{\mathcal P,l} }
\qquad
\infer{\DerPos{\Gamma}{A\andP B}{} {\mathcal P}} 
    {\DerPos{\Gamma}{A}{} {\mathcal P} 
    \quad \DerPos{\Gamma}{B}{}{\mathcal P} }
\]
\[\infer{\DerPos{\Gamma}{A_1\orP A_2}{} {\mathcal P,l}} 
    {\DerPos{\Gamma}{A_i}{} {\mathcal P,l} }
\qquad
\infer{\DerPos{\Gamma}{A_1\orP A_2}{} {\mathcal P}} 
    {\DerPos{\Gamma}{A_i}{} {\mathcal P} }
\]
\[
\infer{\DerPos{\Gamma}{\EX x A}{} {\mathcal P,l}} 
    {\DerPos{\Gamma}{\subst A x t}{} {\mathcal P,l} }
\qquad
\infer{\DerPos{\Gamma}{\EX x A}{} {\mathcal P}} 
    {\DerPos{\Gamma}{\subst A x t}{} {\mathcal P} }
\]
\[
\infer{\DerPos{\Gamma}{\top^+}{} {\mathcal P,l}} 
    {}
\qquad
\infer{\DerPos{\Gamma}{\top^+}{} {\mathcal P}} 
    {}
\]

\item (\Release): We assume
$$
\infer{\DerPos{\Gamma}{A}{}{\mathcal P,l}}
{\DerNeg{\Gamma}{A}{}{\mathcal P,l}}
$$
where $A$ is not $\mathcal P,l$-positive.

By induction hypothesis (admissibility of $\Pol$) we can build:
$$
\infer{\DerPos{\Gamma}{A}{}{\mathcal P}}
{\DerNeg{\Gamma}{A}{}{\mathcal P}}
$$

\item ($\Init[1]$): We assume
$$
\infer{\DerPos{\Gamma}{l'}{}{\mathcal P,l}}
{}
$$
where $l'$ is $\mathcal P,l$-positive and $\atmCtxtP{\mathcal P,l}{\Gamma},\non {l'} \models_{\mathcal T}$.
\begin{itemize}
\item If $l' \neq l$, then $l'$ is $\mathcal P$-positive and we can build
\[
\infer{\DerPos{\Gamma}{l'}{}{\mathcal P}}{\atmCtxtP{\mathcal P}{\Gamma},\non{l'}\models_{\mathcal T}}
\]
The condition $\atmCtxtP{\mathcal P}{\Gamma},\non{l'}\models_{\mathcal T}$ holds for the following reasons:\\
If $l\notin\Gamma$, then $\atmCtxtP{\mathcal P}{\Gamma}=\atmCtxtP{\mathcal P,l}{\Gamma}$ and the condition is that of the hypothesis.\\
If $l\in\Gamma$, then the side-condition of $(\Pol_a)$ implies $\atmCtxtP{\mathcal P}{\Gamma},\non l \models_{\mathcal T}$; moreover, the condition of the hypothesis can be rewritten as
$\atmCtxtP{\mathcal P}{\Gamma},l,\non{l'} \models_{\mathcal T}$; the fact that semantical inconsistency admits cuts then proves the desired condition.

\item If $l'=l$ then we build
  \[
  \infer{\DerPos{\Gamma}{l}{}{\mathcal P}}{
    \infer{\DerNeg{\Gamma}{l}{}{\mathcal P}}{
      \DerNeg{\Gamma,\non l}{}{}{ \mathcal P, {\non l}}
    } 
  }
  \]
which we close as follows:
If $l\in \Gamma$ then we can apply $\Iden[2]$, otherwise we apply $\Init[2]$: the condition  $\atmCtxtP{\mathcal P, \non l}{\Gamma,\non l}\models_\mathcal {T}$ holds
   because $\atmCtxtP{\mathcal P, \non l}{\Gamma,\non l}=\atmCtxtP{\mathcal P}{\Gamma},\non l = \atmCtxtP{\mathcal P,l}{\Gamma},\non l$ and the condition of the hypothesis is $\atmCtxtP{\mathcal P,l}{\Gamma},\non l\models_{\mathcal T}$.
 \end{itemize}
 \end{itemize}
\end{proof}

\begin{corollary} The $({\Store[=]})$ rule is invertible, and the $(\Select[-])$ rule is admissible:
  \[
  \infer[{[({\Store[=]})]
      \mbox{\begin{tabular}l
          $A$ is literal\\
          or is $\mathcal P$-positive
      \end{tabular}}
  }]{
    \DerNeg\Gamma{A,\Delta}{}{\mathcal P}
  }{
    \DerNeg{\Gamma,\non A}{\Delta}{}{\mathcal P}
  }
  \qquad
  \qquad
    \infer[{[({\Select[-]})]}]
    {\DerNeg {\Gam,\non l} {}{} {\mathcal{P},\non l}} 
    {\DerPos {\Gam,\non l} {l} {} {\mathcal{P},\non l}}
  \]
\end{corollary}

\begin{proof}
  \begin{itemize}
  \item[$({\Store[=]})$]
    Using the invertibility of $(\Store)$, we get a proof of $\DerNeg{\Gamma,\non A}{\Delta}{}{\polar{\non A}}$. If $A$ is polarised, then $\polar {\non A}=\mathcal P$ and we are done.
    Otherwise we have a proof of $\DerNeg{\Gamma,\non A}{\Delta}{}{\mathcal P,{\non A}}$ and we apply the admissibility of $(\Pol)$ to conclude.
  \item[$({\Select[-]})$]
    We first apply the admissibility of $(\Pol_a)$ to prove 
    $\DerPos {\Gam,\non l} {l} {} {\mathcal{P}}$, then the standard $(\Select)$ rule, then the invertibility of $(\Pol_i)$ to get $\DerNeg {\Gam,\non l} {}{} {\mathcal{P},\non l}$.
  \end{itemize}
\end{proof}

\section{Cut-elimination}
\label{sec:cutelim}
Cut-elimination is an important feature of all sequent calculi. In this section we present some admissible cut-rules in \LKThp\ and show how to eliminate them.

\subsection{Cuts with the theory}

\begin{theorem}[$\cut_1$ and $\cut_2$]\strut

  The following rules are admissible in \LKThp, assuming $l\notin\UP$:
 \[\begin{array}{c}
    \infer[\cut_1]{\DerNeg {\Gam}{\Delta}{} {\mathcal P}}
    {{\atmCtxtP{\mathcal P}\Gam,\non l\models_{\mathcal T}}
    \quad \DerNeg {\Gam,l}{\Delta}{} {\mathcal P}}
    \qquad
    \infer[\cut_2]{\DerPos{\Gam}{B}{} {\mathcal P}}
    {{\atmCtxtP{\mathcal P}\Gam,\non l\models_{\mathcal T}}
	\quad 
    \DerPos {\Gam,l}{B}{} {\mathcal P}}
  \end{array}
  \]
\end{theorem}
\begin{proof}

  By simultaneous induction on the derivation of the right premiss.

  We reduce $\cut_ 1$ by case analysis on the last rule used to prove the right premiss.
	\begin{itemize}
  \item ($\andN$)
    \[
    \begin{array}c
  \infer[\cut_1]{\DerNeg {\Gam}{B\andN C, \Delta}{} {\mathcal P}}
    {
      {\atmCtxtP{\mathcal P}\Gam,\non l\models_{\mathcal T}}
	  \quad 
      \infer{\DerNeg {\Gam,l}{B\andN C,\Delta}{} {\mathcal P}}
      {
        \DerNeg {\Gam,l}{B,\Delta}{} {\mathcal P} \quad
        \DerNeg {\Gam,l}{C, \Delta}{} {\mathcal P}
      }
    }
    \\\\ \mbox{reduces to} \qquad
    \infer{\DerNeg {\Gam}{B\andN C,\Delta}{} {\mathcal P}}
    {
      \infer[\cut_1]{\DerNeg {\Gam}{B,\Delta}{} {\mathcal P}}
      {{\atmCtxtP{\mathcal P}\Gam,\non l\models_{\mathcal T}}
	 \quad \DerNeg {\Gam,l}{B,\Delta}{} {\mathcal P}}
      \quad
      \infer[\cut_1]{\DerNeg {\Gam}{C,\Delta}{} {\mathcal P}}
      {{\atmCtxtP{\mathcal P}\Gam,\non l\models_{\mathcal T}}
	 \quad \DerNeg {\Gam,l}{C,\Delta}{} {\mathcal P}}
    }
  \end{array}
  \]
\item ($\orN$)
  \[
  \infer[\cut_1]{\DerNeg {\Gam}{B_1\orN B_2,\Delta}{} {\mathcal P}}
    {
      {\atmCtxtP{\mathcal P}\Gam,\non l\models_{\mathcal T}}
	     \quad 
      \infer{
        \DerNeg {\Gam,l}{B_1\orN B_2,\Delta}{} {\mathcal P}
      }
      { \DerNeg {\Gam,l}{B_1,B_2,\Delta} {} {\mathcal P}   }
    }
    \quad\mbox{reduces to} \qquad
    \infer { \DerNeg {\Gam}{{B_1\orN B_2,\Delta}}{} {\mathcal P} }
    {  \infer[\cut_1]{\DerNeg {\Gam}{B_1, B_2,\Delta}{} {\mathcal P}}
     {
        {\atmCtxtP{\mathcal P}\Gam,\non l\models_{\mathcal T}}
		 \quad \DerNeg {\Gam,l}{B_1, B_2,\Delta}{} {\mathcal P}
      }} 
    \]
 \item ($\FA$)
   \[
    \infer[\cut_1]{\DerNeg {\Gam}{\FA x B,\Delta}{} {\mathcal P}}
    {
      {\atmCtxtP{\mathcal P}\Gam,\non l\models_{\mathcal T}}
	     \quad 
      \infer{
        \DerNeg {\Gam,l}{\FA x B,\Delta}{} {\mathcal P}
      }
      { \DerNeg {\Gam,l}{B,\Delta}{} {\mathcal P}    }
    }
    \qquad\mbox{reduces to} \qquad
    \infer{ \DerNeg {\Gam}{{\FA x B,\Delta}}{} {\mathcal P} }
    {  \infer[\cut_1] {\DerNeg {\Gam}{B,\Delta}{} {\mathcal P}} 
    {
    {\atmCtxtP{\mathcal P}\Gam,\non l\models_{\mathcal T}}
	  \quad \DerNeg {\Gam,l}{B,\Delta}{} {\mathcal P}
      }} 
    \]

 \item ($\Store$) where $B$ is a literal or $\mathcal P$-positive formula. 
  \[
    \infer[\cut_1]{\DerNeg {\Gam}{B,\Delta}{} {\mathcal P}}{
      {\atmCtxtP{\mathcal P}\Gam,\non l\models_{\mathcal T}}
      \quad 
      \infer{
        \DerNeg {\Gam,l}{B,\Delta}{} {\mathcal P}
      }{
        \DerNeg {\Gam,l,\non B}{\Delta}{} {\polar{\non B}}    
      }
    }
    \qquad\mbox{ reduces to} \qquad
    \infer { \DerNeg {\Gam}{B,\Delta}{} {\mathcal P} }{
      \infer[\cut_1]{\DerNeg {\Gam, \non B}{\Delta}{} {\polar{\non B}}}{
        \atmCtxtP{\polar{\non B}}{\Gam,\non B},\non l\models_{\mathcal T}
        \quad
        \DerNeg {\Gam,l,\non B}{\Delta}{} {\polar{\non B}}
      }
    } 
  \]
  
 We have $\atmCtxtP{\polar{\non B}}{\Gam,\non B},\non l\models_{\mathcal T}$ since ${\atmCtxtP{\mathcal P}{\Gam}, \non l} \subseteq {\atmCtxtP{\polar{\non B}}{\Gam,\non B}, \non l}$ and we assume semantical inconsistency to satisfy weakening.

 \item ($\bot^-$)
   \[
    \infer[\cut_1]{\DerNeg {\Gam}{\bot^-,\Delta}{}{\mathcal P}}
    {
      {\atmCtxtP{\mathcal P}\Gam,\non l\models_{\mathcal T}}
    \quad 
      \infer{
        \DerNeg {\Gam,l}{\bot^-,\Delta}{}{\mathcal P}
      }
      { \DerNeg {\Gam,l}{\Delta}{}{\mathcal P}    }
    }
    \qquad\mbox{reduces to} \qquad
    \infer { \DerNeg {\Gam}{{\bot^-,\Delta}}{}{\mathcal P} }
    {  \infer[\cut_1]{\DerNeg {\Gam}{\Delta}{}{\mathcal P}} 
    {     {\atmCtxtP{\mathcal P}\Gam,\non l\models_{\mathcal T}}
       \quad \DerNeg {\Gam,l}{\Delta} {}{\mathcal P}
      }} 
    \]

\item ($\top^-$)
  \[
    \infer[\cut_1]{\DerNeg {\Gam}{\top^-,\Delta}{}{\mathcal P}}
    {
      {\atmCtxtP{\mathcal P}\Gam,\non l\models_{\mathcal T}}
	    \quad 
      \infer{  \DerNeg {\Gam,l}{\top^-,\Delta} {}{\mathcal P}}{  }
    }
    \qquad\mbox{reduces to} \qquad
    \infer { \DerNeg {\Gam}{\top^-,\Delta}{}{\mathcal P} } {  } 
\]

\item $(\Select)$ where $\non P\in \Gamma,l$ and $P$ is not $\mathcal P$-negative. 

  If $\non P\in\Gamma$,
 \[
    \infer[\cut_1]{\DerNeg {\Gam}{}{}{\mathcal P}}
    {
{\atmCtxtP{\mathcal P}\Gam,\non l\models_{\mathcal T}}
  \quad 
      \infer{
        \DerNeg {\Gam,l}{}{}{\mathcal P}
      }
      { \DerPos {\Gam,l}{P}{}{\mathcal P}   }
    }
    \qquad\mbox{reduces to} \qquad
    \infer { \DerNeg {\Gam}{} {}{\mathcal P} }
    {  \infer[\cut_2]{\DerPos {\Gam}{P}{}{\mathcal P}} {
        {\atmCtxtP{\mathcal P}\Gam, \non l\models_{\mathcal T}}
	    \quad \DerPos {\Gam, l}{P} {}{\mathcal P}
      }} 
  \]

  If $\non P=l$, then as $P$ is not $\mathcal P$-negative and $l\notin\UP$ we get that $\non l$ is $\mathcal P$-positive, so
 \[
    \infer[\cut_1]{\DerNeg {\Gam}{}{}{\mathcal P}}{
      {\atmCtxtP{\mathcal P}\Gam,\non l\models_{\mathcal T}}
      \quad 
      \infer{
        \DerNeg {\Gam,l}{}{}{\mathcal P}
      }{
        \infer{\DerPos {\Gam,l}{\non l}{}{\mathcal P}}{
          \atmCtxtP{\mathcal P}\Gam,l\models_{\mathcal T}
        }
      }
    }
    \qquad\mbox{reduces to} \qquad
    \infer[{{\Init[2]}}]{
      \DerNeg {\Gam}{} {}{\mathcal P} 
    }{
      \atmCtxtP{\mathcal P}\Gam\models_{\mathcal T}
    }
  \]
  since semantical inconsistency admits cuts.

 \item ($\Init[2]$)

   \[
    \infer[\cut_1]{\DerNeg {\Gam}{}{}{\mathcal P}}
    {
      {\atmCtxtP{\mathcal P}\Gam,\non l\models_{\mathcal T}}
      \quad 
      \infer{
        \DerNeg {\Gam,l}{}{}{\mathcal P}
      }
      { \atmCtxtP{\mathcal P}\Gam,l\models_{\mathcal T} 
      }
    }
    \qquad\mbox{reduces to} \qquad
      \infer{\DerNeg {\Gam} {} {}{\mathcal P}}
      {\atmCtxtP{\mathcal P}\Gam\models_{\mathcal T}}
   \]
  since semantical inconsistency admits cuts. 
\end{itemize}

  We reduce $\cut_2$ again by case analysis on the last rule used to prove the right premiss.
\begin{itemize}
 \item ($\andP$) 
  \[
  \begin{array}c
  \infer[\cut_2]{\DerPos {\Gam}{B\andP C}{}{\mathcal P}}
    {
      {\atmCtxtP{\mathcal P}\Gam,\non l\models_{\mathcal T}}
	 \quad 
      \infer{
        \DerPos {\Gam,l}{B\andP C}{}{\mathcal P}
      }
      {
        \DerPos {\Gam,l}{B}{}{\mathcal P}
        \quad\DerPos {\Gam,l}{C}{}{\mathcal P}
      }
    }   
    \hfill\strut
    \\\\ \mbox{reduces to}\\\\
    \infer{\DerPos {\Gam}{B\andP C}{}{\mathcal P}}
    {
      \infer[\cut_2]{\DerPos {\Gam}{B}{}{\mathcal P}}
      {
      {\atmCtxtP{\mathcal P}\Gam,\non l\models_{\mathcal T}}
	   \quad \DerPos {\Gam,l}{B}{}{\mathcal P}}
    	  \quad
      \infer[\cut_2]{\DerPos {\Gam}{C}{}{\mathcal P}}
      {
      {\atmCtxtP{\mathcal P}\Gam,\non l\models_{\mathcal T}}
    \quad \DerPos {\Gam,l}{C}{}{\mathcal P}}
    }
    \end{array}
    \]

 \item ($\orP$)
   \[
    \infer[\cut_2]{\DerPos {\Gam}{B_1\orP B_2}{}{\mathcal P}}
    {
     {\atmCtxtP{\mathcal P}\Gam,\non l\models_{\mathcal T}}
	  \quad 
      \infer{
        \DerPos {\Gam,l}{B_1\orP B_2}{}{\mathcal P}
      }
      { \DerPos {\Gam,l}{B_i}{}{\mathcal P}    }
    }
    \qquad\mbox{reduces to} \qquad
    \infer { \DerPos {\Gam}{B_1\orP B_2}{}{\mathcal P} }
    {  \infer[\cut_2]{\DerPos {\Gam}{B_i}{}{\mathcal P}} 
    {  {\atmCtxtP{\mathcal P}\Gam,\non l\models_{\mathcal T}}
        \quad \DerPos {\Gam,l}{B_i}{}{\mathcal P}
      }} 
\]

\item ($\EX$) 
  \[
    \infer[\cut_2] {\DerPos {\Gamma}{\EX x B}{}{\mathcal P}}
    {    {\atmCtxtP{\mathcal P}\Gam,\non l\models_{\mathcal T}}
      \quad 
      \infer{\DerPos {\Gam,l}{\EX x B} {}{\mathcal P}}
      { \DerPos {\Gam,l} {\subst B x t}{}{\mathcal P}    }
    }
    \qquad\mbox{reduces to} \qquad
    \infer { \DerPos {\Gam}{\EX x B}{}{\mathcal P} }
    {  \infer[\cut_2]{\DerPos {\Gam} {\subst B x t} {}{\mathcal P}} 
    {
        {\atmCtxtP{\mathcal P}\Gam,\non l\models_{\mathcal T}}
        \quad \DerPos {\Gam,l}{\subst B x t}{}{\mathcal P}
      }} 
    \]

\item ($\top^+$) 
  \[
    \infer[\cut_2]{\DerPos {\Gam}{\top^+}{}{\mathcal P}}
    {
      {\atmCtxtP{\mathcal P}\Gam,\non l\models_{\mathcal T}}
          \quad 
      \infer{  \DerPos {\Gam,l}{\top^+} {}{\mathcal P}} { }
    }
    \qquad\mbox{reduces to} \qquad
    \infer { \DerPos {\Gam}{\top^+}{}{\mathcal P} }
      {  } 
      \]
 \item ($\Release$)  
   \[
    \infer[\cut_2]{\DerPos {\Gam}{N}{}{\mathcal P}}
    {
	{\atmCtxtP{\mathcal P}\Gam,\non l\models_{\mathcal T}}
     \quad 
      \infer{\DerPos {\Gam,l}{N} {}{\mathcal P}}
      { \DerNeg {\Gam,l}{N}{}{\mathcal P}    }
    }
    \qquad\mbox{reduces to} \qquad
    \infer { \DerPos {\Gam}{N}{}{\mathcal P} }
    {  \infer[\cut_1]{\DerNeg {\Gam}{N}{}{\mathcal P}} 
    {
        {\atmCtxtP{\mathcal P}\Gam,\non l\models_{\mathcal T}} 
        \quad \DerNeg {\Gam,l}{N}{}{\mathcal P}
      }} 
\]
%
 \item ($\Init[1]$)
  $$ \begin{array}{c}
    \infer[\cut_2]{\DerPos {\Gam}{p}{}{\mathcal P}}
    {
      {\atmCtxtP{\mathcal P}\Gam,\non {l}\models_{\mathcal T}}
      \quad 
      \infer{
        \DerPos {\Gam,l}{p}{}{\mathcal P}
      }
      {\atmCtxtP{\mathcal P}\Gam,l,\non {p}\models_{\mathcal T}}
      }
    \qquad\mbox{reduces to} \qquad
    \infer{\DerPos{\Gamma  }{p}{}{\mathcal P}} 
    {\atmCtxtP{\mathcal P}\Gam,\non {p}\models_{\mathcal T}}
    \end{array}
  $$
since weakening gives $\atmCtxtP{\mathcal P}\Gam, \non l,\non {p} \models_\mathcal{T}$ and semantical inconsistency admits cuts.

\end{itemize}
\end{proof}

\subsection{Safety and instantiation}

Now we would like to prove the admissibility of other cuts, where both premisses are derived as a judgement of \LKThp. Unfortunately, the expected cut-rules are not necessarily admissible unless we consider the following notion of safety.

\begin{definition}[Safety]\strut
  \label{lem:Pinv}
  \begin{itemize}
  \item
    A pair $(\Gamma,\mathcal P)$ (of a context and a polarisation set) is said to be \Index[safety]{safe} if:\\
    for all $\Gamma'\supseteq\Gamma$, for all semantically consistent sets of literals $\mathcal R$ with $\atmCtxtP{\mathcal P}{\Gamma'}\subseteq\mathcal R\subseteq \atmCtxtP{\mathcal P}{\Gamma'}\cup \closubst\UP$,
    and for all $\mathcal P$-positive literal $l$, if $\mathcal R,\non l\models_\mathcal{T}$ then $\atmCtxtP{\mathcal P}{\Gamma'},\non l\models_\mathcal{T}$.
  \item
    A sequent $\DerPos \Gam A {}{\mathcal P}$ (\resp $\DerNeg \Gam \Delta {}{\mathcal P}$) is said to be \Index{safe}\\ if the pair $(\Gamma,\mathcal P)$ (\resp $((\Gamma,\non \Delta),\mathcal P)$) 
is safe.
  \end{itemize}
\end{definition}

\begin{remark}
Safety is a property that is monotonic in its first argument: if $(\Gamma,\mathcal P)$ is safe and $\Gamma\subseteq\Gamma'$ then $(\Gamma',\mathcal P)$ is safe (this property is built into the definition by the quantification over $\Gamma'$).
\end{remark}

When restricted to safe sequents, the expected cuts are indeed admissible.
In order to show that the safety condition is not very restrictive, we show the following lemma:
\begin{lemma}[Cases of safety]\strut
  \begin{enumerate}
  \item Empty theory:\\
    When the theory is empty (semantical inconsistency coincides with syntactical inconsistency), the safety of $(\Gamma,\mathcal P)$ means that either $\atmCtxtP{\mathcal P}{\Gamma}$ is syntactically inconsistent, or every $\mathcal P$-positive literal that is an instance of a $\mathcal P$-unpolarised literal must be in $\Gamma$\hfill (\ie $\mathcal P\cap\closubst\UP\subseteq\Gamma$).\\
    In the particular case of propositional logic ($\subst l x t=l$ for every $l\in\mathcal L$), every sequent is safe.
  \item Full polarisation:\\
    When every literal is polarised ($\UP=\emptyset$), every sequent (with polarisation set $\mathcal P$) is safe.
  \item No polarisation:\\
    When every literal is unpolarised ($\UP=\mathcal L$), every sequent (with polarisation set $\mathcal P$) is safe.
  \item Safety is an invariant of proof-search:\\
    for every rule of \LKThp, if its conclusion is safe then each of its premisses is safe.
  \end{enumerate}
\end{lemma}
\begin{proof}
  \begin{enumerate}
  \item In the case of the empty theory, if $\mathcal R$ is consistent then $\mathcal R,\non l\models_\mathcal{T}$ means that $l\in\mathcal R$, so either $l\in\atmCtxtP{\mathcal P}{\Gamma'}$ or $l\in\closubst\UP$; that this should imply $\atmCtxtP{\mathcal P}{\Gamma'},\non l\models_\mathcal{T}$ means that $l\in\atmCtxtP{\mathcal P}{\Gamma'}$ anyway, unless $\atmCtxtP{\mathcal P}{\Gamma'}$ is syntactically inconsistent. In particular for $\Gamma'=\Gamma$.\\
In the case of propositional logic, there are no $\mathcal P$-positive literals that are in $\closubst\UP=\UP$, so every sequent is safe. 
  \item When every literal is polarised ($\UP=\emptyset$), then $\mathcal R=\atmCtxtP{\mathcal P}{\Gamma'}$ and the result is trivial.
  \item When every literal is unpolarised ($\UP=\mathcal L$), the property holds trivially.
  \item For every rule of \LKThp, if its conclusion is safe then each of its premisses is safe.\\
    Every rule is trivial (considering monotonicity) except $(\Store)$, for which it suffices to show:\\
    Assume $(\Gamma,\mathcal P)$ is safe and $A\in\Gamma$; then $(\Gamma,(\polar A))$ is safe.\\
    Consider $\Gamma'\supseteq\Gamma$ and $\mathcal R$ such that 
    $\atmCtxtP{\polar A}{\Gamma'}\subseteq\mathcal R\subseteq \atmCtxtP{\polar A}{\Gamma'}\cup \closubst{\UP[\polar A]}$.
    \begin{itemize}
    \item
      If $A\in\UP$, then $\polar A=\mathcal P,A$ and the inclusions can be rewritten as
      \[\atmCtxtP{\mathcal P}{\Gamma'},A\ \subseteq\mathcal R\subseteq \atmCtxtP{\mathcal P}{\Gamma'}, A\ \cup\closubst{\UP[\mathcal P, A]}\]
      Since $\UP[\mathcal P, A]\subseteq\UP$ we have $\closubst{\UP[\mathcal P, A]}\subseteq\closubst\UP$ and therefore 
      \[\atmCtxtP{\mathcal P}{\Gamma'} \subseteq\mathcal R\subseteq \atmCtxtP{\mathcal P}{\Gamma'}\cup\closubst\UP\]
      Hence, $\mathcal R$ is a set for which safety of $(\Gamma,\mathcal P)$ implies $\atmCtxtP{\mathcal P}{\Gamma'},\non l\models_\mathcal{T}$ for every $l\in\mathcal P$ such that $\mathcal R,\non l\models_\mathcal{T}$.\\
      For $l=A$, then trivially $\atmCtxtP{\mathcal P,A}{\Gamma'},\non l\models_\mathcal{T}$ as $A\in\Gamma'$.
    \item If $A\not\in\UP$, then $\polar A=\mathcal P$ and the result is trivial.
    \end{itemize}
  \end{enumerate}

\end{proof}

Now cut-elimination in presence of quantifiers relies heavily on the fact that, if a proof can be constructed with a free variables $x$, then it can be replayed when $x$ is instantiated by a particular term throughout the proof.
In a polarised world, this is made difficult by the fact that a polarisation set $\mathcal P$ (\ie a set that is syntactically consistent) might not remain a polarisation set after instantiation (\ie $\subst{\mathcal P}x t$ might not be syntactically consistent: imagine $p(x,3)$ is $\mathcal P$-positive and $p(3,x)$ is $\mathcal P$-negative, then after substituting $3$ for $x$, what is the polarity of $p(3,3)$?). Hence, polarities will have to be changed and therefore the exact same proof may not be replayed, but under the hypothesis that the substituted sequent is safe, we manage to reconstruct \emph{some} proof. The first step to prove this is the following lemma: 

\begin{lemma}[Admissibility of instantiation with the theory]\label{lem:adminstant}
  Let $\mathcal P$ be a polarisation set such that $x\not\in\FV{\mathcal P}$, let $l_1,\ldots,l_n$ be $n$ literals, $\mathcal A$ be a set of literals, $x$ be a variable and $t$ be a term with $x\notin\FV t$.

  Let $\mathcal P_i\eqdef \polar[\polar[\polar{l_1}]{\ldots}]{l_i}$ with $\mathcal P_0\eqdef \mathcal P$, and similarly let $\mathcal P'_i\eqdef \polar[\polar[\polar{\subst{l_1}xt}]{\ldots}]{\subst{l_i}xt}$ with $\mathcal P'_0\eqdef \mathcal P$.

Assume
  \begin{itemize}
  \item for all $i$ such that $1\leq i\leq n$, we have $l_i\in\Gamma$;
  \item $(\subst \Gam x t,\mathcal P'_n)$ is safe;
  \item $\atmCtxtP{\mathcal P_n}\Gam,\mathcal A\models_{\mathcal T}$.
  \end{itemize}

Then either
$\atmCtxtP{\mathcal P'_n}\Gam,\subst{\mathcal A}xt\models_{\mathcal T}$ or 
$\DerNeg {\subst \Gam xt} {}{} {\mathcal P'_n}$ is derivable in \LKThp.
\end{lemma}
\begin{proof}
  Let $\{l'_1,\ldots,l'_m\}$ be the set of literals $\{l\in\atmCtxtP{\mathcal P_n}\Gam\sep \subst lxt\mbox{ is not $\mathcal P'_n$-positive}\}$.
  We have 
  \[\subst{\atmCtxtP{\mathcal P_n}\Gam}xt\subseteq\atmCtxtP{\mathcal P'_n}{\subst\Gam xt},\subst{l'_1}xt,\ldots,\subst{l'_m}xt\]

  Since $\atmCtxtP{\mathcal P_n}\Gam,\mathcal A\models_{\mathcal T}$ and semantical inconsistency is stable under instantiation and weakening, we have
  $\atmCtxtP{\mathcal P'_n}{\subst\Gam xt},\subst{l'_1}xt,\ldots,\subst{l'_m}xt,\subst{\mathcal A}xt\models_{\mathcal T}$.
  \begin{itemize}
  \item
    If all of the sets $(\atmCtxtP{\mathcal P'_n}{\subst\Gam xt},\subst{\non {l'_j}}xt)_{1\leq j\leq n}$ are semantically inconsistent, then from 
    \[\atmCtxtP{\mathcal P'_n}{\subst\Gam xt},\subst{l'_1}xt,\ldots,\subst{l'_m}xt,\subst{\mathcal A}xt\models_{\mathcal T}\]
    we get $\atmCtxtP{\mathcal P'_n}{\subst\Gam xt},\subst{\mathcal A}xt\models_{\mathcal T}$, since semantically inconsistency admits cuts.
  \item
    Otherwise, there is some $l'_j\in\atmCtxtP{\mathcal P_n}\Gam$ such that $\subst {l'_j}xt$ is not $\mathcal P'_n$-positive and such that $\mathcal R\eqdef\atmCtxtP{\mathcal P'_n}{\subst\Gam xt},\subst{\non {l'_j}}xt$ is semantically consistent.

    Notice that $l'_j$ is not $\mathcal P$-positive, otherwise $\subst {l'_j}xt$ would also be $\mathcal P$-positive (since $x\notin\FV{\mathcal P}$), so $l'_j=l_i$ for some $i$ such that $1\leq i\leq n$, with $l_i\in\UP[\mathcal P_{i-1}]$.

    Now, if $\subst \Gam xt$ is syntactically inconsistent, we build
    \[    
    \infer[{[{\Id[2]}]}]{
      \DerNeg {\subst \Gam xt} {}{} {\mathcal P'_n}
    }{}
    \]

    If on the contrary $\subst \Gam xt$ is syntactically consistent, then $\{\subst{l_1}xt,\ldots,\subst{l_n}xt\}$ is also syntactically consistent (as every element is assumed to be in $\subst \Gam xt$).

    Therefore, $\subst {l_i}xt$ must be $\mathcal P$-negative, otherwise it would ultimately be $\mathcal P'_{n}$-positive.\\
    So $\subst {\non l_i}xt$ is $\mathcal P$-positive, and ultimately $\mathcal P'_{n}$-positive.

    Now $(\subst \Gam xt, \mathcal P'_n)$ is assumed to be safe, so we want to apply this property to $\Gamma'\eqdef\Gam$, to the semantically consistent set $\mathcal R$, and to the $\mathcal P'_{n}$-positive literal $\subst {\non l_i}xt$, so as to conclude
    \[\atmCtxtP{\mathcal P'_n}{\subst\Gam xt},\subst{l_i}xt\models_{\mathcal T}\]
    To apply the safety property, we note that that $\mathcal R,\subst{l_i}xt\models_{\mathcal T}$ and that
    \[\atmCtxtP{\mathcal P'_n}{\subst\Gam xt}\subseteq\mathcal R\subseteq\atmCtxtP{\mathcal P'_n}{\subst\Gam xt}\cup\closubst{\UP[\mathcal P'_n]}\]
    provided we have $l_i\in\UP[\mathcal P'_n]$.

    In order to prove that proviso, first notice that $l_i\in\UP[\mathcal P]$, since $l_i\in\UP[\mathcal P_i]$. Now we must have $x\in\FV{l_i}$, otherwise $l_i=\subst{l_i}xt$ and we know that $\subst{l_i}xt$ is $\mathcal P$-negative. Since none of the literals $(\subst{l_k}xt)_{1\leq k\leq n}$ have $x$ as a free variable, we conclude the proviso $l_i\in\UP[\mathcal P'_n]$.

    Therefore safety ensures $\atmCtxtP{\mathcal P'_n}{\subst\Gam xt},\subst{l_i}xt\models_{\mathcal T}$
    and we can finally build
    \[    
    \infer[{[\Select]}]{
      \DerNeg {\subst \Gam xt} {}{} {\mathcal P'_n}
    }{
      \infer[{[{\Init[1]}]}]
            {\DerPos {\subst \Gam xt} {\subst{\non l_i}xt}{} {\mathcal P'_n}}
            {}
    }
    \]
    as $\subst{\non l_i}xt$ is $\mathcal P'_n$-positive.
  \end{itemize}
\end{proof}

We can finally state and prove the admissibility of instantiation:

\begin{lemma}[Admissibility of instantiation]
\label{Ladm}
  Let $\mathcal P$ be a polarisation set such that $x\not\in\FV{\mathcal P}$, let $l_1,\ldots,l_n$ be $n$ literals, $x$ be a variable and $t$ be a term with $x\notin\FV t$.

  Let $\mathcal P_i\eqdef \polar[\polar[\polar{l_1}]{\ldots}]{l_i}$ with $\mathcal P_0\eqdef \mathcal P$, and similarly let $\mathcal P'_i\eqdef \polar[\polar[\polar{\subst{l_1}xt}]{\ldots}]{\subst{l_i}xt}$ with $\mathcal P'_0\eqdef \mathcal P$.

  The following rules are admissible in \LKThp:\footnote{The admissibility of $(\Inst[f])$ means that if $\DerPos \Gam B {}{\mathcal P_n}$ is derivable in \LKThp\ then either $\DerPos {\subst \Gam x t}{\subst B x t} {}{\mathcal P'_n}$ or $\DerNeg {\subst \Gam x t}{} {}{\mathcal P'_n}$ is derivable in \LKThp.}
  \[
  \infer[{[({\Inst})]}]{
    \DerNeg {\subst \Gam x t}{\subst \Del x t}{}{\mathcal P'_n}
  }{
    \DerNeg \Gam \Del {}{\mathcal P_n}
  }
  \qquad
  \infer[{[({\Inst[f]})]}]{
    \DerPos {\subst \Gam x t}{\subst B x t} {}{\mathcal P'_n}
    \mbox{ or }
    \DerNeg {\subst \Gam x t}{} {}{\mathcal P'_n}    
  }{
    \DerPos \Gam B {}{\mathcal P_n}
  }
  \]
  where we assume
  \begin{itemize}
  \item for all $i$ such that $1\leq i\leq n$, we have $l_i\in\Gamma$;
  \item $\DerNeg {\subst \Gam x t}{\subst \Del x t}{}{\mathcal P'_n}$ is safe in $(\Inst)$;
  \item $({\subst \Gam x t},{\mathcal P'_n})$ is safe in $(\Inst[f])$.
  \end{itemize}
\end{lemma}
\begin{proof}
  By induction on the derivation of the premiss.


  



  \begin{itemize}
  \item ($\andN$),($\orN$),($\FA$),($\bot^-$),($\top^-$),($\andP$),($\orP$),($\EX$),($\top^+$)\\
    These rules are straightforward as the polarisation set is not involved.
  \item ($\Store$)
    We assume
    \[
    \infer{\DerNeg \Gam {A,\Del} {} {\mathcal P_n}}{
      \DerNeg {\Gam,\non A} {\Del} {} {\polar[\mathcal P_n]{\non A}}
    }
    \]
    where $A$ is a literal or is $\mathcal P_n$-positive.

    Using the induction hypothesis on the premiss we can build
    \[
    \infer{\DerNeg {\subst \Gam xt} {\subst Axt,\subst\Del xt} {} {\mathcal P'_n}}{
      \DerNeg {\subst\Gam xt,\subst{\non A}xt} {\subst\Del xt} {} {\polar[\mathcal P'_n]{\subst{\non A}xt}}
    }
    \]
    since $\subst Axt$ is a literal or is $\mathcal P'_n$-positive.

    \item $(\Select)$
      We assume
      \[
      \infer{\DerNeg {\Gam,\non P} {}{} {\mathcal P_n}}{
        \DerPos {\Gam,\non P} {P} {} {\mathcal P_n}
      }
      \]
      where $P$ is not $\mathcal P_n$-negative.

      If $\subst P x t$ is not $\mathcal P'_n$-negative, then we can apply the induction hypothesis and build
      \[
      \infer{\DerNeg {\subst\Gam xt,\subst{\non P}xt} {}{} {\mathcal P'_n}}{
        \DerPos {\subst\Gam xt,\subst{\non P}xt} {\subst Pxt} {} {\mathcal P'_n}
      }
      \]
      Otherwise, $\subst Pxt$ is a $\mathcal P'_n$-negative literal and we can do the same as above with the $(\Select[-])$ rule instead of $(\Select)$.

    \item $(\Init[2])$
      We assume
      \[    
      \infer{
        \DerNeg {\Gam} {}{} {\mathcal P_n}
      }{
        \atmCtxtP{\mathcal P_n}\Gam\models_{\mathcal T}
      }
      \]

      We use Lemma~\ref{lem:adminstant} with $\mathcal A\eqdef\emptyset$, since we know $\atmCtxtP{\mathcal P_n}\Gam\models_{\mathcal T}$.

      If we get $\atmCtxtP{\mathcal P'_n}{\subst\Gam xt}\models_{\mathcal T}$, we build a proof with the same rule $(\Init[2])$:
      \[    
      \infer{
        \DerNeg {\subst \Gam xt} {}{} {\mathcal P'_n}
      }{
        \atmCtxtP{\mathcal P'_n}{\subst\Gam xt}\models_{\mathcal T}
      }
      \]
      If not, we directly get a proof of $\DerNeg {\subst \Gam xt} {}{} {\mathcal P'_n}$.

    \item $(\Init[1])$
      We assume
      \[    
      \infer{
        \DerPos {\Gam} {p}{} {\mathcal P_n}
      }{
        \atmCtxtP{\mathcal P_n}\Gam,\non p\models_{\mathcal T}
      }
      \]
      where $p$ is $\mathcal P_n$-positive.

      We use Lemma~\ref{lem:adminstant} with $\mathcal A\eqdef\{p\}$, since we know $\atmCtxtP{\mathcal P_n}\Gam,\non p\models_{\mathcal T}$.

      If we get $\atmCtxtP{\mathcal P'_n}{\subst\Gam xt},\subst{\non p}xt\models_{\mathcal T}$, we build a proof with the same rule $(\Init[1])$:
      \[    
      \infer{
        \DerPos {\subst \Gam xt} {\subst pxt}{} {\mathcal P'_n}
      }{
        \atmCtxtP{\mathcal P'_n}{\subst\Gam xt},\subst{\non p}xt\models_{\mathcal T}
      }
      \]
      If not, we directly get a proof of $\DerNeg {\subst \Gam xt} {}{} {\mathcal P'_n}$.

    \item $(\Release)$
      We assume
      \[
      \infer{\DerPos {\Gam} {N}{} {\mathcal P_n}}{
        \DerNeg {\Gam} {N} {} {\mathcal P_n}
      }
      \]
      where $N$ is not $\mathcal P_n$-positive.

      If $\subst N x t$ is not $\mathcal P'_n$-positive, then we can apply the induction hypothesis and build
      \[
      \infer{\DerPos {\subst\Gam xt} {\subst Nxt}{} {\mathcal P'_n}}{
        \DerNeg {\subst\Gam xt} {\subst Nxt} {} {\mathcal P'_n}
      }
      \]
      Otherwise, $N$ is a literal $l$ that is not $\mathcal P_n$-positive, but such that $\subst lxt$ is $\mathcal P'_n$-positive.
      \begin{itemize}
      \item
        If $\atmCtxtP{\mathcal P'_n}{\subst\Gam xt},\subst{l}xt\models_{\mathcal T}$, then we build
        \[    
        \infer[{[{\cut_1}]}]{
          \DerNeg {\subst \Gam xt} {}{} {\mathcal P'_n}
        }{
          \atmCtxtP{\mathcal P'_n}{\subst\Gam xt},\subst{l}xt\models_{\mathcal T}
          \qquad
          \DerNeg {\subst\Gam xt,\subst{\non l}xt}{}{} {\mathcal P'_n} 
        }
        \]
        where the right premiss is proved as follows:

        Notice that the assumed derivation of $\DerNeg {\Gam} {l} {}
        {\mathcal P_n}$ necessarily contains a sub-derivation
        concluding $\DerNeg {\Gam,\non l} {} {} {\polar[\mathcal
            P_n]{\non l}}$, and applying the induction hypothesis on
        this yields a derivation of $\DerNeg {\subst\Gam
          xt,\subst{\non l}xt}{}{} {\mathcal P'_n}$.
      \item Assume now that $\mathcal R\eqdef \atmCtxtP{\mathcal P'_n}{\subst\Gam xt},\subst{l}xt$ is semantically consistent.
        We build
        \[    
        \infer[{[{\Init[1]}]}]{
          \DerPos {\subst \Gam xt} {\subst{l}xt}{} {\mathcal P'_n}
        }{}
        \]
        and we have to prove the side-condition $\atmCtxtP{\mathcal P'_n}{\subst\Gam xt},\subst{\non l}xt\models_{\mathcal T}$.

        This is trivial if $\subst{l}xt\in\subst\Gam xt$ (as $\subst{l}xt$ is $\mathcal P'_n$-positive).

        If on the contrary $\subst{l}xt\notin\subst\Gam xt$, then we get it from the assumed safety of $(\subst \Gam xt, \mathcal P'_n)$, applied to $\Gamma'\eqdef\Gam$, to the semantically consistent set $\mathcal R$, and to the $\mathcal P'_{n}$-positive literal $\subst {l}xt$.
        To apply the safety property, we note that $\mathcal R,\subst{\non l}xt\models_{\mathcal T}$ and that
        \[\atmCtxtP{\mathcal P'_n}{\subst\Gam xt}\subseteq\mathcal R\subseteq\atmCtxtP{\mathcal P'_n}{\subst\Gam xt}\cup\closubst{\UP[\mathcal P'_n]}\]
        provided we have $\subst{l}xt\in\closubst{\UP[\mathcal P'_n]}$.

        We prove that $l\in\UP[\mathcal P'_n]$ as follows:\\ First
        notice that $l\in\UP[\mathcal P]$, otherwise $l$ would be
        $\mathcal P$-negative and so would be $\subst{l}xt$ (since
        $x\notin\FV{\mathcal P}$).  Then notice that $\subst{l}xt$
        must be $\mathcal P$-positive, since it is $\mathcal
        P'_n$-positive but $\subst{l}xt\notin\subst\Gam xt$.
        Therefore $l\neq\subst{l}xt$, so $x\in\FV l$, and finally we
        get $l\in\UP[\mathcal P'_n]$, since none of the literals
        $(\subst{l_k}xt)_{1\leq k\leq n}$ have $x$ as a free variable.
      \end{itemize}
  \end{itemize}

\end{proof}

\subsection{More general cuts}

\begin{theorem}[$\cut_3$, $\cut_4$ and $\cut_5$]
  The following rules are admissible in \LKThp:\footnote{The admissibility of $\cut_5$ means that if $\DerNeg \Gam {N}{}{\mathcal P}$ and $\DerPos {\Gam,N}{B}{}{ \polar N}$ are derivable in \LKThp\ then either $\DerPos {\Gam}{B}{}{{\mathcal P}}$ or $\DerNeg {\Gam}{}{}{\mathcal P}$ is derivable in \LKThp.}
  \[\begin{array}{c}
    \infer[{[(\cut_3)]}]{
      \DerNeg {\Gam}{\Del}{}{\mathcal P}
    }{
      \DerPos \Gam A {}{\mathcal P} 
      \quad 
      \DerNeg {\Gam}{\non A,\Del}{}{\mathcal P}
    }
    \\\\
    \infer[{[(\cut_4)]}]{
      \DerNeg {\Gam}{\Del}{}{\mathcal P}
    }{
      \DerNeg \Gam {N}{}{\mathcal P}
      \quad 
      \DerNeg {\Gam,N}{\Del}{}{\polar N}
    }
    \qquad
    \infer[{[(\cut_5)]}]{
      \DerPos {\Gam}{B}{}{\mathcal P} \mbox{ or } \DerNeg {\Gam}{}{}{\mathcal P} 
    }{
      \DerNeg \Gam {N}{}{\mathcal P} 
      \quad 
      \DerPos {\Gam,N}{B}{}{\polar N} 
    }
  \end{array}
  \]
  where
  \begin{itemize}
  \item $N$ is assumed to not be $\mathcal P$-positive in $\cut_4$ and $\cut_5$;
  \item the sequent $\DerNeg {\Gam}{\Del}{}{\mathcal P}$ in $\cut_3$ and $\cut_4$, and the pair $(\Gamma,\mathcal P)$ in $\cut_5$, are all assumed to be safe.
  \end{itemize}
\end{theorem}

\begin{proof}
  By simultaneous induction on the following lexicographical measure:
  \begin{itemize}
  \item the size of the cut-formula ($A$ or $N$)
  \item the fact that the cut-formula ($A$ or $N$) is positive or negative\\
    (if of equal size, a positive formula is considered smaller than a negative formula)
  \item the height of the derivation of the right premiss
  \end{itemize}

  Weakenings and contractions (as they are admissible in the system) are implicitly used throughout this proof.

  In order to eliminate $\cut_3$, we analyse which rule is used to prove the left premiss. We then use invertibility of the negative phase so that the last rule used in the right premiss is its dual one.

  \begin{itemize}
\item  ($\andP$)
  $$ 
    \infer[\cut_3]{\DerNeg{\Gamma}{\Delta}{}{\mathcal P}}
    {\infer{\DerPos{\Gamma}{A \andP B}{}{\mathcal P}} 
    	{ {\DerPos{\Gamma}{A}{}{\mathcal P}} 
			\quad 
			{\DerPos{\Gamma}{B}{}{\mathcal P}}} 
	\qquad 
		\infer{\DerNeg{\Gamma}{A \orN B,\Delta}{}{\mathcal P}}
				{\DerNeg{\Gamma}{\non A,\non B,\Delta}{}{\mathcal P}}}
$$				
    reduces to
    $$
    \infer[\cut_3]{\DerNeg{\Gamma}{\Delta}{}{\mathcal P}}
    	{
		{\DerPos{\Gamma}{B}{}{\mathcal P}} 
		\quad {\infer[\cut_3]{\DerNeg{\Gamma}{\non B,\Delta}{}{\mathcal P}} 	   				{{\DerPos{\Gamma}{A}{}{\mathcal P}} 
		\quad {\DerNeg{\Gamma}{\non A,\non B,\Delta}{}{\mathcal P}}}}}
  $$

 \item ($\orP$)

  $$
    \infer[\cut_3]{\DerNeg{\Gamma}{\Delta}{}{\mathcal P}}
    {\infer{\DerPos{\Gamma}{A_1 \orP A_2}{}{\mathcal P}} 
    {\DerPos{\Gamma}{A_i}{}{\mathcal P}}  
    \qquad 
    \infer{\DerNeg{\Gamma}{A_1 \andN A_2,\Delta}{}{\mathcal P}}		    {{\DerNeg{\Gamma}{\non A_1,\Delta}{}{\mathcal P}} 
    \quad {\DerNeg{\Gamma}{\non A_2,\Delta}{}{\mathcal P}}}} \quad
    $$
    reduces to
   $$ \infer[\cut_3]{\DerNeg{\Gamma}{\Delta}{}{\mathcal P}}
    {\DerPos{\Gamma}{A_i}{}{\mathcal P}
     \qquad {\DerNeg{\Gamma}{\non A_i,\Delta}{}{\mathcal P}}}
   $$

\item ($\EX$) 
  \[ 
    \infer[\cut_3]{
      \DerNeg{\Gamma}{\Delta}{}{\mathcal P}
    }{
      \infer{\DerPos{\Gamma}{\EX x A}{}{\mathcal P}}{
        \DerPos{\Gamma}{\subst A x t}{}{\mathcal P}
      } 
      \qquad 
      \infer[x \notin \FV{\Gamma,\Delta,\mathcal P}]{
        \DerNeg{\Gamma}{(\FA x \non A),\Delta}{}{\mathcal P}
      }{\DerNeg{\Gamma}{\non A,\Delta}{}{\mathcal P}}} 
   \]
  reduces to \[
    \infer[\cut_3]{\DerNeg{\Gamma}{\Delta}{}{\mathcal P}}
    {{\DerPos{\Gamma}{\subst A x t}{}{\mathcal P}} 
    \qquad 
    \infer
    {\DerNeg{\Gamma}{(\subst {\non A} x t ),\Delta}{}{\mathcal P}}
    {\DerNeg{\Gamma}{\non A,\Delta}{}{\mathcal P}}}
    \]
    using Lemma~\ref{Ladm} (admissibility of instantiation) with $n=0$, noticing that
 $x\notin\FV{\mathcal P}$ and that $\DerNeg{\Gamma}{(\subst {\non A} x t ),\Delta}{}{\mathcal P}$ is safe (since ${\DerNeg{\Gamma}{\Delta}{}{\mathcal P}}$ is safe).\footnote{Using $\alpha$-conversion, we can also pick $x$ such that $x\notin\FV t$.}

\item ($\top^+$)
  $$
  \infer[\cut_3]{\DerNeg{\Gamma}{\Delta}{}{\mathcal P}}
  {\infer{\DerPos{\Gamma}{\top^+}{}{\mathcal P}} {} \qquad 
  \infer{\DerNeg{\Gamma}{\bot^-,\Delta}{}{\mathcal P}}
  {\DerNeg{\Gamma }{\Delta}{}{\mathcal P}}} 
  \quad
  \mbox{reduces to} \qquad \DerNeg{\Gamma}{\Delta}{}{\mathcal P}
  $$

 \item ($\Init[1]$)
  \[
  \infer[\cut_3]{\DerNeg{\Gamma}{\Delta}{}{\mathcal P}}{
    \infer{\DerPos{\Gamma}{p}{}{\mathcal P}}{
      \atmCtxtP{\mathcal P}\Gam,\non p\models_{\mathcal T}
    }
    \qquad 
    \infer{\DerNeg{\Gamma}{(\non p),\Delta}{}{\mathcal P}}{
      \DerNeg{\Gamma, p }{\Delta}{}{\mathcal P}
    }
  } 
  \quad\mbox{reduces to} \qquad 
  \infer[\cut_1]{\DerNeg{\Gamma}{\Delta}{}{\mathcal P}}{
    \atmCtxtP{\mathcal P}\Gam,\non p\models_{\mathcal T}
    \qquad
    \DerNeg{\Gamma, p}{\Delta}{}{\mathcal P}
  }
  \]
  with $p\in\mathcal P$.
 \item ($\Release$)
  $$
  \infer[\cut_3]{\DerNeg{\Gamma}{\Delta}{}{\mathcal P}}
  {\infer{\DerPos{\Gamma}{N}{}{\mathcal P}} 
  {\DerNeg{\Gamma}{N}{}{\mathcal P}} 
  	\qquad \infer{\DerNeg{\Gamma}{(\non N),\Delta}{}{\mathcal P}} 
	{\DerNeg{\Gamma, N }{\Delta}{}{ \polar N }  } } 
	\quad
  \mbox{reduces to} 
  \qquad \infer[\cut_4]{\DerNeg{\Gamma}{\Delta}{}{\mathcal P}}
  {{\DerNeg{\Gamma}{N}{}{\mathcal P}} 
  \qquad {\DerNeg{\Gamma, N}{\Delta} {}{\polar N}}}
  $$
 where $N$ is not $\mathcal P$-positive. We will describe below how $\cut_4$  is reduced.
  
\end{itemize}


  In order to reduce $\cut_4$, we analyse which rule is used to prove the right premiss.
\begin{itemize}
\item ($\andN$)  
  $$
  \infer[\cut_4]{\DerNeg {\Gam}{B\andN C,\Delta}{}{\mathcal P}}
    {
      \DerNeg {\Gam} {N}{}{\mathcal P}
      \quad 
      \infer{
        \DerNeg {\Gam,N}{B\andN C,\Delta}{}{\polar N}
      }
      {
        \DerNeg {\Gam,N}{B,\Delta}{}{\polar N}
        \quad \DerNeg {\Gam,N}{C,\Delta}{}{\polar N}
      }
    } $$
    reduces to
   $$ \infer{
      \DerNeg {\Gam}{B\andN C,\Delta}{}{\mathcal P}
    }
    {
      \infer[\cut_4]{\DerNeg {\Gam}{B,\Delta}{}{\mathcal P}}
      {\DerNeg {\Gam} {N}{}{\mathcal P} 
      \quad \DerNeg {\Gam,N}{B,\Delta}{}{\polar N}}
      \quad
      \infer[\cut_4]{\DerNeg {\Gam}{C,\Delta}{}{\mathcal P}}
      {\DerNeg {\Gam} {N}{}{\mathcal P} \quad 
      \DerNeg {\Gam,N}{C,\Delta}{}{\polar N}}
    }
  $$
 
 \item ($\orN$)
  $$\begin{array}{c}
    \infer[\cut_4]{\DerNeg {\Gam}{B\orN C,\Delta}{}{\mathcal P}}
    {
      \DerNeg {\Gam} {N}{}{\mathcal P}
      \quad 
      \infer{
        \DerNeg {\Gam,N}{B\orN C,\Delta}{}{\polar N}
      }
      {
        \DerNeg {\Gam,N}{B,C,\Delta}{}{\polar N}
      }
    }
    \qquad\mbox{reduces to} \qquad
    \infer{
      \DerNeg {\Gam}{B\orN C,\Delta}{}{\mathcal P}
    }
    {
      \infer[\cut_4]{\DerNeg {\Gam}{B,C,\Delta}{}{\mathcal P}}
      {\DerNeg {\Gam} {N}{}{\mathcal P} \quad 
      \DerNeg {\Gam,N}{B,C,\Delta}{}{\polar N}}
    }
  \end{array}
  $$

 \item ($\FA$)
  $$ \begin{array}{c}
    \infer[\cut_4]{\DerNeg {\Gam}{\forall x B,\Delta}{}{\mathcal P}}
    {
      \DerNeg {\Gam} {N}{}{\mathcal P}
      \quad 
      \infer{
        \DerNeg {\Gam,N}{\forall x B,\Delta}{}{\polar N}
      }
      {
        \DerNeg {\Gam,N}{B,\Delta}{}{\polar N}
      }
    }
    \qquad\mbox{reduces to} \qquad
    \infer{
      \DerNeg {\Gam}{\forall x B,\Delta}{}{\mathcal P}
    }
    {
      \infer[\cut_4]{\DerNeg {\Gam}{B,\Delta}{}{\mathcal P}}
      {\DerNeg {\Gam} {N}{}{\mathcal P} \quad 
      \DerNeg {\Gam,N}{B,\Delta}{}{\polar N}}
    }
  \end{array}
  $$

\item ($\bot^-$)
  \[
  \infer[\cut_4]{\DerNeg {\Gam}{\bot^-,\Delta}{}{\mathcal P}}
    {
      \DerNeg {\Gam} {N}{}{\mathcal P}
      \quad 
      \infer{
        \DerNeg {\Gam,N}{\bot^-,\Delta}{}{\polar N}
      }
      {
        \DerNeg {\Gam,N}{\Delta}{}{\polar N}
      }
    }
    \qquad\mbox{reduces to} \qquad
    \infer{
      \DerNeg {\Gam}{\bot^-,\Delta}{}{\mathcal P}
    }
    {
      \infer[\cut_4]{\DerNeg {\Gam}{\Delta}{}{\mathcal P}}
      {\DerNeg {\Gam} {N}{}{\mathcal P} \quad 
      \DerNeg {\Gam,N}{\Delta}{}{\polar N}}
    }
  \]

 \item ($\Store$)
   \[
   \infer[\cut_4]{\DerNeg {\Gam}{B,\Delta}{}{\mathcal P}}{
     \DerNeg {\Gam} {N}{}{\mathcal P}
     \quad 
     \infer{\DerNeg {\Gam,N}{B,\Delta}{}{\polar N}}{
       \DerNeg {\Gam,N,\non B}{\Delta}{}{\polar[\polar N]{\non B}}
     }
   }
   \qquad\mbox{reduces to} \qquad
   \infer{\DerNeg {\Gam}{B,\Delta}{}{\mathcal P}}{
     \infer[\cut_4]{\DerNeg {\Gam,\non B}{\Delta}{}{\polar {\non B}}}{
       \DerNeg {\Gam,\non B} {N}{}{\polar {\non B}}
       \quad 
       \DerNeg {\Gam,N,\non B}{\Delta}{}{\polar[\polar {\non B}]N}
     }
   }
   \]
   whose left branch is closed by using
   \begin{itemize}
   \item possibly the admissibility of $(\Pol)$ (if $B\in\UP$), so as to get $\DerNeg {\Gam,\non B} {N}{}{\mathcal P}$,
   \item then the admissibility of $(\weak)$ (on $\non B$), to get to the provable premiss $\DerNeg {\Gam} {N}{}{\mathcal P}$;
   \end{itemize}
   whose right branch is the same as the provable $\DerNeg {\Gam,N,\non B}{\Delta}{}{\polar[\polar N]{\non B}}$ unless $B= N\in\UP$, in which case the commutation $\polar[\polar {\non B}]N=\polar[\polar N]{\non B}$ does not hold.
   In this last case, we build
   \[\infer[{[{(\weak[r])}]}]{
     \DerNeg {\Gam}{B,\Delta}{}{\mathcal P}
   }{
     \DerNeg {\Gam}{B}{}{\mathcal P}
   }
   \]

 \item ($\Init[2]$) when $N\not\in\UP$, in which case $\polar N=\mathcal P$ and $\atmCtxtP{\mathcal P}{\Gamma,N}=\atmCtxtP{\mathcal P}{\Gamma}$ (since $N\not\in\mathcal P$ either):
   \[
   \infer[\cut_4]{\DerNeg {\Gam}{}{}{\mathcal P}}{
     \DerNeg \Gam {N}{}{\mathcal P}
     \quad 
     \infer{\DerNeg {\Gam,N}{}{}{\polar N}}{\atmCtxtP{\mathcal P}\Gam\models_{\mathcal T}}
   }
   \qquad\mbox{reduces to} \qquad
   \infer{\DerNeg{\Gamma}{}{}{\mathcal P}} {\atmCtxtP{\mathcal P}\Gam\models_{\mathcal T}}
   \]
 \item ($\Init[2]$) when $N\in\UP$, in which case $\atmCtxtP{\polar N}{\Gam,N}=\atmCtxtP{\mathcal P}{\Gam},N$:
   \[
   \infer[\cut_4]{\DerNeg {\Gam}{}{}{\mathcal P}}{
     \infer{\DerNeg \Gam {N}{}{\mathcal P}}{\DerNeg {\Gam, \non N}{}{}{\mathcal P, {\non N}}}
     \quad
     \infer{\DerNeg {\Gam,N}{}{}{\polar N}}{\atmCtxtP{\mathcal P}\Gam,N\models_{\mathcal T}}
   }
   \qquad\mbox{reduces to} \qquad
   \infer[\cut_1]{\DerNeg {\Gam}{}{}{\mathcal P}}{
     \atmCtxtP{\mathcal P, {\non N}}\Gam,N\models_{\mathcal T}
     \qquad
     \DerNeg {\Gam, \non N}{}{}{\mathcal P,{\non N}}
   }
   \]
  since ${\atmCtxtP{\mathcal P}\Gam,N\models_{\mathcal T}}$ implies ${\atmCtxtP{\mathcal P, {\non N}}\Gam,N\models_{\mathcal T}}$.
 
 \item $(\Select)$ on formula $\non N$
  \[
  \begin{array}{lcl}
    \infer[\cut_4]{\DerNeg {\Gam}{}{}{\mathcal P}}
    {\DerNeg \Gam {N}{}{\mathcal P}\quad \infer{\DerNeg {\Gam,N}{}{}{\polar N}}{\DerPos {\Gam,N}{\non N}{}{\polar N}}}
    &
    \qquad\mbox{reduces to} \qquad
    &\infer[\cut_3]{\DerNeg {\Gam}{}{}{\mathcal P}}
    {\infer[\cut_5]{\DerPos {\Gam}{\non N}{}{\mathcal P}}
    {\DerNeg \Gam {N}{}{\mathcal P}\quad \DerPos {\Gam,N}{\non N}{}{\polar N}} \qquad \DerNeg \Gam {N}{}{\mathcal P}}
    \\\\
    &\mbox{ or to }&
    \infer[\cut_5]{\DerNeg{\Gamma}{}{}{\mathcal P}}
    {\DerNeg{\Gamma}{N}{}{\mathcal P} \quad \DerPos{\Gamma,N}{\non N}{}{\polar N}
    }
  \end{array}
  \]
  depending on the outcome of $\cut_5$

\item $(\Select)$ on a formula $P$ that is not $\polar N$-negative

  $$\begin{array}{lcl}
    \infer[\cut_4]{\DerNeg {\Gam,\non P}{}{}{\mathcal P}}
    {\DerNeg {\Gam,\non P} {N}{}{\mathcal P}
    \quad \infer{\DerNeg {\Gam,\non P,N}{}{}{\polar N}}
    {\DerPos {\Gam,\non P,N}{P}{}{\polar N}}}
    &\qquad\mbox{reduces to} \qquad&
    \infer{\DerNeg {\Gam,\non P}{}{}{\mathcal P}}
    {\infer[\cut_5]{\DerPos {\Gam,\non P}{P}{}{\mathcal P}}
    {\DerNeg {\Gam,\non P} {N}{}{\mathcal P} 
     \quad \DerPos {\Gam,\non P,N}{P}{}{\polar N}}}
     \\\\
    &\mbox{ or to }&
    \infer[\cut_5]{\DerNeg{\Gamma,\non P}{}{}{\mathcal P}}
    {\DerNeg{\Gamma,\non P}{N}{}{\mathcal P} \quad 
    \DerPos{\Gamma,\non P, N}{\non N}{}{\polar N}
    }
  \end{array}
  $$
  depending on the outcome of $\cut_5$

 \end{itemize} 

  We have reduced all cases of $\cut_4$; we now reduce the cases for $\cut_5$ (again, by case analysis on the last rule used to prove the right premiss).

\begin{itemize}
\item ($\andP$) We are given
  \[
    \DerNeg {\Gam} {N}{}{\mathcal P}
    \quad\mbox{ and }\quad
    \infer{\DerPos {\Gam,N}{B_1\andP B_2}{}{\polar N}}{
      \DerPos {\Gam,N}{B_1}{}{\polar N} \quad
      \DerPos {\Gam,N}{B_2}{}{\polar N}
  }
  \]
and by $\cut_5$ we want to derive either $\DerPos {\Gam}{B_1\andP B_2}{}{\mathcal P}$ or $\DerNeg {\Gam}{}{}{\mathcal P}$.

If we can, we build   
  \[
    \infer{
      \DerPos {\Gam}{B_1\andP B_2}{}{\mathcal P}
    }
    {
      \infer[\cut_5]{\DerPos {\Gam}{B_1}{}{\mathcal P}}
      {\DerNeg {\Gam} {N}{}{\mathcal P} \quad 
      \DerPos {\Gam,N}{B_1}{}{\polar N}}
      \quad
      \infer[\cut_5]{\DerPos {\Gam}{B_2}{}{\mathcal P}}
      {\DerNeg {\Gam} {N}{}{\mathcal P} \quad 
      \DerPos {\Gam,N}{B_2}{}{\polar N}}
    }
  \]
Otherwise we build
\[ \infer[\cut_5]{\DerNeg{\Gamma}{}{}{\mathcal P}}
		{ \DerNeg{\Gamma}{N}{}{\mathcal P} \quad 
		\DerPos{\Gamma,N}{B_i}{}{ \polar N}
		}		
\]
where $i$ is (one of) the premiss(es) for which $\cut_5$ produces a proof of $\DerNeg{\Gamma}{}{}{\mathcal P}$.

 \item ($\orP$) We are given
  \[
      \DerNeg {\Gam} {N}{}{\mathcal P}
    \quad\mbox{ and }\quad
      \infer{
        \DerPos {\Gam,N}{B_1\orP B_2}{}{\polar N}
      }
      { \DerPos {\Gam,N}{B_i} {}{\polar N}}
\]
and by $\cut_5$ we want to derive either $\DerPos {\Gam}{B_1\orP B_2}{}{\mathcal P}$ or $\DerNeg {\Gam}{}{}{\mathcal P}$.

If we can, we build   
\[
\infer{ \DerPos {\Gam}{B_1\orP B_2}{}{\mathcal P} }
      {\infer[\cut_5] { \DerPos {\Gam}{B_i}{}{\mathcal P} }
        {\DerNeg {\Gam} {N}{}{\mathcal P} \quad 
          \DerPos {\Gam,N}{B_i}{}{\polar N}}
      }
      \]
Otherwise we build
$$ \infer{\DerNeg{\Gamma}{}{}{\mathcal P}}
		{ \DerNeg{\Gamma}{N}{}{\mathcal P} \quad 
		\DerPos{\Gamma,N}{B_i}{}{ \polar N}
		}		
$$
  
 \item ($\EX$) We are given
  \[
  \DerNeg {\Gam} {N}{}{\mathcal P}
  \quad\mbox{ and }\quad
  \infer{
    \DerPos {\Gam,N}{\EX x B}{}{\polar N}
  }
        { \DerPos {\Gam,N}{\subst B x t}{}{\polar N}}
\]
and by $\cut_5$ we want to derive either $\DerPos {\Gam}{\EX x B}{}{\mathcal P}$ or $\DerNeg {\Gam}{}{}{\mathcal P}$.

If we can, we build   
\[\infer { \DerPos {\Gam}{\EX x B}{}{\mathcal P} }
    {   \infer[\cut_5] { \DerPos {\Gam}{\subst B x t}{}{\mathcal P} }
      {\DerNeg {\Gam} {N}{}{\mathcal P}\quad 
      \DerPos {\Gam,N}{\subst B x t}{}{\polar N}}
      } 
\]
Otherwise we build
$$ \infer{\DerNeg{\Gamma}{}{}{\mathcal P}}
		{ \DerNeg{\Gamma}{N}{}{\mathcal P} \quad 
		\DerPos{\Gamma,N}{\subst B x t}{}{ \polar N}
		}		
$$

 \item ($\top^+$) We are given
  \[
      \DerNeg {\Gam} {N}{}{\mathcal P}
      \quad\mbox{ and }\quad
      \infer{\DerPos{\Gamma,N}{\top^+}{}{\polar N}} {}
\]
and by $\cut_5$ we want to derive either $\DerPos {\Gam}{\top^+}{}{\mathcal P}$ or $\DerNeg {\Gam}{}{}{\mathcal P}$.

We build
\[
    \infer { \DerPos {\Gam}{\top^+}{}{\mathcal P}} {}
\]

 \item ($\Release$)
  We are given:
\[
\DerNeg {\Gam} {N}{}{\mathcal P}
\qquad\mbox{ and }\qquad 
\infer{
  \DerPos {\Gam,N}{N'}{}{ \polar N }
}
      { \DerNeg {\Gam,N}{N'}{}{ \polar N } }
        \]
 where $N'$ is not $\polar N$-positive;\\
and by $\cut_5$ we want to derive either $\DerPos {\Gam}{N'}{}{\mathcal P}$ or $\DerNeg {\Gam}{}{}{\mathcal P}$.

We build
 \[ \infer{\DerPos {\Gam}{N'}{}{\mathcal P}}
    {
      \infer[\cut_4]{\DerNeg {\Gam}{N'}{}{\mathcal P}}
      {
        \DerNeg {\Gam}{N}{}{\mathcal P}
        \quad
        \DerNeg {\Gam,N}{N'}{}{\polar N}
      }
    }
    \]
    since $N'$ is not $\mathcal P$-positive.

\item ($\Init[1]$)
  We are given:
  \[\DerNeg {\Gam} {N}{}{\mathcal P}
  \qquad\mbox{ and }\qquad 
  \infer{\DerPos{\Gamma,N}{p}{}{ \polar N }}{
    \atmCtxtP{\polar N}{\Gam,N},\non p\models_{\mathcal T}      
  }
  \]
with $p\in\polar N$,\\
and by $\cut_5$ we want to derive either $\DerPos {\Gam}{p}{}{\mathcal P}$ or $\DerNeg {\Gam}{}{}{\mathcal P}$.
 
 If $N$ is $\mathcal P$-negative \\
	then $\polar N=\mathcal P$ and $p$ is $\mathcal P$-positive. So
 $\atmCtxtP{\polar N}{\Gamma, N},\non p= \atmCtxtP{\mathcal P}{\Gamma},\non p$
 and we build
 $$
 \infer[{[({\Init[1]})]}]{\DerPos{\Gamma}{p}{}{\mathcal P}} {}
 $$
 
 If $N\in \UP$ $(\atmCtxtP{\polar N} {\Gamma,N},\non p=\atmCtxtP{\mathcal P} 
 {\Gamma},N,\non p)$ 
  \begin{itemize}
   \item if $p=N$ then we build 
   $$ \infer{\DerPos{\Gamma}{N}{}{\mathcal P}}
   	 {\DerNeg{\Gamma}{N}{}{\mathcal P}}
   $$
         as $N$ is not $\mathcal P$-positive;
   \item if $p\neq N$ then $p$ is $\mathcal P$-positive 
     \begin{enumerate}
     \item if $\atmCtxtP{\mathcal P}{\Gamma},N\models_{\mathcal T}$\\
     then applying invertibility of ($\Store[=]$) on $\DerNeg{\Gamma}{N}{}{\mathcal P}$ gives $\DerNeg{\Gamma,\non N}{}{}{\mathcal P}$ and we build:
	$$  \infer[\cut_1]{\DerNeg{\Gamma}{}{}{\mathcal P}}
			{ \atmCtxtP{\mathcal P}{\Gamma}, N\models_{\mathcal T}
			\quad {\DerNeg{\Gamma,\non N}{}{}{\mathcal P}}
			}
	$$

     \item if $\atmCtxtP{\mathcal P}{\Gamma},N \not\models_{\mathcal T}$\\
     then $\mathcal R \eqdef \atmCtxtP{\mathcal P}{\Gamma},N$ is a set of literals
     satisfying $\atmCtxtP{\mathcal P}{\Gamma}\subseteq \mathcal R \subseteq \atmCtxtP{\mathcal P}{\Gamma} \cup \UP$ (since $N\in \UP$) and $\mathcal R,\non p \models_{\mathcal T} $. 
     
     Hence we get $\atmCtxtP{\mathcal P}{\Gamma},\non p\models_{\mathcal T}$ as well, since ($\Gamma, \mathcal P$) is assumed to be safe.

     We can finally build
\[  \infer[{[{(\Init[1])}]}]{\DerPos{\Gamma}{p}{}{ \mathcal P }}{
    \atmCtxtP{\mathcal P}{\Gam},\non p\models_{\mathcal T}      
  }
  \]

	\end{enumerate}
  \end{itemize}
 
\end{itemize}
\end{proof}

\begin{theorem}[$\cut_6$, $\cut_7$, and $\cut_8$]
\label{Thcut}
  The following rules are admissible in \LKTh. 
  \[
  \begin{array}{c}
    \infer[\cut_6]{\DerNeg {\Gam}{\Del}{}{\mathcal P}}
    {\DerNeg \Gam {N,\Del}{}{\mathcal P} \quad 
    \DerNeg {\Gam,N}{\Del}{}{ \polar N }}
    \qquad
    \infer[\cut_7] {\DerNeg {\Gamma}{\Del}{}{\mathcal P}}
    {\DerNeg{\Gamma}{A,\Del}{}{\mathcal P} \quad 
    \DerNeg{\Gamma}{\non A,\Del}{}{\mathcal P}}  
    \qquad
    \infer[\cut_8] {\DerNeg {\Gamma}{\Del}{}{\mathcal P}}
    {\DerNeg{\Gamma,l}{\Del}{}{\polar l} \quad 
    \DerNeg{\Gamma,\non l}{\Del}{}{\polar {\non l}}}  
  \end{array}
  \]
\end{theorem}
\begin{proof}
$\cut_6$ is proved admissible by induction on the multiset $\Del$: the base case is the admissibility of $\cut_4$, and the other cases just require the inversion of the connectives in $\Del$ (using $(\Store[=])$ instead of $(\Store)$, to avoid modifying the polarisation set).

For $\cut_7$, we can assume without loss of generality (swapping $A$ and $\non A$) that $A$ is not $\mathcal P$-positive.
Applying inversion on $\DerNeg{\Gamma}{\non A,\Del}{}{\mathcal P}$ gives a proof of $\DerNeg{\Gamma,A}{\Del}{}{\polar { A}}$, and 
$\cut_7$ is then obtained by $\cut_6$:
\[    \infer[\cut_6] {\DerNeg {\Gamma}{\Del}{}{\mathcal P}}
    {\DerNeg{\Gamma}{A,\Del}{}{\mathcal P} \quad 
      \DerNeg{\Gamma,A}{\Del}{}{ \polar A}
    }
\]

$\cut_8$ is obtained as follows:
\[    \infer[\cut_7] {\DerNeg {\Gamma}{\Del}{}{\mathcal P}}
    {\infer
      {\DerNeg{\Gamma}{l,\Del}{}{\mathcal P}}
      {\DerNeg{\Gamma,\non l}{\Del}{}{\polar {\non l}}}
      \quad 
      \infer
      {\DerNeg{\Gamma}{\non l,\Del}{}{\mathcal P}}
      {\DerNeg{\Gamma,l}{\Del}{}{\polar l}}
    }
\]


\end{proof}

\section{Changing the polarity of connectives}

In this section, we show that changing the polarity of connectives does not change provability in \LKThp. To prove this property of the \LKThp\ system, we genealise it into a new system \LKThEx.

\begin{definition}[\LKThEx]
  The sequent calculus \LKThEx\ manipulates one kind of sequent:
  \[
  \DerOSEx {\Gamma} {\mathcal X}{\Del}
  \qquad
  \mbox{ where } \mathcal X \recdef \bullet \sep A
  \]

  Here, $\mathcal P$ is a polarisation set, $\Gam$ is a multiset of literals and $\mathcal P$-negative formulae, $\Delta$ is a multiset of formulae, and $\mathcal X$ is said to be in the \emph{focus} of the sequent.

  The rules of \LKThEx, given in Figure~\ref{fig:LKThEx}, are again of three kinds: synchronous rules, asynchronous rules, and structural rules.
\end{definition}

\begin{figure}[h]
  \[
  \begin{array}{|c|}
    \upline
    \mbox{\textsf{Synchronous rules}}
    \hfill\strut\\
    \infer[{[(\andP)]}]{\DerOSEx{\Gamma}{A\andP B}{\Del}}
    {\DerOSEx{\Gamma}{A}{\Del} \qquad \DerOSEx{\Gamma}{B}{\Del}}
    \qquad
    \infer[{[(\orP)]}]{\DerOSEx{\Gamma}{A_1\orP A_2}{\Del}}
    {\DerOSEx{\Gamma}{A_i}{\Del}}
    \qquad
    \infer[{[(\EX)]}]{\DerOSEx{\Gamma}{\EX x A}{\Del}}
    {\DerOSEx{\Gamma}{\subst A x t}{\Del}}
    \\\\
    \infer[{[(\top^+)]}]{
      \DerOSEx{\Gamma}{\top^+}{\Del}
    }{\strut}
    \qquad
    \infer[{[({\Init[1]})]l\mbox{ is $\mathcal P$-positive}}]{
      \DerOSEx{\Gamma}{l}{\Del}
    }{
      \atmCtxt\Gam,\non l,\atmCtxt[\mathcal L]{\non \Del}\models_{\mathcal T}
    }
    \qquad
    \infer[{[(\Release)] N\mbox{ not $\mathcal P$-positive}}]{\DerOSEx {\Gam} {N}{} }{\DerOSEx {\Gam} {\bullet} {N} }
    \\\midline
    \mbox{\textsf{Asynchronous rules}}\hfill\strut\\
    \infer[{[(\andN)]}]{\DerOSEx{\Gamma}{\mathcal X}{A\andN B,\Delta}}
    {\DerOSEx{\Gamma}{\mathcal X}{A,\Delta} \qquad \DerOSEx
      {\Gamma}{\mathcal X}{B,\Delta}}
    \quad
    \infer[{[(\orN)]}]{\DerOSEx {\Gamma}{\mathcal X} {A_1\orN A_2,\Delta} }
    {\DerOSEx {\Gamma}{\mathcal X} {A_1,A_2,\Delta} }
    \quad
    \infer[{[(\FA)] x\notin\FV{\Gam,\mathcal X,\Delta,\mathcal P} }]{
      \DerOSEx{\Gamma}{\mathcal X}{(\FA x A),\Delta}
    }{
      \DerOSEx {\Gamma}{\mathcal X} {A,\Delta}
    }\\\\
    \infer[{[(\bot^-)]}]{
      \DerOSEx {\Gam} {\mathcal X} {\bot^-,\Delta}
    }{
      \DerOSEx {\Gam} {\mathcal X}{\Delta}
    }
    \qquad
    \infer[{[(\top^-)]}]{
      \DerOSEx {\Gam} {\mathcal X} {\top^-}
    }{\strut}
    \qquad
    \infer[{[(\Store)] A\mbox{ literal or $\mathcal P$-positive}}]{
      \DerOSEx {\Gam} {\mathcal X} {A,\Delta}
    }{
      \DerOSEx {\Gam,\non A} {\mathcal X}{\Delta}[\polar{\non A}]
    }\\
    \midline
    \mbox{\textsf{Structural rules}}\hfill\strut\\
    \infer[{[(\Select)]\begin{array}l P \mbox{ not $\mathcal P$-negative}
    \end{array}}]{\DerOSEx {\Gam,\non P}{\bullet} {\Del}} {\DerOSEx {\Gam,\non P} {P} {\Del} }
    \qquad
    \infer[{[({\Init[2]})]}]{\DerOSEx {\Gam}{\bullet} {\Del}}
    {\atmCtxt\Gam,\atmCtxt[\mathcal L]{\non\Del}\models_\mathcal{T}}
    \downline
  \end{array}
  \]
  \caption{System \LKThEx}
  \label{fig:LKThEx}
\end{figure}

\begin{remark}
  The \LKThEx\ system is an extension system of \LKThp: the \LKThp\ system is the fragment of \LKThEx\ where every sequent $\DerOSEx {\Gam,\non P}{\bullet} {\Del}$ is requested to have either $\mathcal X=\bullet$ or $\Delta$ is empty. In terms of bottom-up proof-search, this only restricts the structural rules to the case where $\Delta$ is empty.

  As in \LKThp, (left-)weakening and (left-)contraction are height-preserving admissible in \LKThEx. 
\end{remark}

We can now prove a new version of identity:
\begin{lemma}[Identities]
  \label{con1}
For all $\mathcal P$, $A$, $\Delta$, the sequent $\DerOSEx{} {\non A} {A,\Delta}$ is provable in \LKThEx.
\end{lemma}

\begin{proof}
  By induction on $A$ using an extended but well-founded order on formulae:\\
  a formula is smaller than another one when
  \begin{itemize}
  \item either it contains fewer connectives
  \item or the number of connectives is equal, neither formulae are literals, and the former formula is negative and the latter is positive.
  \end{itemize}
  We now treat all possible shapes for the formula $A$:
  \begin{itemize}
  \item $A= A_1\andN A_2$
    \[
    \infer{
      \DerOSEx{}{\non{A_1} \orP \non{A_2}}{A_1\andN A_2,\Delta} 
    }{
      \infer{
        \DerOSEx {} {\non{A_1} \orP \non{A_2} } {A_1,\Delta} 
      }{\DerOSEx {} {\non{A_1}} {A_1,\Delta} }
      \quad
      \infer{
        \DerOSEx{} {\non{A_1} \orP \non{A_2} } {A_2,\Delta} 
      }{\DerOSEx{}{\non{A_2}}{A_2,\Delta}}
    }
    \]

    We can complete the proof on the left-hand side by applying the induction hypothesis on $A_1$ and 
    on the right-hand side by applying the induction hypothesis on $A_2$.

  \item $A= A_1\orN A_2$
    \[
    \infer{
      \DerOSEx{}{\non{A_1} \andP \non{A_2}}{A_1\orN A_2,\Delta}
    }{
      \infer{
        \DerOSEx{}{\non{A_1} \andP \non{A_2} }{A_1,A_2,\Delta}
      }{
        \DerOSEx {} {\non{A_1}} {A_1,A_2,\Delta}
        \quad
        \DerOSEx{} {\non{A_2}}  {A_1,A_2,\Delta}
      }
    }
    \]

    We can complete the proof on the left-hand side by applying the induction hypothesis on $A_1$ and 
    on the right-hand side by applying the induction hypothesis on $A_2$.

  \item $A= \FA x A$
    \[
    \infer[x \notin \FV{\EX x \non A,\Delta}]{
      \DerOSEx{}{\EX x \non A}{\FA x A,\Delta} 
    }{
      \infer{\DerOSEx{}{\EX x \non A}{A,\Delta}}{
        \iinfer[\mbox{choosing t=x}]{\DerOSEx {} { \{t/x\}\non A } {A,\Delta} }{
          \DerOSEx {} {\non A} {A,\Delta} 
        }
      }
    }
    \]

    We can complete the proof by applying the induction hypothesis on $A$.

  \item $A=\bot^-$ 
    \[
    \infer[\top^+]{\DerOSEx{}{\top^+}{\bot-,\Delta}}{}
    \]

  \item $A=\non p$, with $p$ not being $\mathcal P$-negative:
    \[
    \infer{\DerOSEx{}{p}{\non p,\Delta}}{
      \infer{\DerOSEx{p}{p}{\Delta}[\polar{p}]}{}
    }
    \]
    as $p$ is then $\polar p$-positive.

  \item $A=P$ where $P$ is $\mathcal P$-positive:
    \[
    \infer{\DerOSEx{}{\non P}{P,\Delta}}{
      \infer {\DerOSEx {\non P} {\non P} {\Delta}}{
        \infer {\DerOSEx {\non P,\non \Delta}{\non P}{}[\mathcal P']}{
          \infer {\DerOSEx {\non P}{\non P}{}[\mathcal P']}{
            \infer{\DerOSEx{\non P}{\bullet}{\non P}[\mathcal P']}{
              \infer{\DerOSEx{\non P}{P}{\non P}[\mathcal P']}{
                \DerOSEx{}{P}{\non P } 
              }
            }
          }
        }
      }
    }
    \]
    If $P$ is a literal, we complete the proof with the case just above.
    If it is not a literal, then $P$ is smaller than $\non P$ and we complete the proof by applying the induction hypothesis on $P$.
  \end{itemize}
\end{proof}

We now want to show that all asynchronous rules are invertible in \LKThEx. We first start with the following lemma:
\begin{lemma}[Generalised $(\Init)$ and negative $\Select$]\label{lem:weakselect}\strut

  The following rules are height-preserving admissible in \LKThEx:
  \[
    \infer[{[({\Init})]}]{\DerOSEx {\Gam}{\mathcal X} {\Del}}
    {\atmCtxt\Gam,\atmCtxt[\mathcal L]{\non\Del}\models_\mathcal{T}}
    \qquad
  \infer[{[{(\Select[-])}]}]{\DerOSEx{\Gamma}{\bullet}{\Delta}[\polar {\non l}]}{\DerOSEx{\Gamma}{l}{\Delta}[\polar {\non l}]}
 \]
  where $\non l\in\Gamma$ and it is not $\mathcal P$-negative in $(\Select[-])$.
\end{lemma}
\begin{proof}
  For each rule, by induction on the proof of the premiss.

  For $({\Init})$:
  \begin{itemize}
  \item if it is obtained by $(\andN),(\orN),(\forall),(\bot^-)$, we can straightforwardly use the induction hypothesis on the premiss(es), and if it is $(\top^-)$ it is trivial;
  \item if it is obtained by 
    \[
    \infer{\DerOSEx{\Gamma}{\mathcal X}{A,\Delta'}}{\DerOSEx{\Gamma,\non A}{\mathcal X}{\Delta'}[{\polar{\non A}}]}
    \]
    then we can use the induction hypothesis on the premiss as $\atmCtxt[\polar{\non A}]{\Gam,\non A},\atmCtxt[\mathcal L]{\non{\Del'}}=
\atmCtxt{\Gam},\atmCtxt[\mathcal L]{\non A,\non{\Del'}}$;
  \item the last possible way to obtain it is with $\Delta=\emptyset$ and
    \[
    \infer{\DerOSEx{\Gamma}{N}{}}{
      \DerOSEx{\Gamma}{\bullet}{N}
    }
    \]
    for some $N$ that is not $\mathcal P$-positive, 
    and we conclude with $(\Init[2])$.
  \end{itemize}

  For $(\Select[-])$, first notice that $l$ is $\polar{\non l}$-negative, and then:
  \begin{itemize}
  \item if again it is obtained by $(\andN),(\orN),(\forall),(\bot^-)$, we can straightforwardly use the induction hypothesis on the premiss(es), and if it is $(\top^-)$ it is trivial;
  \item if it is obtained by 
    \[
    \infer{\DerOSEx{\Gamma}{l}{A,\Delta'}[\polar{\non l}]}{\DerOSEx{\Gamma,\non A}{l}{\Delta'}[{\polar[\polar{\non l}]{\non A}}]}
    \]
    then we can use the induction hypothesis on the premiss, if $A$ is not $\non l$ (so that $\polar[\polar{\non l}]{\non A}=\polar[\polar{\non A}]{\non l}$ and $\non l$ is not $\polar{\non A}$-negative);
    if $A=\non l$, then we build
    \[
    \infer[{[({\Init[2]})]}]{\DerOSEx{\Gamma}{\bullet}{A,\Delta'}[\polar {\non l}]}{
      \atmCtxt[\polar{\non l}]\Gam,{\atmCtxt[\mathcal L]{\non A,\non{\Del'}}}\models_{\mathcal T}
    }
    \]
    as $A\in\atmCtxt[\polar{\non l}]\Gam$.    
  \item the last possible way to obtain it is with $\Delta=\emptyset$ and
    \[
    \infer{\DerOSEx{\Gamma}{l}{}[\polar {\non l}]}{
      \infer{ \DerOSEx{\Gamma}{\bullet}{l}[\polar {\non l}]}{
        \DerOSEx{\Gamma,\non l}{\bullet}{}[\polar {\non l}]
      }
    }
    \]
    and we conclude with the height-preserving admissibility of contraction.
  \end{itemize}

\end{proof}

We can now state and prove the invertibility of asynchronous rules:

\begin{lemma}[Invertibility of asynchronous rules]\strut

  All asynchronous rules are height-preserving invertible in \LKThEx.
\end{lemma}
\begin{proof}
  By induction on the derivation proving the conclusion of the asynchronous rule considered.
  \begin{itemize}
  \item Inversion of $A \andN B$: by case analysis on the last rule actually used
    \begin{itemize}
    \item $\infer{\DerOSEx{\Gamma}{\mathcal X}{A\andN B,C\andN D,\Delta}} {\DerOSEx{\Gamma}{\mathcal X}{A\andN B,C,\Delta} \quad \DerOSEx{\Gamma}{\mathcal X}{A\andN B,D,\Delta}}$\\

      By induction hypothesis we get

      \hfill
      $\infer{\DerOSEx{\Gamma}{\mathcal X}{A,C\andN D,\Delta}} {\DerOSEx{\Gamma}{\mathcal X}{A,C,\Delta} \qquad \DerOSEx{\Gamma}{\mathcal X}{A,D,\Delta}}$
      \hfill and\hfill
      $\infer{\DerOSEx{\Gamma}{\mathcal X}{B,C\andN D,\Delta}} {\DerOSEx{\Gamma}{\mathcal X}{B,C,\Delta} \quad \DerOSEx{\Gamma}{\mathcal X}{B,D,\Delta}}$

    \item
      $\infer{\DerOSEx{\Gamma}{\mathcal X}{A\andN B,C\orN D,\Delta}} {\DerOSEx{\Gamma}{\mathcal X}{A\andN B,C, D,\Delta}}$\\

      By induction hypothesis we get

      \hfill
      $\infer{\DerOSEx{\Gamma}{\mathcal X}{A,C\orN D,\Delta}} {\DerOSEx{\Gamma}{\mathcal X}{A,C, D,\Delta}}$ \hfill and\hfill
      $\infer{\DerOSEx{\Gamma}{\mathcal X}{B,C\orN D,\Delta}} {\DerOSEx{\Gamma}{\mathcal X}{B,C, D,\Delta}}$

    \item
      $\infer[x\notin\FV{\Gam,\mathcal X,\Delta,A\andN B}]{\DerOSEx{\Gamma}{\mathcal X}{A\andN B,(\FA x C),\Delta}} {\DerOSEx{\Gamma}{\mathcal X}{A\andN B,C,\Delta}}$ \\

      By induction hypothesis we get

      $\infer[x\notin\FV{\Gam,\mathcal X,\Delta,A}]{\DerOSEx{\Gamma}{\mathcal X}{A,(\FA x C),\Delta}} {\DerOSEx{\Gamma}{\mathcal X}{A,C,\Delta}}$ and 
      $\infer[x\notin\FV{\Gam,\mathcal X,\Delta,B}]{\DerOSEx{\Gamma}{\mathcal X}{B,(\FA x C),\Delta}} {\DerOSEx{\Gamma}{\mathcal X}{B,C,\Delta}}$

    \item
      $\infer[\mbox{\begin{tabular}l$C$ literal or\\ $\mathcal P$-positive\end{tabular}}]
      {\DerOSEx{\Gamma}{\mathcal X}{A\andN B,C,\Delta}} {\DerOSEx{\Gamma,\non C}{\mathcal X}{A\andN B,\Delta}[\polar{\non C}]}$ \\ 

      By induction hypothesis we get

      \hfill
      $\infer[\mbox{\begin{tabular}l$C$ literal or\\ $\mathcal P$-positive\end{tabular}}]{
        \DerOSEx{\Gamma}{\mathcal X}{A,C,\Delta}[\polar{\non C}]
      }{\DerOSEx{\Gamma,\non C}{\mathcal X}{A,\Delta}}$\hfill   and\hfill
      $\infer[\mbox{\begin{tabular}l$C$ literal or\\ $\mathcal P$-positive\end{tabular}}]{\DerOSEx{\Gamma}{\mathcal X}{ B,C,\Delta}} {\DerOSEx{\Gamma,\non C}{\mathcal X}{B,\Delta}[\polar{\non C}]} $ 

    \item
      $\infer{\DerOSEx{\Gamma}{\mathcal X}{A\andN B,\bot^-,\Delta}} {\DerOSEx{\Gamma}{\mathcal X}{A\andN B,\Delta}}$ \\ 

      By induction hypothesis we get

      \hfill
      $\infer{\DerOSEx{\Gamma}{\mathcal X}{A,\bot^-,\Delta}} {\DerOSEx{\Gamma}{\mathcal X}{A,\Delta}}$\hfill and\hfill
      $\infer{\DerOSEx{\Gamma}{\mathcal X}{ B,\bot^-,\Delta}} {\DerOSEx{\Gamma}{\mathcal X}{B,\Delta}} $ 

    \item
      $\infer{\DerOSEx{\Gamma}{\mathcal X}{A\andN B,\top^-,\Delta}} {\strut}$\\

      We get
      
      \hfill
      $\infer{\DerOSEx{\Gamma}{\mathcal X}{A,\top^-,\Delta}} {}$
      \hfill and\hfill$\infer{\DerOSEx{\Gamma}{\mathcal X}{B,\top^-,\Delta}} {}$

    \item $\infer{\DerOSEx{\Gamma}{C\andP D,}{A\andN B,\Delta}} {\DerOSEx{\Gamma}{C}{A\andN B,\Delta} \quad \DerOSEx{\Gamma}{D}{A\andN B,\Delta}}$\\

      By induction hypothesis we get

      \hfill
      $\infer{\DerOSEx{\Gamma}{C\andP D}{A,\Delta}} {\DerOSEx{\Gamma}{C}{A,\Delta} \qquad \DerOSEx{\Gamma}{D}{A,\Delta}}$
      \hfill and\hfill
      $\infer{\DerOSEx{\Gamma}{C\andP D}{B,\Delta}} {\DerOSEx{\Gamma}{C}{B,\Delta} \quad \DerOSEx{\Gamma}{D}{B,\Delta}}$

    \item
      $\infer{\DerOSEx{\Gamma}{C_1\orP C_2}{A\andN B,\Delta}} {\DerOSEx{\Gamma}{C_i}{A\andN B,\Delta}}$\\

      By induction hypothesis we get

      \hfill
      $\infer{\DerOSEx{\Gamma}{C_1\orP C_2}{A,\Delta}} {\DerOSEx{\Gamma}{C_i}{A,\Delta}}$ \hfill and\hfill
      $\infer{\DerOSEx{\Gamma}{C_1\orP C_2}{B,\Delta}} {\DerOSEx{\Gamma}{C_i}{B,\Delta}}$

    \item
      $\infer{\DerOSEx{\Gamma}{\EX x C}{A\andN B,\Delta}} {\DerOSEx{\Gamma}{\subst C x t}{A\andN B,\Delta}} $\\

      By induction hypothesis we get

      \hfill $\infer{\DerOSEx{\Gamma}{\EX x C}{A,\Delta}} {\DerOSEx{\Gamma}{\subst C x t}{A,\Delta}}$ \hfill and\hfill
      $\infer{\DerOSEx{\Gamma}{\EX x C}{B,\Delta}} {\DerOSEx{\Gamma}{\subst C x t}{B,\Delta}}$

      \item 
         $\infer{\DerOSEx{\Gamma}{\top^+}{A\andN B,\Delta}} {\strut}$\\
         
         We get

         \hfill
         $\infer{\DerOSEx{\Gamma}{\top^+}{A,\Delta}} {}$   \hfill and\hfill
         $\infer{\DerOSEx{\Gamma}{\top^+}{ B,\Delta}} {} $ 

    \item
      $\infer[\atmCtxt\Gam,\non p,{\atmCtxt[\mathcal L]{\non\Del}}\models_{\mathcal T}]
	{\DerOSEx{\Gamma}{p}{A\andN B,\Delta}} {\strut}$
        \qquad with $p$ being $\mathcal P$-positive\\

        We get

        \hfill
        $  \infer[\atmCtxt\Gam,\non p,{\atmCtxt[\mathcal L]{\non\Del}}\models_{\mathcal T}]
	{\DerOSEx{\Gamma}{p}{A,\Delta}} {\strut} $
	\hfill and\hfill
        $  \infer[\atmCtxt\Gam,\non p,{\atmCtxt[\mathcal L]{\non\Del}}\models_{\mathcal T}]
	{\DerOSEx{\Gamma}{p}{B,\Delta}} {\strut} $

    \item
      $\infer[\atmCtxt\Gam,{\atmCtxt[\mathcal L]{\non\Del}}\models_{\mathcal T}]
	{\DerOSEx{\Gamma}{\bullet}{A\andN B,\Delta}} {\strut}$\\

        We get

        \hfill
        $  \infer[\atmCtxt\Gam,{\atmCtxt[\mathcal L]{\non\Del}}\models_{\mathcal T}]
	{\DerOSEx{\Gamma}{\bullet}{A,\Delta}} {\strut} $
	\hfill and\hfill
        $  \infer[\atmCtxt\Gam,{\atmCtxt[\mathcal L]{\non\Del}}\models_{\mathcal T}]
	{\DerOSEx{\Gamma}{\bullet}{B,\Delta}} {\strut} $

    \item
      $\infer{\DerOSEx{\Gamma}{\bullet}{A\andN B,\Delta}} {\DerOSEx{\Gamma}{P}{A\andN B,\Delta}}$ where $\non P\in\Gamma$ is not $\mathcal P$-positive\\ 

      By induction hypothesis we get

      \hfill
      $\infer{\DerOSEx{\Gamma}{\bullet}{A,\Delta}} {\DerOSEx{\Gamma}{P}{A,\Delta}}$\hfill and\hfill
      $\infer{\DerOSEx{\Gamma}{\bullet}{ B,\Delta}} {\DerOSEx{\Gamma}{P}{B,\Delta}} $ 

    \end{itemize}

  \item Inversion of $A \orN B$   
    \begin{itemize}
    \item $\infer{\DerOSEx{\Gamma}{\mathcal X}{A\orN B,C\andN D,\Delta}} {\DerOSEx{\Gamma}{\mathcal X}{A\orN B,C,\Delta} \quad \DerOSEx{\Gamma}{\mathcal X}{A\orN B,D,\Delta}}$\\

      By induction hypothesis we get \hfill $ \infer{\DerOSEx{\Gamma}{\mathcal X}{A,B,C\andN D,\Delta}} {\DerOSEx{\Gamma}{\mathcal X}{A,B,C,\Delta}\quad\DerOSEx{\Gamma}{\mathcal X}{A,B,D,\Delta}}$ 

    \item
      $\infer{\DerOSEx{\Gamma}{\mathcal X}{A\orN B,C\orN D,\Delta}} {\DerOSEx{\Gamma}{\mathcal X}{A\orN B,C, D,\Delta}}$  

      By induction hypothesis we get \hfill $\infer{\DerOSEx{\Gamma}{\mathcal X}{A,B,C\orN D,\Delta}} {\DerOSEx{\Gamma}{\mathcal X}{A,B,C, D,\Delta}}$ 

    \item
      $\infer[x\notin\FV{\Gam,\mathcal X,A\orN B,\Delta}]{\DerOSEx{\Gamma}{\mathcal X}{A\orN B,(\FA x C),\Delta}} {\DerOSEx{\Gamma}{\mathcal X}{A\orN B,C,\Delta}}$   

      By induction hypothesis we get \hfill $\infer[x\notin\FV{\Gam,\mathcal X,A, B,\Delta}]{\DerOSEx{\Gamma}{\mathcal X},{A,B,(\FA x C),\Delta}} {\DerOSEx{\Gamma}{\mathcal X}{A,B,C,\Delta}}$  

    \item
      $\infer[\mbox{\begin{tabular}l$C$ literal or\\ $\mathcal P$-positive\end{tabular}}]{\DerOSEx{\Gamma}{\mathcal X}{A\orN B,C,\Delta}} 
      {\DerOSEx{\Gamma,\non C}{\mathcal X}{A\orN B,\Delta}[\polar{\non C}]}$  

      By induction hypothesis we get \hfill  $\infer[\mbox{\begin{tabular}l$C$ literal or\\ $\mathcal P$-positive\end{tabular}}]{\DerOSEx{\Gamma}{\mathcal X}{A,B,C,\Delta}} {\DerOSEx{\Gamma,\non C}{\mathcal X}{A,B,\Delta}[\polar{\non C}]}$   
      
      \item
      $\infer{\DerOSEx{\Gamma}{\mathcal X}{A\orN B,\bot^-,\Delta}} {\DerOSEx{\Gamma}{\mathcal X}{A,B,\Delta}}$ \\ 

      By induction hypothesis we get\hfill $\infer{\DerOSEx{\Gamma}{\mathcal X}{A,B,\bot^-,\Delta}} {\DerOSEx{\Gamma}{\mathcal X}{A,B,\Delta}}$  

    \item
      $\infer{\DerOSEx{\Gamma}{\mathcal X}{A\orN B,\top^-,\Delta}} {}$ \\ 

      We get\hfill  $\infer{\DerOSEx{\Gamma}{\mathcal X}{A,B,\top^-,\Delta}} {}$

    \item $\infer{\DerOSEx{\Gamma}{C\andP D}{A\orN B,\Delta}} {\DerOSEx{\Gamma}{C}{A\orN B,\Delta} \quad \DerOSEx{\Gamma}{D}{A\orN B,\Delta}}$   

      By induction hypothesis we get  \hfill $ \infer{\DerOSEx{\Gamma}{C\andP D}{A,B,C\andN D,\Delta}} {\DerOSEx{\Gamma}{C}{A,B,\Delta}\quad\DerOSEx{\Gamma}{D}{A,B,\Delta}}$ 

    \item
      $\infer{\DerOSEx{\Gamma}{C_1\orP C_2}{A\orN B, \Delta}} {\DerOSEx{\Gamma}{C_i}{A\orN B,\Delta}}$  

      By induction hypothesis we get\hfill $\infer{\DerOSEx{\Gamma}{C_1\orP C_2}{A,B,\Delta}} {\DerOSEx{\Gamma}{C_i}{A,B,\Delta}}$ 

    \item
      $\infer{\DerOSEx{\Gamma}{\EX x C}{A\orN B,\Delta}} {\DerOSEx{\Gamma}{\subst C x t}{A\orN B,\Delta}}$   

      By induction hypothesis we get \hfill$\infer{\DerOSEx{\Gamma}{\EX x C}{A,B,\Delta}} {\DerOSEx{\Gamma}{\subst C x t}{A,B,\Delta}}$  

    \item
      $\infer{\DerOSEx{\Gamma}{\top^+}{A\orN B,\Delta}} {}$  \\
      
      We get \hfill
      $\infer{\DerOSEx{\Gamma}{\top^+}{A,B,\Delta}} {}$

    \item
       $  \infer[\atmCtxt\Gam,\non p,{\atmCtxt[\mathcal L]{\non\Del}}\models_{\mathcal T}]
	{\DerOSEx{\Gamma}{p}{A\orN B,\Delta}} {\strut} $ with $p$ being $\mathcal P$-positive

         We get\hfill
       $  \infer[\atmCtxt\Gam,\non p,{\atmCtxt[\mathcal L]{\non\Del}}\models_{\mathcal T}]
	{\DerOSEx{\Gamma}{p}{A,B,\Delta}} {\strut} $
   
    \item
       $  \infer[\atmCtxt\Gam,{\atmCtxt[\mathcal L]{\non\Del}}\models_{\mathcal T}]
	{\DerOSEx{\Gamma}{\bullet}{A\orN B,\Delta}} {\strut} $

         We get\hfill
       $  \infer[\atmCtxt\Gam,{\atmCtxt[\mathcal L]{\non\Del}}\models_{\mathcal T}]
	{\DerOSEx{\Gamma}{\bullet}{A,B,\Delta}} {\strut} $
   
    \item
      $\infer{\DerOSEx{\Gamma}{\bullet}{A\orN B,\Delta}} {\DerOSEx{\Gamma}{P}{A\orN B,\Delta}}$ where $\non P\in\Gamma$ is not $\mathcal P$-positive\\ 

      By induction hypothesis we get  \hfill
      $\infer{\DerOSEx{\Gamma}{\bullet}{A,B,\Delta}} {\DerOSEx{\Gamma}{P}{A,B,\Delta}}$ 
 
    \end{itemize}

  \item Inversion of $\FA x A$  

    \begin{itemize}
    \item $\infer{\DerOSEx{\Gamma}{\mathcal X}{(\FA x A),C\andN D,\Delta}} {\DerOSEx{\Gamma}{\mathcal X}{(\FA x A),C,\Delta} \quad \DerOSEx{\Gamma}{\mathcal X}{(\FA x A),D,\Delta}}$ 

      By induction hypothesis  we get $\infer[x\notin\FV{\Gam,\mathcal X,\Delta}]{\DerOSEx{\Gamma}{\mathcal X}{A,C\andN D,\Delta}} {{\DerOSEx{\Gamma}{\mathcal X}{A,C,\Delta}}$ \quad ${\DerOSEx{\Gamma}{\mathcal X}{A,D,\Delta}}}$ 
      
    \item
      $\infer{\DerOSEx{\Gamma}{\mathcal X}{(\FA x A),C\orN D,\Delta}} {\DerOSEx{\Gamma}{\mathcal X}{(\FA x A),C, D,\Delta}}$  

      By induction hypothesis we get\hfill $\infer{\DerOSEx{\Gamma}{\mathcal X}{A,C\orN D,\Delta}} {\DerOSEx{\Gamma}{\mathcal X}{A,C, D,\Delta}}$ 

    \item  
      $\infer[y\notin\FV{\Gam,\mathcal X,(\FA x A),\Delta}]{
      \DerOSEx{\Gamma}{\mathcal X}{(\FA x A),(\FA y D),\Delta}} {
      \DerOSEx{\Gamma}{\mathcal X}{(\FA x A),D,\Delta}}$  

      By induction hypothesis we get\hfill $\infer[y\notin\FV{\Gam,\mathcal X,A,\Delta}]{\DerOSEx{\Gamma}{\mathcal X}{A,(\FA y D),\Delta}} {\DerOSEx{\Gamma}{\mathcal X}{A,D,\Delta}} $

    \item
      $\infer[\mbox{\begin{tabular}l$C$ literal or\\ $\mathcal P$-positive\end{tabular}}]{\DerOSEx{\Gamma}{\mathcal X}{(\FA x A),C,\Delta}} {\DerOSEx{\Gamma,\non C}{\mathcal X}{(\FA x A),\Delta}[\polar{\non C}]}$  

      By induction hypothesis we get\hfill $\infer[\mbox{\begin{tabular}l$C$ literal or\\ $\mathcal P$-positive\end{tabular}}]{\DerOSEx{\Gamma}{\mathcal X}{A,C,\Delta}} {\DerOSEx{\Gamma,\non C}{\mathcal X}{A,\Delta}[\polar{\non C}]}$     
      
      \item
      $\infer{\DerOSEx{\Gamma}{\mathcal X}{(\FA x A),\bot^-,\Delta}} {\DerOSEx{\Gamma}{\mathcal X}{(\FA x A),\Delta}}$ \\ 

      By induction hypothesis we get\hfill $\infer{\DerOSEx{\Gamma}{\mathcal X}{A,\bot^-,\Delta}} {\DerOSEx{\Gamma}{\mathcal X}{A,\Delta}}$

    \item
      $\infer{\DerOSEx{\Gamma}{\mathcal X}{(\FA x A),\top^-,\Delta}} {}$ \\ 
      We get \hfill  $\infer{\DerOSEx{\Gamma}{\mathcal X}{A,\top^-,\Delta}} {}$

    \item
      $\infer{\DerOSEx{\Gamma}{C\andP D}{(\FA x A),\Delta}} {\DerOSEx{\Gamma}{C}{(\FA x A),\Delta} \quad \DerOSEx{\Gamma}{D}{(\FA x A),\Delta}}$ 

      By induction hypothesis  we get\hfill $\infer{\DerOSEx{\Gamma}{C\andP D}{A,\Delta}} {\DerOSEx{\Gamma}{C}{A,\Delta} \quad \DerOSEx{\Gamma}{D}{A,\Delta}}$ 

    \item
      $\infer{\DerOSEx{\Gamma}{C_1\orP C_2}{(\FA x A),\Delta}} {\DerOSEx{\Gamma}{C_i}{(\FA x A),\Delta}}$  

      By induction hypothesis we get\hfill $\infer{\DerOSEx{\Gamma}{C_1\orP C_2}{A,\Delta}} {\DerOSEx{\Gamma}{C_i}{A,\Delta}}$ 

    \item  
      $\infer{\DerOSEx{\Gamma}{\EX x D}{(\FA x A),\Delta}} {\DerOSEx{\Gamma}{\subst D x t}{(\FA x A),\Delta}}$  

      By induction hypothesis we get\hfill $\infer{\DerOSEx{\Gamma}{\EX x D}{A,\Delta}} {\DerOSEx{\Gamma}{\subst D x t}{A,\Delta}} $

      \item
         $\infer{\DerOSEx{\Gamma}{\top^+}{(\FA x A),C,\Delta}} {}$  
         
         We get\hfill $\infer{\DerOSEx{\Gamma}{\top^+}{A,\Delta}} {}$   	    

       \item
       $  \infer[\atmCtxt\Gam,\non p,{\atmCtxt[\mathcal L]{\non\Del}}\models_{\mathcal T}]
	{\DerOSEx{\Gamma}{p}{(\FA x A),\Delta}} {\strut} $ with $p$ being $\mathcal P$-positive

         We get\hfill
       $  \infer[\atmCtxt\Gam,\non p,{\atmCtxt[\mathcal L]{\non\Del}}\models_{\mathcal T}]
	{\DerOSEx{\Gamma}{p}{A,\Delta}} {\strut} $
	    
       \item
       $  \infer[\atmCtxt\Gam,{\atmCtxt[\mathcal L]{\non\Del}}\models_{\mathcal T}]
	{\DerOSEx{\Gamma}{\bullet}{(\FA x A),\Delta}} {\strut} $

         We get\hfill
       $  \infer[\atmCtxt\Gam,{\atmCtxt[\mathcal L]{\non\Del}}\models_{\mathcal T}]
	{\DerOSEx{\Gamma}{\bullet}{A,\Delta}} {\strut} $
	    
    \item
      $\infer{\DerOSEx{\Gamma}{\bullet}{(\FA x A),\Delta}} {\DerOSEx{\Gamma}{P}{(\FA x A),\Delta}}$ where $\non P\in\Gamma$ is not $\mathcal P$-positive\\ 

      By induction hypothesis we get  \hfill
      $\infer{\DerOSEx{\Gamma}{\bullet}{A,\Delta}} {\DerOSEx{\Gamma}{P}{A,\Delta}}$ 

    \end{itemize}

  \item Inversion of storing a literal or $\mathcal P$-positive formulae $A$
    \begin{itemize}
    \item $\infer{\DerOSEx{\Gamma}{\mathcal X}{A,C\andN D,\Delta}} {\DerOSEx{\Gamma}{\mathcal X}{A,C,\Delta} \quad \DerOSEx{\Gamma}{\mathcal X}{A,D,\Delta}}$  

      By induction hypothesis we get\hfill $\infer{\DerOSEx{\Gamma,\non A}{\mathcal X}{C\andN D,\Delta}[\polar{\non A}]} {\DerOSEx{\Gamma,\non A}{\mathcal X}{C,\Delta}[\polar{\non A}] \quad \DerOSEx{\Gamma,\non A}{\mathcal X}{D,\Delta}[\polar{\non A}]}$ 

    \item
      $\infer{\DerOSEx{\Gamma}{\mathcal X}{A,C\orN D,\Delta}} {\DerOSEx{\Gamma}{\mathcal X}{A,C, D,\Delta}}$  

      By induction hypothesis we get\hfill $\infer{\DerOSEx{\Gamma,\non A}{\mathcal X}{C\orN D,\Delta}[\polar{\non A}]} {\DerOSEx{\Gamma,\non A}{\mathcal X}{C, D,\Delta}[\polar{\non A}]}$ 

    \item
      $\infer[x\notin\FV{\Gam,\mathcal X,A,\Delta}]{\DerOSEx{\Gamma}{\mathcal X}{A,(\FA x D),\Delta}} {\DerOSEx{\Gamma}{\mathcal X}{A,D,\Delta}}$

      By induction hypothesis we get\hfill 
      $\infer[x\notin\FV{\Gam,\non A,\mathcal X,\Delta}]{\DerOSEx{\Gamma,\non A}{\mathcal X}{(\FA x D),\Delta}[\polar{\non A}]} {\DerOSEx{\Gamma,\non A}{\mathcal X}{D,\Delta}[\polar{\non A}]}$  

    \item
      $\infer[\mbox{\begin{tabular}l$B$ literal or\\ $\mathcal P$-positive\end{tabular}}]{\DerOSEx{\Gamma}{\mathcal X}{A,B,\Delta}} {\DerOSEx{\Gamma,\non B}{\mathcal X}{A,\Delta}[\polar{\non B}]}$

      We build\hfill $\infer[\mbox{\begin{tabular}l$B$ literal or\\ $\mathcal P$-positive\end{tabular}}]{\DerOSEx{\Gamma,\non A}{\mathcal X}{B,\Delta}[\polar{\non A}]} {\DerOSEx{\Gamma,\non A,\non B}{\mathcal X}{\Delta}[{\polar[\polar{\non A}]{\non B}}]}$\\
      proving the premiss using the induction hypothesis in case $\polar[\polar{\non B}]{\non A}=\polar[\polar{\non A}]{\non B}$, which holds unless $A=\non B$ and $A\in \UP$.

      In that case we have $\polar{\non A}=\mathcal P,\non A$, and we prove $\DerOSEx{\Gamma,\non A}{\mathcal X}{B,\Delta}[\polar{\non A}]$ with $(\Init)$ (Lemma~\ref{lem:weakselect}), as $\Theory{\atmCtxt[\polar{\non A}]{\Gam,\non A},\atmCtxt[\mathcal L]{\non B,\non \Delta}}$.

     \item
      $\infer{\DerOSEx{\Gamma}{\mathcal X}{A,\bot^-,\Delta}} {\DerOSEx{\Gamma}{\mathcal X}{A,\Delta}}$

      By induction hypothesis we get\hfill
      $\infer{\DerOSEx{\Gamma,\non A}{\mathcal X}{\bot^-,\Delta}[\polar{\non A}]} {\DerOSEx{\Gamma,\non A}{\mathcal X}{\Delta}[\polar{\non A}]}$  

    \item
      $\infer{\DerOSEx{\Gamma}{\mathcal X}{A,\top^-,\Delta}} {}$

We get\hfill
   $\infer{\DerOSEx{\Gamma,\non A}{\mathcal X}{\top^-,\Delta}[\polar{\non A}]} {}$ \\ 

    \item $\infer{\DerOSEx{\Gamma}{C\andP D}{A,\Delta}} {\DerOSEx{\Gamma}{C}{A,\Delta} \quad \DerOSEx{\Gamma}{D}{A,\Delta}}$  

      By induction hypothesis we get\hfill $\infer{\DerOSEx{\Gamma,\non A}{C\andP D}{\Delta}[\polar{\non A}]} {\DerOSEx{\Gamma,\non A}{C}{\Delta} \quad \DerOSEx{\Gamma,\non A}{D}{\Delta}[\polar{\non A}]}$ 

    \item
      $\infer{\DerOSEx{\Gamma}{C_1\orP C_2}{A,\Delta}} {\DerOSEx{\Gamma}{C_i}{A,\Delta}}$  

      By induction hypothesis we get\hfill $\infer{\DerOSEx{\Gamma,\non A}{C_1\orP C_2}{\Delta}[\polar{\non A}]} {\DerOSEx{\Gamma,\non A}{C_i}{\Delta}[\polar{\non A}]}$ 

    \item
      $\infer{\DerOSEx{\Gamma}{\EX x D}{A,\Delta}} {\DerOSEx{\Gamma}{\subst D x t}{A,\Delta}}$

      By induction hypothesis we get\hfill $\infer{\DerOSEx{\Gamma,\non A}{ \EX x D}{\Delta}[\polar{\non A}]} {\DerOSEx{\Gamma,\non A}{\subst D x t}{\Delta}[\polar{\non A}]}$  

    \item 
      $\infer{\DerOSEx{\Gamma}{\top^+}{A,\Delta}} {}$
      
      We get\hfill $\infer{\DerOSEx{\Gamma,\non A}{\top^+}{\Delta}[\polar{\non A}]} {}$ 

    \item 
      $  \infer[\atmCtxt\Gam,\non p,{\atmCtxt[\mathcal L]{\non A,\non\Del}}\models_{\mathcal T}]
      {\DerOSEx{\Gamma}{p}{A,\Delta}} {\strut} $ with $p$ being $\mathcal P$-negative

      We get\hfill
      $  \infer[{{\atmCtxt[\polar{\non A}]{\Gam,\non A},\non p,\atmCtxt[\mathcal L]{\non \Del}\models_{\mathcal T}}}]
      {\DerOSEx{\Gamma,\non A}{p}{\Delta}[\polar{\non A}]} {\strut} $\\
      as $p$ is also $\polar{\non A}$-positive.

    \item 
      $  \infer[\atmCtxt\Gam,{\atmCtxt[\mathcal L]{\non A,\non\Del}}\models_{\mathcal T}]
      {\DerOSEx{\Gamma}{\bullet}{A,\Delta}} {\strut} $

      We get\hfill
      $  \infer[{{\atmCtxt[\polar{\non A}]{\Gam,\non A},\atmCtxt[\mathcal L]{\non \Del}\models_{\mathcal T}}}]
      {\DerOSEx{\Gamma,\non A}{\bullet}{\Delta}[\polar{\non A}]} {\strut} $

    \item
      $\infer{\DerOSEx{\Gamma}{\bullet}{A,\Delta}} {\DerOSEx{\Gamma}{P}{A,\Delta}}$ where $\non P\in\Gamma$ is not $\mathcal P$-positive\\ 

      By induction hypothesis we get  \hfill
      $\infer{\DerOSEx{\Gamma,\non A}{\bullet}{\Delta}[\polar{\non A}]} {\DerOSEx{\Gamma,\non A}{P}{\Delta}[\polar{\non A}]}$\\
      using either $(\Select)$ or $(\Select[-])$ depending on whether $P$ is $\polar {\non A}$-negative.
    \end{itemize}

  \item Inversion of ($\bot^-$)
    \begin{itemize}
    \item $\infer{\DerOSEx{\Gamma}{\mathcal X}{\bot^-,C\andN D,\Delta}} {\DerOSEx{\Gamma}{\mathcal X}{\bot^-,C,\Delta} \quad \DerOSEx{\Gamma}{\mathcal X}{\bot^-,D,\Delta}}$  

      By induction hypothesis we get\hfill $\infer{\DerOSEx{\Gamma}{\mathcal X}{C\andN D,\Delta}} {\DerOSEx{\Gamma}{\mathcal X}{C,\Delta} \quad \DerOSEx{\Gamma}{\mathcal X}{D,\Delta}}$ 

    \item
      $\infer{\DerOSEx{\Gamma}{\mathcal X}{\bot^-,C\orN D,\Delta}} {\DerOSEx{\Gamma}{\mathcal X}{\bot^-,C, D,\Delta}}$  

      By induction hypothesis we get\hfill $\infer{\DerOSEx{\Gamma}{\mathcal X}{C\orN D,\Delta}} {\DerOSEx{\Gamma}{\mathcal X}{C, D,\Delta}}$ 

    \item
      $\infer[x\notin\FV{\Gam,\mathcal X,\Delta}]{\DerOSEx{\Gamma}{\mathcal X}{\bot^-,(\FA x D),\Delta}} {\DerOSEx{\Gamma}{\mathcal X}{\bot^-,D,\Delta}}$  \\
      By induction hypothesis we get\hfill $\infer[x\notin\FV{\Gam,\mathcal X,\Delta}]{\DerOSEx{\Gamma}{\mathcal X}{(\FA x D),\Delta}} {\DerOSEx{\Gamma}{\mathcal X}{D,\Delta}}$  

    \item
      $\infer[\mbox{\begin{tabular}l$B$ literal or\\ $\mathcal P$-positive\end{tabular}}]{\DerOSEx{\Gamma}{\mathcal X}{\bot^-,B,\Delta}} {\DerOSEx{\Gamma,\non B}{\mathcal X}{\bot^-,\Delta}[\polar{\non B}]}$

      By induction hypothesis we get\hfill $\infer[\mbox{\begin{tabular}l$B$ literal or\\ $\mathcal P$-positive\end{tabular}}]{\DerOSEx{\Gamma}{\mathcal X}{B,\Delta}} {\DerOSEx{\Gamma,\non B}{\mathcal X}{\Delta}[\polar{\non B}]}$ 
      
     \item
      $\infer{\DerOSEx{\Gamma}{\mathcal X}{\bot^-,\bot^-,\Delta}} {\DerOSEx{\Gamma}{\mathcal X}{\bot^-,\Delta}}$ \\ 

      By induction hypothesis we get\hfill
      $\infer{\DerOSEx{\Gamma}{\mathcal X}{\bot^-,\Delta}} {\DerOSEx{\Gamma}{\mathcal X}{\Delta}}$  

    \item
      $\infer{\DerOSEx{\Gamma}{\mathcal X}{\bot^-,\top^-,\Delta}} {}$

      We get\hfill
      $\infer{\DerOSEx{\Gamma}{\mathcal X}{\top^-,\Delta}} {}$

    \item $\infer{\DerOSEx{\Gamma}{C\andP D}{\bot^-,\Delta}} {\DerOSEx{\Gamma}{C}{\bot^-,\Delta} \quad \DerOSEx{\Gamma}{D}{\bot^-,\Delta}}$  

      By induction hypothesis we get\hfill $\infer{\DerOSEx{\Gamma}{C\andP D}{\Delta}} {\DerOSEx{\Gamma}{C}{\Delta} \quad \DerOSEx{\Gamma}{D}{\Delta}}$ 

    \item
      $\infer{\DerOSEx{\Gamma}{C_1\orP C_2}{\bot^-,\Delta}} {\DerOSEx{\Gamma}{C_i}{\Delta}}$  

      By induction hypothesis we get\hfill $\infer{\DerOSEx{\Gamma}{C_1\orP C_2}{\Delta}} {\DerOSEx{\Gamma}{C_i}{\Delta}}$ 

    \item
      $\infer{\DerOSEx{\Gamma}{\EX x D}{\bot^-,\Delta}} {\DerOSEx{\Gamma}{\subst D x t}{\bot^-,\Delta}}$  \\
      By induction hypothesis we get\hfill $\infer{\DerOSEx{\Gamma}{ \EX x D}{\Delta}} {\DerOSEx{\Gamma}{\subst D x t}{\Delta}}$  

   \item 
   
         $\infer{\DerOSEx{\Gamma}{\top^+}{\bot^-,\Delta}} {}$
         
         We get\hfill $\infer{\DerOSEx{\Gamma}{\top^+}{\Delta}} {}$ 

    \item 
      $  \infer[\atmCtxt\Gam,\non p,{\atmCtxt[\mathcal L]{\non\Del}}\models_{\mathcal T}]
	{\DerOSEx{\Gamma}{p}{\bot^-,\Delta}} {\strut} $ with $p$ being $\mathcal P$-positive
        
         By induction hypothesis we get\hfill
       $  \infer[\atmCtxt\Gam,\non p,{\atmCtxt[\mathcal L]{\non\Del}}\models_{\mathcal T}]
	{\DerOSEx{\Gamma,\non A}{p}{\Delta}} {\strut} $

    \item 
                $  \infer[\atmCtxt\Gam,{\atmCtxt[\mathcal L]{\non\Del}}\models_{\mathcal T}]
	{\DerOSEx{\Gamma}{\bullet}{\bot^-,\Delta}} {\strut} $
        
         By induction hypothesis we get\hfill
       $  \infer[\atmCtxt\Gam,{\atmCtxt[\mathcal L]{\non\Del}}\models_{\mathcal T}]
	{\DerOSEx{\Gamma,\non A}{\bullet}{\Delta}} {\strut} $
   
    \item
      $\infer{\DerOSEx{\Gamma}{\bullet}{\bot^-,\Delta}} {\DerOSEx{\Gamma}{P}{\bot^-,\Delta}}$ where $\non P\in\Gamma$ is not $\mathcal P$-positive\\ 

      By induction hypothesis we get  \hfill
      $\infer{\DerOSEx{\Gamma}{\bullet}{\Delta}} {\DerOSEx{\Gamma}{P}{\Delta}}$
    \end{itemize}
  \item Inversion of $\top^-$: nothing to do.
  \end{itemize}
\end{proof}

Now that we have proved the invertibility of asynchronous rules, we can use it to transform any proof of \LKThEx\ into a proof of \LKThp.

\begin{lemma}[Encoding \LKThEx\ in \LKThp]
  \label{LKex}\strut\
  \begin{enumerate} 
  \item 
    If $\DerOSEx {\Gam} {A}{} $ is provable in \LKThEx, then $\DerPos{\Gam}{A}{}{\mathcal P}$ is provable in \LKThp.
  \item 
    If $\DerOSEx {\Gam} {\bullet}{\Del}$ is provable in \LKThEx, then $\DerNeg{\Gam}{\Del}{}{\mathcal P}$ is provable in \LKThp.
  \end{enumerate}
\end{lemma}

\begin{proof}
  By simultaneous induction on the assumed derivation.
  \begin{enumerate}
  \item For the first item we get, by case analysis on the last rule of the derivation:  
    \begin{itemize}
    \item $ \infer{\DerOSEx{\Gam}{A_1 \andP A_2}{}}
      { \DerOSEx{\Gam}{A_1}{} \quad
        \DerOSEx{\Gam}{ A_2}{}    
      }
      $
      with $A=A_1 \andP A_2$.

      The induction hypothesis on $\DerOSLKEx{\Gam}{A_1}{}$ gives $\DerPosLK{\Gam}{A_1}$ and the induction hypothesis on $\DerOSLKEx{\Gam}{A_2}{}$ gives $\DerPosLK{\Gam}{A_2}$. We get:
      $$ \infer{\DerOSPos{\Gam}{A_1 \andP A_2}}
      { \DerOSPos{\Gam}{A_1} \quad
        \DerOSPos{\Gam}{ A_2}    
      }
      $$

    \item $ \infer{\DerOSEx{\Gam}{A_1 \orP A_2}{}}
      { \DerOSEx{\Gam}{A_i}{}     
      }
      $
      with $A=A_1 \orP A_2$.

      The induction hypothesis on $\DerOSLKEx{\Gam}{A_i}{}$ gives $\DerPosLK{\Gam}{A_i}$. We get:
      $$ \infer{\DerOSPos{\Gam}{A_1 \orP A_2}}
      { \DerOSPos{\Gam}{A_i}}
      $$

    \item $ \infer{\DerOSEx{\Gam}{\EX x A}{}}
      { \DerOSEx{\Gam}{\{t/x\}A} {} }
      $
      with $A=\EX x A$.

      The induction hypothesis on $\DerOSLKEx{\Gam}{\{t/x\}A}{}$ gives $\DerPosLK{\Gam}{\{t/x\}A}$. We get:
      $$ \infer{\DerOSPos{\Gam}{\EX x A}}
      { \DerOSPos{\Gam}{\{t/x\}A}}
      $$	

    \item 
    $   \infer[\atmCtxt\Gam,\non p\models_{\mathcal T}]
			{\DerOSEx{\Gamma  }{p}{}} {\strut} 
                        $
      with $A=p$ where $p$ is a $\mathcal P$-positive literal.
      
      We can perform the same step in \LKThp:
	$$	\infer[\atmCtxt\Gam,\non p\models_{\mathcal T}]
	{\DerOSPos{\Gamma  }{p}} {\strut}  $$
	
    \item $ \infer{\DerOSEx{\Gam}{N}{}}
      { \DerOSEx{\Gam}{\bullet}{N}     
      }
      $
      with $A=N$ and $N$ is not $\mathcal P$-positive.

      The induction hypothesis on $\DerOSLKEx{\Gam}{\bullet}{N}$ gives $\DerNegLK{\Gam}{N}$. We get:
      $$ \infer{\DerOSPos{\Gam}{N}}
      { \DerOSNeg{\Gam}{N}}
      $$
      
    \end{itemize}

  \item For the second item, we use the height-preserving invertibility of the asynchronous rules, so that we can assume without loss of generality that if $\Del$ is not empty then the last rule of the derivation decomposes one of its formulae. 

    \begin{itemize}
    \item 
      $ \infer{\DerOSEx{\Gam}{\bullet}{A_1 \andN A_2,\Delta_1}}
      { \DerOSEx{\Gam}{\bullet}{A_1,\Delta_1} \quad
        \DerOSEx{\Gam}{\bullet}{ A_2,\Delta_1}   
      }
      $
      with $\Del=A_1 \andN A_2,\Del_1$.

      The induction hypothesis on $\DerOSLKEx{\Gam}{\bullet}{A_1,\Delta_1}$ gives $\DerNegLK{\Gam}{A_1,\Delta_1}$ and the induction hypothesis on $\DerOSLKEx{\Gam}{\bullet}{A_2,\Delta_2}$ gives $\DerNegLK{\Gam}{A_2,\Delta_2}$. We get:
      $$ \infer{\DerOSNeg{\Gam}{A_1 \andN A_2,\Delta_1}}
      { \DerOSNeg{\Gam}{A_1,\Delta_1} \quad
        \DerOSNeg{\Gam}{ A_2,\Delta_1}    
      }
      $$

    \item  
      $ \infer{\DerOSEx{\Gam}{\bullet}{A_1 \orN A_2,\Delta_1}}
      { \DerOSEx{\Gam}{\bullet}{A_1,A_2,\Delta_1} }
      $
      with $\Del=A_1 \orN A_2,\Del_1$.

      The induction hypothesis on $\DerOSLKEx{\Gam}{\bullet}{A_1,A_2,\Delta_1}$ gives $\DerNegLK{\Gam}{A_1,A_2,\Delta_1}$ and we get:
      $$ \infer{\DerOSNeg{\Gam}{A_1 \orN A_2,\Delta_1}}
      { \DerOSNeg{\Gam}{A_1,A_2,\Delta_1}    
      }
      $$

    \item 
      $ \infer[x\not\in\FV{\Gamma,\Del_1}]{\DerOSEx{\Gam}{\bullet}{\FA x A,\Delta_1}}
      { \DerOSEx{\Gam}{\bullet}{A,\Delta_1}  }
      $
      with $\Delta=\FA x A,\Delta_1$.

      The induction hypothesis on $\DerOSLKEx{\Gam}{\bullet}{A,\Delta_1}$ gives $\DerNegLK{\Gam}{A,\Delta_1}$. We get:
      $$ \infer[x\not\in\FV{\Gamma,\Del_1}]{\DerOSNeg{\Gam}{\FA x A,\Delta_1}}
      {\DerOSNeg{\Gam}{A,\Delta_1}}
      $$
      
    \item 	
      $
      \infer{\DerOSEx{\Gam}{\bullet}{A,\Delta_1}}{ \DerOSEx{\Gam,\non A}{\bullet}{\Delta_1}[\polar{\non A}] }
      $
      with $\Del=A,\Del_1$ and $A$ is a literal or is $\mathcal P$-positive.

      The induction hypothesis on $\DerOSLKEx{\Gam,\non A}{\bullet}{\Delta_1}[\polar{\non A}]$ gives $\DerNegLK[\polar{\non A}]{\Gam,\non {A}}{\Delta_1}$. We get:
      \[
      \infer{\DerOSNeg{\Gam}{A,\Delta_1}}
      { \DerOSNeg[\polar{\non A}]{\Gam,\non {A}}{\Delta_1}}
      \]

  \item 	
    $ \infer{\DerOSEx{\Gam}{\bullet}{\bot^-,\Delta_1}}
      { \DerOSEx{\Gam}{\bullet}{\Delta_1} }
      $
      with $\Del=\bot^-,\Del_1$.

      The induction hypothesis on $\DerOSLKEx{\Gam}{\bullet}{\Delta_1}$ gives $\DerNegLK{\Gam}{\Delta_1}$. We get:
      $$ \infer{\DerOSNeg{\Gam}{\bot^-,\Delta_1}}
      { \DerOSNeg{\Gam}{\Delta_1}}
      $$

 \item 	
    $ \infer{\DerOSEx{\Gam}{\bullet}{\top^-,\Delta_1}}{  } $
    with $\Del=\top^-,\Del_1$.

      We get:
      \[ \infer{\DerOSNeg{\Gam}{\top^-,\Delta_1}} { } \]

    \item 
      $ \infer{\DerOSEx{\Gam,\non P}{\bullet}{\Del}}
      { \DerOSEx{\Gam,\non P}{P}{\Del} }
      $
      where $P$ is not $\mathcal P$-negative.

      As already mentioned, we can assume without loss of generality that $\Del$ is empty.
      The induction hypothesis on $\DerOSLKEx{\Gam,\non P}{P}{}$ gives $\DerPosLK{\Gam,\non {P}}{P}$. We get:
      \[ \infer{\DerOSNeg{\Gam,\non{P}}{}}
      { \DerOSPos{\Gam,\non {P}}{P}}
      \]

    \item 
      $	\infer[\atmCtxt\Gam,{\atmCtxt[\mathcal L]{\non\Del}}\models_{\mathcal T}]
	{\DerOSEx{\Gamma  }{\bullet}{\Del}} {\strut}$

      As already mentioned, we can assume without loss of generality that $\Del$ is empty.	
      We get:
      \[   \infer[\atmCtxt\Gam\models_{\mathcal T}]{\DerOSNeg{\Gamma}{}} {\strut} 
      \]
    \end{itemize}
  \end{enumerate}
\end{proof}

\begin{lemma}
  \label{Lconn1}
  We have:
  \begin{enumerate}
  \item $\DerNegLK{}{\non {\top^+},\top^-}$, and
  \item $\DerNegLK{}{\non {\top^-},\top^+}$, and
  \item $\DerNegLK{}{\non {(A \andP B)},( A \andN B)}$, and
  \item $\DerNegLK{}{\non {(A \andN B)},( A \andP B)}$, provided that sequent is safe.
  \end{enumerate}
\end{lemma}
\begin{proof}
  \begin{enumerate}
  \item For the first item we get:
    \[
    \infer{\DerOSNeg{}{\non {\top^+},\top^-}}{}
    \]
  \item For the second item we get:
    \[
    \infer{\DerOSNeg{}{\non {\top^-},\top^+}}{
      \infer{\DerOSNeg{\top^-}{\top^+}}{
        \infer{\DerOSNeg{\top^-,\non{\top^+}}{}}{
          \infer{\DerOSPos{\top^-,\non{\top^+}}{\top^+}}{            
          }
        }
      }
    }
    \]
  \item For the third item we get:
    \[
    \infer[{\mbox{Lemma}~\ref{LKex}(2)}]{\DerOSNeg{}{\non {(A \andP B)},(A \andN B)}}{
      \infer{\DerOSEx{}{\bullet}{\non {(A \andP B)},(A \andN B)}}{
        \Infer{\DerOSEx{}{\bullet}{(\non A \orN \non B),(A \andN B)}}{
          \infer{\DerOSEx{}{\bullet}{\non A,\non B,A}}{
            \infer{\DerOSEx{A}{\bullet}{\non B,A}[\polar{A}]}{
              \iinfer{\DerOSEx{A}{\non A}{\non B,A}[\polar{A}]}{
                \DerOSEx{}{\non A}{\non B,A}[\polar{A}] 
              }
            }
          }
          \quad
          \infer{\DerOSEx{}{\bullet}{\non A,\non B,B}}{
            \infer{\DerOSEx{B}{\bullet}{\non A,B}[\polar{B}]}{
              \iinfer{\DerOSEx{B}{\non B}{\non A,B}[\polar{B}]}{
                \DerOSEx{}{\non B}{\non A,B}[\polar{B}]
              }
            }
          }
        }
      }
    }
    \]
    Both left hand side and right hand side can be closed by Lemma~\ref{con1}.
  \item For the fourth item, we get:
    \[ \small
    \hspace{-20pt}
    \Infer{\DerOSNeg{}{\non{(A \andN B)},( A \andP B)}}
    { \infer[\cut_7]{\DerOSNeg{( A \andN B),\non {(A \andP B)}} {}}
      { \hspace{-20pt}
        \iinfer{\DerOSNeg{(A \andN B),( \non A \orN \non B)} {A}}
        {	\infer[\mbox{Lemma}~\ref{LKex}(2)]{\DerOSNeg{A \andN B}{A} }
          {\infer{\DerOSEx{A \andN B}{\bullet}{ A} }
            {\iinfer{\DerOSEx{A \andN B}{\non A \orP \non B}{ A} }
              {\infer{\DerOSEx{}{\non A \orP \non B}{ A} }
                {\DerOSEx{}{\non A}{A} }
              }
            }	}
        }
        \hspace{-10pt}
        \quad
        {\infer[\cut_7]{\DerOSNeg{(A \andN B),( \non A \orN \non B)} {\non A}}
          { \iinfer{\DerOSNeg{(A \andN B),( \non A \orN \non B)} {\non A,B}}
            {\iinfer{\DerOSNeg{A \andN B}{\non A,B}}
              {\infer[\mbox{Lemma}~\ref{LKex}(2)]{\DerOSNeg{A \andN B}{B}}
                {\infer{\DerOSEx{A \andN B} {\bullet}{B}}
                  {\iinfer{\DerOSEx{A \andN B} {\non A \orP \non B}{B}}
                    {\infer{\DerOSEx{} {\non A \orP \non B}{B}}
                      {\DerOSEx{} {\non B}{B}}
                    }
                  }
                }}}
            \hspace{-15pt}
            \quad 		
            { \iinfer{\DerOSNeg{(A \andN B),( \non A \orN \non B)} {\non A,\non B}}
              {\infer[\mbox{Lemma}~\ref{LKex}(2)]{\DerOSNeg{\non A \orN \non B} {\non A,\non B}}
                {\infer{\DerOSEx{\non A \orN \non B}{\bullet} {\non A,\non B}}
                  {\iinfer{\DerOSEx{\non A \orN \non B}{A \andP B} {\non A,\non B}}	
                    {\infer{\DerOSEx{} {A \andP B}{\non A,\non B}}
                      {\DerOSEx{} {A}{\non A,\non B}
                      \quad
                       \DerOSEx{} {B}{\non A,\non B}
                      }
                    }
                  }
		}
	      }
	    }
          }
        } 
      } 
    } 
    \]
  \end{enumerate}
  All branches are closed by Lemma~\ref{con1}.
\end{proof}

\begin{lemma}\strut
  \label{Lconn2} 

  If $\DerNegLK{\Gamma}{\Delta,C}$ and $\DerNegLK{\Gamma}{D,\non C}$ then $\DerNegLK{\Gamma}{\Delta,D}$, provided that sequent is safe.
\end{lemma}
\begin{proof}
  \[
  \infer[\cut_7]{\DerOSNeg{\Gamma}{\Delta,D}}
  { \iinfer{\DerOSNeg{\Gamma}{D,\Del,C} }
    {\DerOSNeg{\Gamma}{\Del,C}}
    \quad
    \iinfer{\DerOSNeg{\Gamma}{\Del,D,\non C}}
    {\DerOSNeg{\Gamma}{D,\non C}}
  }
  \]
\end{proof}		

\begin{corollary}[Changing the polarity of connectives]
  \label{conF}
  Provided those sequents are safe,
  \begin{enumerate}
  \item If $\DerOSNeg{\Gamma}{\top^+,\Del}$ then $\DerOSNeg{\Gamma}{\top^-,\Del}$;
  \item If $\DerOSNeg{\Gamma}{\top^-,\Del}$ then $\DerOSNeg{\Gamma}{\top^+,\Del}$;
  \item If $\DerOSNeg{\Gamma}{\bot^+,\Del}$ then $\DerOSNeg{\Gamma}{\bot^-,\Del}$;
  \item If $\DerOSNeg{\Gamma}{\bot^-,\Del}$ then $\DerOSNeg{\Gamma}{\bot^+,\Del}$;
  \item If $\DerOSNeg{\Gamma}{A \andP B,\Del}$ then $\DerOSNeg{\Gamma}{A \andN B,\Del}$;
  \item If $\DerOSNeg{\Gamma}{A \andN B,\Del}$ then $\DerOSNeg{\Gamma}{A \andP B,\Del}$;
  \item If $\DerOSNeg{\Gamma}{A \orP B,\Del}$ then $\DerOSNeg{\Gamma}{A \orN B,\Del}$;
  \item If $\DerOSNeg{\Gamma}{A \orN B,\Del}$ then $\DerOSNeg{\Gamma}{A \orP B,\Del}$.
  \end{enumerate}
  Furthermore, notice that in each implication, the safety of one sequent implies the safety of the other.
\end{corollary}
\begin{proof}
  \begin{enumerate}
  \item By Lemma~\ref{Lconn2} and Lemma~\ref{Lconn1}(1). 
  \item By Lemma~\ref{Lconn2} and Lemma~\ref{Lconn1}(2). 
  \item By Lemma~\ref{Lconn2} and Lemma~\ref{Lconn1}(1). 
  \item By Lemma~\ref{Lconn2} and Lemma~\ref{Lconn1}(2). 
  \item By Lemma~\ref{Lconn2} and Lemma~\ref{Lconn1}(3). 
  \item By Lemma~\ref{Lconn2} and Lemma~\ref{Lconn1}(4). 
  \item By Lemma~\ref{Lconn2} and Lemma~\ref{Lconn1}(3). 
  \item By Lemma~\ref{Lconn2} and Lemma~\ref{Lconn1}(4). 
  \end{enumerate}
\end{proof}
We have proven that changing the polarities of the connectives that are present in a sequent, does not change the provability of that sequent in \LKThp.

\section{Completeness}
\LKTh\ is a complete system for first-order logic modulo a theory. To show this, we review the grammar of first-order formulae and map those formulae to polarised formulae.

\begin{definition}[Plain formulae]\strut
  \label{def:FOLform}
  Let $P^a_{\Sigma}$ be a sub-signature of the first-order predicate signature $P_{\Sigma}$ such that for every predicate symbol $P/n$ of $P_{\Sigma}$, $P/n$ is in $P^a_{\Sigma}$ \iff\ $\non P/n$ is not in $P^a_{\Sigma}$.

  Let $\mathcal A$ be the subset of $\mathcal L$ consisting of those literals whose predicate symbols are in $P^a_{\Sigma}$. Literals in $\mathcal A$, denoted $a$, $a'$, etc, are called \Index[atom]{atoms}.

  The formulae of first-order logic, here called \Index[plain formula]{plain formulae}, are given by the following grammar:
  \[
  \begin{array}{lll}
    \mbox{Plain formulae }&A,B,\ldots&\recdef a\mid A\vee B\mid A\wedge B \mid \FA x A \mid \EX x A \mid \neg A
  \end{array}
  \]
  where $a$ ranges over atoms.
\end{definition}

\renewcommand\bar[1]{#1}

\begin{definition}[$\bar\psi$]
  Let $\bar\psi$ be the function that maps every plain formula to a set of formulae (in the sense of Definition~\ref{def:polarformulae}) defined as follows:
  \[
  \begin{array}{ l c r }
    \bar\psi(a) & \eqdef & \{ a  \} \\	
    \bar\psi(A \andF B) & \eqdef & \{ A' \andN B',A' \andP B' \mid A' \in \psi(A), B' \in \psi(B)  \} \\
    \bar\psi(A \orF B) & \eqdef & \{ A' \orN B',A' \orP B' \mid A' \in \psi(A), B' \in \psi(B)  \} \\
    \bar\psi(\EX x A) & \eqdef & \{ \EX x A' \mid A' \in \psi(A) \} \\
    \bar\psi(\FA x A) & \eqdef & \{ \FA x A' \mid A' \in \psi(A) \} \\
    \bar\psi(\neg A) & \eqdef & \{ \non {A'} \mid A' \in \psi(A) \} \\
    \bar\psi(\Del,A) & \eqdef & \{ \Del',A' \mid \Del' \in \psi(\Del), A' \in \psi(A) \} \\
    \bar\psi(\emptyset) & \eqdef & \emptyset \\	 
  \end{array}
  \]
\end{definition}

\begin{remark}
  \label{rm1}
  \label{rm2}
  \begin{enumerate}
  \item
    \(\bar\psi(A)\neq\emptyset\)
  \item
    If $A' \in \bar\psi(A)$, then $ \subst {A'} x t \in \bar\psi  (\subst {A'} x t)$.
  \item
    If $C' \in \bar\psi  (\subst {A} x t)$, then $C'=\subst {A'} x t$ for some $A'\in\bar\psi(A)$.
  \end{enumerate}
\end{remark}

\begin{notation}
\label{entailment}
  When $F$ is a plain formula and $\Psi$ is a set of plain
  formulae, $\Psi\models F$ means that $\Psi$ entails $F$ in
  first-order classical logic.

  Given a theory $\mathcal T$ (given by a semantical inconsistency
  predicate), we define the set of all \Index[theory lemma]{theory
    lemmas} as
  \[\TLemma \eqdef \{l_1\vee \cdots
  \vee l_n \mid \non {\psi(l_1)} , \cdots , \non {\psi(l_n)} \models_{\mathcal T} \}
  \]

  We generalise the notation $\models_{\mathcal T}$ to write
  $\Psi\models_{\mathcal T} F$ when $\Psi_{\mathcal T},\Psi\models
  F$, in which case we say that $F$ is a \emph{semantical consequence}
  of $\Psi$.  
\end{notation}

\begin{notation}
  In the rest of this section we will use the notation $A\andF^? B$ (\resp $A\orF^? B$) to ambiguously represent either $A\andP B$ or $A\andN B$ (\resp $A\orP B$ or $A\orN B$). This will make the proofs more compact, noticing that Corollary~\ref{conF}(2) and~\ref{conF}(4) respectively imply the admissibility in \LKThp\ of
  \[
  \infer{\DerNeg{\Gamma}{\Delta,A\andF^? B}{}{\mathcal P}}{\DerNeg{\Gamma}{\Delta,A\andN B}{}{\mathcal P}}
  \qquad
  \infer{\DerNeg{\Gamma}{\Delta,A\orF^? B}{}{\mathcal P}}{\DerNeg{\Gamma}{\Delta,A\orN B}{}{\mathcal P}}
  \]
provided the sequents are safe (and note that safety of the conclusion entails safety of the premiss).
\end{notation}

\begin{lemma}[Equivalence between different polarisations]\strut
  \label{Fcomp}

  For all $A', A'' \in \bar\psi{(A)}$, we have $\DerNegLK{\Gam}{A', \non{A''},\Del}{}$, provided the sequent is safe.
\end{lemma}
\begin{proof}
  In the proof below, for any formula $A$, the notations $A'$ and $A''$ will systematically designate elements of $\psi{(A)}$.
  
  The proof is by induction on $A$: 
  \begin{enumerate}
  \item $A=a$

    Let $A',A'' \in \bar\psi(a)=\{a\}$. Therefore $A'=A''=A=a$.
    \[
    \Infer{\DerOSNeg{\Gam}{\psi(a),\non \psi(a),\Del}}{
      \infer[{[{(\Id[2])}]}]{\DerOSNeg[\mathcal P']{\Gam,\non \psi(a),\psi(a),\Gam'}{}}{}   
    } 
    \]

  \item $A= A_1 \andF A_2$  

    Let $A'_1,A''_1 \in \bar\psi {(A_1)}$ , $A'_2,A''_2 \in \bar\psi {(A_2)}$ and $A'= A'_1 \andF^? A'_2$, $A''= A''_1 \andF^? A''_2$.
    \[
    \infer{\DerOSNeg{\Gam}{A', \non{A''},\Del}}
    {  \infer{\DerOSNeg{\Gam}{A', \non{A''_1} \orN \non{A''_2},\Del }}
      { \Infer { \DerOSNeg{\Gam}{A'_1 \andN A'_2, \non{A''_1} \orN \non{A''_2},\Del }}  
        {  { \infer{\DerOSNeg{\Gam}{A'_1, \non{A''_1}, \non{A''_2},\Del }}
            {\DerOSNeg{\Gam}{A'_1, \non{A''_1},\Del}}
          }
          \quad
          {  \infer{\DerOSNeg{\Gam}{A'_2, \non{A''_1}, \non{A''_2},\Del }}
            {\DerOSNeg{\Gam}{A'_2, \non{A''_2},\Del }}
          }
        }
      }
    }
    \]
    We can complete the proof on the left-hand side by applying the induction hypothesis on $A_1$ and on the right-hand side by applying the induction hypothesis on $A_2$.

  \item $A=A_1 \orF A_2$

    By symmetry, using the previous case.


  \item $A= \FA x A_1$

    Let $A' = \FA x A'_1$ and $A'' = \FA x A''_1$.
    \[
    \infer{\DerOSNeg{\Gam} {\FA x A'_1, \EX x \non{A''_1},\Del}}
    { \infer[\mbox{Lemma~\ref{Lconn2}}]{\DerOSNeg{\Gam} {A'_1, \EX x \non{A''_1},\Del}}
      {  \infer[\mbox{Lemma~\ref{LKex}(2)}]{\DerOSNeg{\Gam} {A''_1, \EX x \non{A''_1},\Del}}   
        { \iinfer{\DerOSEx{\Gam,\FA x A''_1} {\bullet}{A''_1,\Del}}
          { \infer{\DerOSEx{}{\EX x A''_1} {A''_1}[\mathcal P']}
            {\DerOSEx{}{\non {A''_1}} {A''_1}[\mathcal P']}
          }
        }
        \quad
        {\DerOSNeg{\Gam} {A'_1, \non{A''_1},\Del}}
      }
    }
    \]
    We can complete the proof on the left-hand side by Lemma~\ref{con1} and the right-hand side by applying the induction hypothesis on $A_1$.
  \item $A= \EX x A_1$

    By symmetry, using the previous case.



  \item $A= \neg A_1$

    Let  $A',A'' \in \bar\psi{(\neg A_1)}$.\\
    Let  $A'= \non {A'_1}$ with $A'_1 \in \bar\psi{(A_1)}$
    and $A''= \non {A''_1}$ with $A''_1 \in \bar\psi{(A_1)}$.

    The induction hypothesis on $A_1$ we get: ${\DerNegLK{\Gam}{A',\non{A''},\Del}}$ and we are done.
  \end{enumerate}
\end{proof}

\begin{definition}[Theory restricting]
A polarisation set \Index[theory restricting]{does not restrict} the theory $\mathcal T$ if for all sets $\mathcal B$ of literals that are semantically inconsistent (\ie $\mathcal B\models_{\mathcal T}$), there is a subset $\mathcal B'\subseteq\mathcal B$ that is already semantically inconsistent and such that at most one literal of $\mathcal B'$ is $\mathcal P$-negative.
\end{definition}

\begin{remark}
The empty polarisation set restricts no theories.
\end{remark}

\begin{theorem}[Completeness of \LKThp]\strut
  \label{complete}
  Assume $\mathcal P$ does not restrict $\mathcal T$ and ${\Delta}\models_{\mathcal T}{A}$.

  Then for all $A' \in \bar\psi(A)$ and $\Del' \in \bar\psi(\Del)$, we have $\DerNegLK{}{A',\non{\Del'}}{}$, provided that sequent is safe.
\end{theorem}
\begin{proof}
  We prove a slightly more general statement:\\
  for all $A' \in \bar\psi(A)$ and all multiset $\Del'$ of formulae that contain an element of $\psi(\Del)$ as a sub-multiset, we have $\DerNegLK{}{A',\non{\Del'}}{}$, provided that sequent is safe.

  We caracterise ${\Delta}\models_{\mathcal T}{A}$ by the derivability of the sequent $\DerOSFl{\TLemma ,\Del}{A}$ in a standard natural deduction system for first-order classical logic. We write $\DerOSFOL {\TLemma ,\Delta} {A}$ for this derivability property.

  For any formula $A$, the notation $A'$ will systematically designate an element of $\psi(A)$.

  The proof is by induction of $\DerOSFOL {\TLemma ,\Delta} {A}$, and case analysis on the last rule: \\
  \begin{itemize}
  \item Axiom:
    $$\infer[A \in \TLemma ,\Del]{\DerOSFl{\TLemma ,\Del}{A}}{}$$
    By case analysis:
    \begin{itemize}
    \item  If $A \in \Del$ then we prove ${\DerOSNeg{}{A',\non{\Del'}}}$
      with $A', A'' \in \bar\psi{(A)}$ and $A''\in\Delta'$, using Lemma~\ref{Fcomp}.

    \item If $A \in \TLemma $ then $A$ is  of the form $l_1 \vee \cdots \vee l_n$ with ${\Theory{\psi\non{(l_1)}, \ldots, \psi\non{(l_n)}}}$.  

      Let $\{\psi\non{(l'_1)},\ldots,\psi\non{(l'_m)}\}$ be a subset of $\{\psi\non{(l_1)}, \ldots, \psi\non{(l_n)}\}$ that is already semantically inconsistent and such that at most one literal is $\mathcal P$-negative, say possibly $\psi\non{(l'_m)}$.

      Let $C' \in \bar\psi{(A)}$. $C'$ is of the form ${\psi{(l_1)} \orF^? \cdots \orF^? \psi{(l_n)}}$.

      We build
      \[ 
      \infer{\DerOSNeg{}{\non {\Del'}, C'}}{
        \Infer{\DerOSNeg{}{\non {\Del'},\psi{(l_1)} \orN \cdots \orN \psi{(l_n)}}}{
          \Infer{\DerOSNeg{}{\non {\Del'},\psi{(l_1)}, \ldots, \psi{(l_n)}}}{
            \infer{\DerOSNeg{}{\psi{(l'_1)}, \ldots, \psi{(l'_m)}}}{
              \DerOSNeg[\mathcal P']{
                \psi\non{(l'_1)},\ldots,\psi\non{(l'_m)}
              }{}
            }  
          }
        }
      }
      \]
      where $\mathcal P'\eqdef
      \polar[{
          \polar[{
              \polar{
                \non{\psi(l'_1)}
              }
          }]{
            \ldots
          }
      }]{
        \non{\psi(l'_m)}
      }
      $.

      If $\psi\non{(l'_1)},\ldots,\psi\non{(l'_m)}$ is syntactically inconsistent, we close with $\Id[2]$.

      Otherwise 
      \[\polar[{
          \polar[{
              \polar{
                \non{\psi(l'_1)}
              }
          }]{
            \ldots
          }
      }]{
        \non{\psi(l'_{m-1})}
      }
      =
      \mathcal P,\non{\psi(l'_1)},\ldots,\non{\psi(l'_{m-1})}
      \]
      as none of the $\non{\psi(l'_i)}$, for $1\leq i\leq m-1$, is $\mathcal P$-negative.
      And for all $i$ such that $1\leq i\leq m-1$, the literal 
      $\non{\psi(l'_i)}$ is $\mathcal P'$-positive.

      Now if 
      $\non{\psi(l'_m)}$ is $\mathcal P'$-positive as well, we have 
      \[
      \atmCtxt[\mathcal P']{\psi\non{(l'_1)}, \ldots, \psi\non{(l'_m)}} = \psi\non{(l'_1)}, \ldots, \psi\non{(l'_m)}
      \]
      and we can close with $(\Init[2])$.

      If $\non{\psi{(l'_m)}}$ is not $\mathcal P'$-positive, we simply have
      \[
      \atmCtxt[\mathcal P']{\psi\non{(l'_1)}, \ldots, \psi\non{(l'_m)}} = \psi\non{(l'_1)}, \ldots, \psi\non{(l'_{m-1})}
      \]
      but we can still build
      \[
      \infer{
        \DerOSNeg[\mathcal P']{\psi\non{(l'_1)},\ldots,\psi\non{(l'_m)}}{}
      }{
        \infer[{[({\Init[1]})]}]{
          \DerOSPos[\mathcal P']{\psi\non{(l'_1)},\ldots,\psi\non{(l'_m)}}{\psi{(l'_m)}}
        }{
          \Theory{\psi\non{(l'_1)}, \ldots, \psi\non{(l'_m)}}}
      }
      \]
    \end{itemize}

  \item And Intro:
    \[\infer{\DerOSFl{\TLemma ,\Delta}{A_1 \andF A_2}} { \DerOSFl{\TLemma ,\Delta}{A_1} \quad \DerOSFl{\TLemma ,\Delta}{A_2}   } \]

    $A' \in \bar\psi(A_1 \andF A_2)$ is of the form $A'_1 \andF^? A_2'$ with $A'_1 \in \psi{(A_1)}$ and $A_2' \in \psi{(A_2)}$.

    Since $\DerOSNeg{}{A'_1 \andF^? A'_2, \non{\Del'}}$ is assumed to be safe, 
    $\DerOSNeg{}{A'_1,\non{\Del'}}$ and $\DerOSNeg{}{A'_2,\non{\Del'}}$ are also safe, and we can apply the induction hypothesis
    \begin{itemize}
    \item on $\DerOSFOL{\TLemma ,\Delta}{A_1 }$ to get $\DerNegLK{}{A'_1,\non{\Del'}}$
    \item and on $\DerOSFOL{\TLemma ,\Delta}{A_2}$ to get $\DerNegLK{}{A'_2,\non{\Del'}}$.
    \end{itemize}
    We build:
    \[\infer{\DerOSNeg{}{A'_1 \andF^? A'_2, \non{\Del'}}}
    {
      \infer{\DerOSNeg{}{A_1' \andN A_2', \non{\Del'}}} { \DerOSNeg{}{A_1', \non{\Del'} } \quad \DerOSNeg{}{A_2',\non{\Del'} }   }
    } 
    \]

  \item And Elim
      \[\infer{\DerOSFl{\TLemma ,\Delta}{A_i}} { \DerOSFl{\TLemma ,\Delta}{A_1 \andF A_{-1}}}\]
      with $i\in\{1,-1\}$.

      Since $\bar\psi(A_{-i}) \neq \emptyset$, let $A_{-i}' \in \bar\psi{(A_{-i})}$ and $C'= A_1' \andN A_{-1}'$ ($C' \in\bar\psi(A_1 \andF A_{-1})$). 

      Since $\DerOSNeg{}{A_i', \non{\Del'}}$ is assumed to be safe, 
      $\DerOSNeg{}{C',A_i',\non{\Del'}}$ is also safe, and we can apply the induction hypothesis on $\DerOSFl{\TLemma ,\Delta}{A_1 \andF A_{-1}}$ (with $\non{A_i'},\Del'$ and $C'$) to get $\DerNegLK{}{C',A_i',\non {\Del'}}$.

      We finally get:
      \[
      \infer[{{\contr[r]}}]{
        \DerOSNeg{}{A_i', \non{\Del'}}
      }{
        \iinfer[\mbox{Lemma~\ref{lem:invert}}]{
          \DerOSNeg{}{A_i',A_i',\non{\Del'}}
        }{
          \DerOSNeg{}{C',A_i',\non{\Del'}}
        }
      }
      \]

  \item Or Intro
    \[\infer{\DerOSFl{\TLemma ,\Delta}{A_1 \orF A_{-1}}} { \DerOSFl{\TLemma ,\Delta}{A_i}}\]
    $A' \in \bar\psi(A_1 \orF A_{-1})$ is of the form $A'_1 \orF^? A_{-1}'$ with $A'_1 \in \psi{(A_1)}$ and $A_{-1}' \in \psi{(A_{-1})}$.

    Since $\DerOSNeg{}{A'_1 \orF^? A'_{-1}, \non{\Del'}}$ is assumed to be safe, 
    $\DerOSNeg{}{A'_1,A'_{-1},\non{\Del'}}$ is also safe, and we can apply the induction hypothesis on $\DerOSFOL{\TLemma ,\Delta}{A_i}$  (with $\non{A'_{-i}},\Delta'$ and $A'_i$) to get $\DerNegLK{}{A'_1,A'_{-1},\non{\Del'}}$ and we build: 
    \[\infer{\DerOSNeg{}{A'_1 \orF^? A'_{-1}, \non{\Del'}}}
    {
      \infer{\DerOSNeg{}{A'_1 \orN A'_{-1}, \non{\Del'}}} 
      {
        \DerOSNeg{}{A'_1 , A'_{-1}, \non{\Del'}}
      }
    }
    \]

  \item Or Elim
    \[\infer{\DerOSFl{\TLemma ,\Delta}{C}} 
    { \DerOSFl{\TLemma ,\Delta}{A_1 \orF A_2} \quad \DerOSFl{\TLemma ,\Delta,A_1}{C} \quad \DerOSFl{\TLemma ,\Delta,A_2}{C}}
    \]    
    Let $D'= A'_1 \orN A'_2 $ with $A'_1 \in \psi{(A_1)}$ and $A'_2 \in \psi{(A_2)}$.
    
    Since $\DerOSNeg{}{C', \non{\Del'}}$ is assumed to be safe, 
    $\DerOSNeg{}{C',\non{A'_1},\non{\Del'}}$ and $\DerOSNeg{}{C',\non{A'_2},\non{\Del'}}$ and $\DerOSNeg{}{C',D',\non{\Del'}}$ are also safe, and we can apply the induction hypothesis
    \begin{itemize}
    \item on $\DerOSFOL{\TLemma ,\Delta,A_1}{C}$ to get $\DerNegLK{}{C',\non{A'_1},\non{\Del'}}$
    \item on $\DerOSFOL{\TLemma ,\Delta,A_2}{C}$ to get $\DerNegLK{}{C',\non{A'_2},\non{\Del'}}$.
    \item and on $\DerOSFOL{\TLemma ,\Delta}{A_1 \orF A_2}$ to get $\DerNegLK{}{C',D',\non{\Del'}}$.
    \end{itemize}
    We build:
    \[
    \infer[\cut_7]{\DerOSNeg{}{C',\non {\Del'}}}{
      \DerOSNeg{}{D',C',\non {\Del'}}
      \quad
      \iinfer{\DerOSNeg{}{\non{(A'_1 \orN A'_2)},C',\non {\Del'}}}{
        \infer{\DerOSNeg{}{\non{A'_1} \andP \non{A'_2},C',\non {\Del'}}}{
          \infer{\DerOSNeg{}{\non{A'_1} \andN \non{A'_2},C',\non {\Del'}}}{
            \DerOSNeg{}{\non{A'_1},C',\non {\Del'}}
            \quad 
            \DerOSNeg{}{\non{A'_2} ,C',\non {\Del'}}
          }
        }
      }
    }
    \]

  \item Universal quantifier Intro
    \[\infer[x \not\in \Gamma]{\DerOSFl{\TLemma ,\Delta}{\FA x A}} { \DerOSFl{\TLemma ,\Delta}{A}} \]
    $C'\in\bar \psi(\FA x A)$ is of the form $\FA x A'$ with $A' \in \bar\psi{(A)}$.

    Since $\DerOSNeg{}{C', \non{\Del'}}$ is assumed to be safe, 
    $\DerOSNeg{}{A',\non{\Del'}}$ is also safe, and we can apply the induction hypothesis on $\DerOSFOL{\TLemma,\Delta}{A}$ to get $\DerNegLK{}{A',\non{\Del'}}$ to get: 
    \[\infer{\DerOSNeg{}{\FA x A', \non{\Del'}}} { \DerOSNeg{}{A', \non{\Del'} }}\]

  \item Universal quantifier Elim
    \[\infer{\DerOSFl{\TLemma ,\Delta}{\subst A x t}} { \DerOSFl{\TLemma ,\Delta}{\FA x A}}\]
    $C'\in \bar\psi(\subst {A} x t)$ is of the form $\subst {A'} x t$ with $A'\in \psi{(A)}$ (by Remark~\ref{rm2}).

    Since $\DerOSNeg{}{C',\non{\Del'}}$ is assumed to be safe, 
    $\DerOSNeg{}{(\FA x {A'}),C',\non{\Del'}}$ is also safe, and we can apply the induction hypothesis on $\DerOSFOL{\TLemma,\Delta}{\FA x A}$ (with $\non{C'},\Delta'$ and $(\FA x {A'})$) to get $\DerNegLK{}{(\FA x {A'}),C',\non{\Del'}}$. 

    We build 
    \[
    \infer[{{\contr[r]}}]{\DerOSNeg{}{\subst {A'} x t, \non{\Del'}}}{
      \infer[\mbox{Lemma~\ref{Ladm}}]{\DerOSNeg{}{\subst {A'} x t, \subst {A'} x t,\non{\Del'}}}{
        \iinfer[\mbox{Lemma~\ref{lem:invert}}]{\DerOSNeg{}{A',\subst {A'} x t, \non{\Del'}}}{
          \DerOSNeg{}{(\FA x A'),\subst {A'} x t, \non{\Del'}}
        }
      }
    }
    \]

  \item Existential quantifier Intro
    \[\infer{\DerOSFl{\TLemma ,\Delta}{\EX x A}} { \DerOSFl{\TLemma ,\Delta}{\subst A x t}}\]
    $C' \in \bar\psi{(\EX x A)}$ is of the form $\EX  x A'$ with $A' \in \bar\psi{(A)}$. 

    Let $A'_t = {\subst  {A'} x t}$ ($A'_t \in \bar\psi{(\subst A x t)}$ by Remark~\ref{rm2}).

    Since $\DerOSNeg{}{C',\non{\Del'}}$ is assumed to be safe, 
    $\DerOSNeg{}{A'_t,\non{\Del'}}$ is also safe, and we can apply the induction hypothesis on $\DerOSFOL{\TLemma,\Delta}{\subst A x t}$ to get $\DerNegLK{}{A'_t,\non{\Del'}}$. 

    By Lemma~\ref{Lconn2} it suffices to prove $\DerNegLK{}{\EX x A', \non{A'_t}} $ in order to get $\DerNegLK{}{C',\non{\Del'}}$:
    \[\infer[\mbox{Lemma~\ref{LKex}(2)}]{\DerOSNeg{}{\EX x A', \non{A'_t}} }
    { \infer{\DerOSEx {\FA x \non{A'} }{\bullet}{\non {A'_t} }  }
      {  \infer{\DerOSEx{} {\EX x A'} {\non {A'_t} }  }
        {\DerOSEx {} {A'_t} {\non {A'_t}} } 
      }
    }
    \]
    We can complete the proof by applying Lemma~\ref{con1}.

  \item Existential quantifier Elim
    \[\infer[x\not\in \Gamma,B]{\DerOSFl{\TLemma ,\Delta}{B}} { \DerOSFl{\TLemma ,\Delta}{\EX x A} \quad \DerOSFl{\Gamma,\Delta,A}{B}}\]
    Let $C' = \EX x A'$ with $A' \in \bar\psi{(A)}$.

    Since $\DerOSNeg{}{B',\non{\Del'}}$ is assumed to be safe, 
    $\DerOSNeg{}{B',C',\non{\Del'}}$ and $\DerOSNeg{}{B',\non{A'},\non{\Del'}}$ are also safe, and we can apply the induction hypothesis
    \begin{itemize}
    \item on $\DerOSFOL{\TLemma ,\Delta}{\EX x A}$ to get $\DerNegLK{}{B',C',\non{\Del'}}$;
    \item on $\DerOSFOL{\Gamma,\Delta,A}{B}$ to get $\DerNegLK{}{B',\non{A'},\non{\Del'}}$. 
    \end{itemize}
    We build
    \[
    \infer[\cut_7]{\DerOSNeg{}{B', \non {\Del'}}}
    { {\DerOSNeg{}{C',B', \non {\Del'}}}
      \quad
      { \iinfer{\DerOSNeg{}{\non{C'},B', \non {\Del'}}}
        {\infer{\DerOSNeg{}{\FA x {(\non{A'})},B', \non {\Del'}}}
          {\DerOSNeg{}{ \non{A'},B', \non {\Del'}}}
        }
      }
    }
    \]

  \item Negation Intro
    \[\infer{\DerOSFl{\TLemma ,\Delta}{\neg A}} { \DerOSFl{\TLemma ,\Delta,A}{B \andF \neg B}}\]
    
    If $C' \in \bar\psi{(\neg A)}$ then $\non {C'} \in \bar\psi{(A)}$. 
    Let $D' = D'_1 \andN  D'_2$ with $D'_1\in \bar\psi{(B)}$ and $D'_2\in \bar\psi{(\neg B)}$.
    Therefore $\non{D'_2}\in \bar\psi{(B)}$, $D' \in \bar\psi{(B \andF \neg B)}$ and $\Del', \non {C'} \in \bar\psi {(\Del, A)}$.

    Since $\DerOSNeg{}{\non {\Del'}, C'}$ is assumed to be safe, 
    $\DerOSNeg{}{\non {\Del'}, C', D'}$ is also safe, and we can apply the induction hypothesis on 
    $\DerOSFOL{\TLemma ,\Delta,A}{B \andF \neg B}$ to get $\DerNegLK{}{\non {\Del'}, C', D'}$.
    We build
    \[
    \infer[\cut_7]{\DerOSNeg{}{\non {\Del'}, C'}}{
      \DerOSNeg{}{\non {\Del'}, C',D'}
      \quad
      \infer[\mbox{Corollary~\ref{conF}(4)}]{\DerOSNeg{}{\non {\Del'}, C',\non {D'}}}{
        \infer{\DerOSNeg{}{\non {\Del'}, C',\non{D'_1} \orN \non{D'_2} }}{
          \infer[\mbox{Lemma~\ref{Fcomp}}]{
            \DerOSNeg{}{\non {\Del'}, C',\non{D'_1}, \non{D'_2} }
          }{}
        }
      }
    }
    \]

  \item Negation Elimination
    \[\infer{\DerOSFl{\TLemma ,\Delta}{A}} { \DerOSFl{\TLemma ,\Delta}{\neg\neg A}}\]

    $A' \in  \bar\psi{(A)}$ is such that $A'\in\bar\psi{(\neg\neg A)}$.

    The induction hypothesis on $\DerOSFl{\TLemma,\Delta}{\neg\neg A}$ gives $\DerOSNeg{}{\non{\Del'},A'}$
    and we are done.
    
  \end{itemize}
\end{proof}

\section{The system used for simulation of \DPLLTh}
\label{sec:DPLLTh}

The motivation for the \LKThp\ system was to perform proof-search
modulo theories, and in particular simulate \DPLLTh\ techniques. 
Therefore, we conclude this report with the actual system that we use in other works~\cite{farooqueTR12,farooque13} to perform the simulation:

It is the \LKThp\ system, extended with the admissible and invertible rules $(\Pol)$ and $(\cut_7)$ (or more precisely restricted versions of them), as shown in Fig~\ref{fig:4DPLL}.

\begin{figure}[!h]
  \[
  \begin{array}{|c|}
    \upline
    \textsf{Synchronous rules}
    \hfill\strut\\[3pt]
    \infer[{[(\andP)]}]{\DerPos{\Gamma}{A\andP B}{}{\mathcal{P}}}
    {\DerPos{\Gamma}{A}
      {}{\mathcal P} \qquad \DerPos{\Gamma}{B}{}{\mathcal{P}}}
    \qquad
    \infer[{[(\orP)]}]{\DerPos{\Gamma}{A_1\orP A_2}{}{\mathcal{P}}}
    {\DerPos{\Gamma}{A_i}{}{\mathcal{P}}}
       \qquad 
    \infer[{[(\EX)]}]{\DerPos{\Gamma}{\EX x A}{}{\mathcal{P}}}
    {\DerPos{\Gamma}{\subst A x t}{}{\mathcal{P}}}   
    \\[15pt]
      \infer[{[{(\top^+)}]}]{\DerPos{\Gamma}{\top^+} {} {\mathcal{P}}} {\strut}
       \qquad 
\infer[{[({\Init[1]})]l \mbox{ is $\mathcal P$-positive}}]
	{\DerPos{\Gamma  }{l}{} {\mathcal{P}}} {\atmCtxtP{\mathcal P}\Gam,\non l\models_{\mathcal T}}    \qquad
    \infer[{[(\Release)]N \mbox{ is not $\mathcal P$-positive}}]{\DerPos {\Gam} {N} {} {\mathcal{P}}}
    {\DerNeg {\Gam} {N} {} {\mathcal{P}}}\\
    \midline
    \textsf{Asynchronous rules}
    \hfill\strut\\[3pt]
    \infer[{[(\andN)]}]{\DerNeg{\Gamma}{A\andN B,\Delta} {} {\mathcal{P}}}
    {\DerNeg{\Gamma}{A,\Delta} {} {\mathcal{P}} 
      \qquad \DerNeg{\Gamma}{B,\Delta} {} {\mathcal{P}}}
    \qquad
    \infer[{[(\orN)]}]{\DerNeg {\Gamma} {A_1\orN A_2,\Delta} {} {\mathcal{P}}}
    {\DerNeg {\Gamma} {A_1,A_2,\Delta} {} {\mathcal{P}}}
    
   \qquad
       \infer[{[(\FA)] x\notin\FV{\Gam,\Delta,\mathcal P} }]
       {\DerNeg{\Gamma}{(\FA x A),\Delta}{} {\mathcal{P}}}
    {\DerNeg {\Gamma} {A,\Delta}{} {\mathcal{P}}}
\\\\
    \infer[{[{(\bot^-)}]}]{\DerNeg{\Gamma} {\Del,\bot^-} {} {\mathcal{P}}}
    {\DerNeg{\Gamma} {\Del} {} {\mathcal{P}}}
     \qquad
    \infer[{[{(\top^-)}]}]{\DerNeg{\Gamma} {\Del,\top^-} {} {\mathcal{P}}}{\strut}
     \qquad
    \infer[{[({\Store})]%
      \begin{array}l%
        A\mbox{ is a literal}\\
        \mbox{or is $\mathcal P$-positive}\\
      \end{array}}]{\DerNeg \Gam {A,\Del} {} {\mathcal{P}}} 
    {\DerNeg {\Gam,\non A} {\Del} {} {\polar{\non A}}}  
    \\\midline
      \textsf{Structural rules}
    \hfill\strut\\[3pt]
    \infer[{[(\Select)]\mbox{$P$ is not $\mathcal P$-negative}}]
    {\DerNeg {\Gam,\non P} {}{} {\mathcal{P}}} 
    {\DerPos {\Gam,\non P} {P} {} {\mathcal{P}}}
    \qquad
    \infer[{[({\Init[2]})]}]{\DerNeg {\Gam} {}{} {\mathcal{P}}}{\atmCtxtP{\mathcal P}\Gam\models_{\mathcal T}}
    \\\midline
      \textsf{Admissible/Invertible rules}
    \hfill\strut\\[3pt]
      \infer[{[(\Pol)]\atmCtxtP{\mathcal {P}}{\Gamma},\non l\models_{\mathcal T}}]
            {\DerNeg{\Gamma}{}{}{\mathcal P}}
            {\DerNeg{\Gamma}{}{}{\mathcal {P},l}}
            \qquad
    \infer[\cut_7] {\DerNeg {\Gamma}{}{}{\mathcal P}}
    {\DerNeg{\Gamma}{l}{}{\mathcal P} \quad 
    \DerNeg{\Gamma}{\non l}{}{\mathcal P}}
    \downline
  \end{array}
  \]
  \begin{tabular}{lll}%
    where& $\polar A \eqdef \mathcal P,A$& if $A\in\UP$\\
    &$\polar A \eqdef \mathcal P$& if not
  \end{tabular}%

  \caption{System for the simulation of \DPLLTh}
  \label{fig:4DPLL}
\end{figure}

\clearpage

\bibliographystyle{Common/good}
\bibliography{Common/abbrev-short,Common/Main,Common/crossrefs}

\end{document}